\documentclass[11pt]{article}
\usepackage[margin=1in]{geometry}
\usepackage[noadjust]{cite}

\usepackage{xcolor}
\usepackage[colorlinks=true,citecolor=blue,linkcolor=blue]{hyperref}

\hypersetup{
    colorlinks,
    linkcolor={red!50!black},
    citecolor={blue!50!black},
    urlcolor={blue!80!black}
}

\usepackage{graphicx}
\usepackage{pdfpages}
\usepackage[]{amsmath,amssymb,amsfonts,latexsym,amsthm,enumerate,fullpage,xcolor,cite,bbm}
\usepackage[nameinlink, noabbrev, capitalize]{cleveref}
\usepackage{nicefrac}

\usepackage[utf8]{inputenc}
\usepackage{amsmath,amssymb}
\usepackage{graphics,graphicx,tikz,color,float}
\usepackage{mathtools}
\usepackage{amsfonts}
\usepackage{tikz}
\usepackage{comment}
\usepackage{framed}
\usepackage[section]{placeins}
\usepackage{algorithm}
\usepackage[noend]{algpseudocode}
\usepackage{xcolor}

\newtheorem{theorem}{Theorem}

\newtheorem{claim}{Claim}
\newtheorem{lemma}{Lemma}

\newtheorem{corollary}{Corollary}

\newtheorem{definition}{Definition}
\newtheorem{remark}{Remark}

\newtheorem{fact}{Fact}

\newcommand{\beq}{\begin{equation}}
	\newcommand{\enq}{\end{equation}}
\newcommand{\bel}{\begin{lemma}}
	\newcommand{\enl}{\end{lemma}}
\newcommand{\bet}{\begin{theorem}}
	\newcommand{\ent}{\end{theorem}}

\newcommand{\tr}{\mathrm{Tr}}

\newcommand{\E}{\mathbb{E}}
\newcommand{\ketbra}[1]{|#1\rangle\langle#1|}

\newcommand{\eps}{\varepsilon}

\newcommand{\generateA}{A= \Ext_1(Y,Z_s) }
\newcommand{\generateT}{C= \Ext_2(Z,A)}
\newcommand{\generateB}{B= \Ext_{1}(Y,C)}
\newcommand{\generateAA}{A^\prime= \Ext_1(Y^\prime,Z^\prime_s) }
\newcommand{\generateTT}{C^\prime= \Ext_2(Z^\prime,A^\prime)}
\newcommand{\generateBB}{B^\prime= \Ext_{1}(Y^\prime,C^\prime)}
\newcommand{\generateAbar}{\overline{A}= \Ext_1(Y,\overline{Z_s}) }
\newcommand{\generateTbar}{\overline{C}= \Ext_2(\overline{Z},\overline{A})}
\newcommand{\generateBbar}{\overline{B}= \Ext_{1}(Y,\overline{C})}
\newcommand{\generateAAbar}{\overline{A}^\prime= \Ext_1(Y^\prime,\overline{Z_s}^\prime) }
\newcommand{\generateTTbar}{\overline{C}^\prime= \Ext_2(\overline{Z}^\prime,\overline{A}^\prime)}
\newcommand{\generateBBbar}{\overline{B}^\prime= \Ext_{1}(Y^\prime,\overline{C}^\prime)}

\newcommand{\generateYSbar}{$\overline{Z}_s =$ Prefix$(\overline{Z},s)$}

\newcommand{\generateYSSbar}{$\overline{Z}_s^\prime =$ Prefix$(\overline{Z}^\prime,s)$}
\newcommand{\sendArl}{$A \longleftarrow A$}

\newcommand{\sendBrl}{$B \longleftarrow B$}

\newcommand{\sendTlr}{$C \longrightarrow C$}
\newcommand{\sendAArl}{$A^\prime \longleftarrow A^\prime$}

\newcommand{\sendBBrl}{$B^\prime \longleftarrow B^\prime$}

\newcommand{\sendTTlr}{$C^\prime \longrightarrow C^\prime$}
\newcommand{\sendAbarrl}{$\overline{A} \longleftarrow \overline{A}$}

\newcommand{\sendBbarrl}{$\overline{B} \longleftarrow \overline{B}$}

\newcommand{\sendTbarlr}{$\overline{C} \longrightarrow \overline{C}$}
\newcommand{\sendAAbarrl}{$\overline{A}^\prime \longleftarrow \overline{A}^\prime$}

\newcommand{\sendBBbarrl}{$\overline{B}^\prime \longleftarrow \overline{B}^\prime$}

\newcommand{\sendTTbarlr}{$\overline{C}^\prime \longrightarrow \overline{C}^\prime$}
\newcommand{\sendYSlr}{$Z_s \longrightarrow Z_s$}

\newcommand{\sendYSbarlr}{$\overline{Z}_s \longrightarrow \overline{Z}_s$}

\newcommand{\sendYSSbarlr}{$\overline{Z}_s^\prime \longrightarrow \overline{Z}_s^\prime$}

\newcommand{\sendYYbarlr}{$\overline{Z}^\prime \longrightarrow \overline{Z}^\prime$}

\newcommand{\poly}{\textnormal{poly}}

\newcommand*{\cSC}{\mathcal{SC}}
\newcommand*{\cA}{\mathcal{A}}

\newcommand*{\cH}{\mathcal{H}}

\newcommand*{\cD}{\mathcal{D}}

\newcommand*{\cO}{\mathcal{O}}

\newcommand*{\cX}{\mathcal{X}}

\newcommand*{\cZ}{\mathcal{Z}}
\newcommand*{\cE}{\mathcal{E}}

\newcommand{\cP}{\mathcal{P}}

\newcommand*{\IP}{\mathsf{IP}}
\newcommand{\Ext}{\mathsf{Ext}}

\newcommand{\pre}{\mathsf{Prefix}}

\newcommand{\advcb}{\mathsf{AdvCB}}
\newcommand{\ff}{\mathsf{FF}}
\newcommand{\ecc}{\mathsf{ECC}}

\newcommand{\supp}{\mathrm{supp}}
\newcommand{\suppress}[1]{}

\newcommand{\defeq}{\ensuremath{ \stackrel{\mathrm{def}}{=} }}
\newcommand{\F}{\mathbb{F}}

\newcommand {\br} [1] {\ensuremath{ \left( #1 \right) }}

\newcommand {\minusspace} {\: \! \!}

\newcommand {\fn} [2] {\ensuremath{ #1 \minusspace \br{ #2 } }}

\newcommand {\dmax} [2] {\fn{\mathrm{D}_{\max}}{#1 \middle\| #2}}

\newcommand {\hminone} [1] {\fn{ \mathrm{H }_{\min}}{#1}}

\newcommand {\hmin} [2] {\fn{ \mathrm{H }_{\min}}{#1 \middle | #2}}

\newcommand {\id} {\ensuremath{\mathbb{I}}}

\newcommand {\Hmin}{\mathrm{H}_{\min}}

\newcommand{\nmre}{\mathsf{nmExt}}
\newcommand{\nmreenc}{\mathsf{nmreEnc}}
\newcommand{\nmredec}{\mathsf{nmreDec}}

\usepackage{afterpage}

\newcommand*{\sm}{\mathsf{same}}
\newcommand*{\qpas}{\mathsf{qpa\mhyphen state}}

\newcommand*{\nmas}{\mathsf{qnm\mhyphen state}}

\newcommand*{\nma}{\mathsf{qnm\mhyphen adv}}

\newcommand*{\nmext}{\mathsf{nmExt}}

\newcommand{\mac}{\mathsf{MAC}}

\newcommand{\X}{\mathcal{X}}

\newcommand*{\cL}{\mathcal{L}}

\newcommand{\bra}[1]{\langle #1|}
\newcommand{\ket}[1]{|#1 \rangle}

\mathchardef\mhyphen="2D

\newcommand*{\enc}{\mathrm{Enc}}
\newcommand*{\dec}{\mathrm{Dec}}

\newenvironment{changemargin}[2]{%
\begin{list}{}{%
\setlength{\topsep}{0pt}%
\setlength{\leftmargin}{#1}%
\setlength{\rightmargin}{#2}%
\setlength{\listparindent}{\parindent}%
\setlength{\itemindent}{\parindent}%
\setlength{\parsep}{\parskip}%
}%
\item[]}{\end{list}}

\makeatletter
\newcommand*{\rom}[1]{\expandafter\@slowromancap\romannumeral #1@}
\makeatother

\mathchardef\mhyphen="2D

\newfloat{Protocol}{tbhp}{lop}

\title{Quantum secure non-malleable randomness encoder and its applications}

\author{
 Rishabh Batra\footnote{Centre for Quantum Technologies, National University of Singapore, \texttt{e0894305@u.nus.edu}.} \and
 Naresh Goud Boddu\footnote{NTT Research, \texttt{naresh.boddu@ntt-research.com}.} \and
 Rahul Jain\footnote{Centre for Quantum Technologies and Department of Computer Science, 
  National University of Singapore and MajuLab, UMI 3654, Singapore,  \texttt{rahul@comp.nus.edu.sg}.} 
}
\date{}

\begin{document}
\maketitle
\begin{abstract}
“Non-Malleable Randomness Encoder” (NMRE) was introduced by  Kanukurthi, Obbattu, and Sekar~\cite{KOS18} as a useful cryptographic primitive helpful in the construction of non-malleable codes. To the best of our knowledge, their construction is not known to be quantum secure.  

We provide a construction of a first rate-$1/2$, $2$-split, quantum secure NMRE and use this in a black-box manner, to construct for the first time the following:
\begin{enumerate}
 \item rate $1/11$, $3$-split, quantum non-malleable code,
    \item rate $1/3$, $3$-split,  quantum secure non-malleable code,
    \item rate $1/5$, $2$-split, average case quantum secure non-malleable code.
    \end{enumerate}
\end{abstract}

\section{Introduction}
\label{sec:intro}

In a seminal work, Dziembowski, Pietrzak and Wichs~\cite{DPW10} introduced non-malleable codes to provide a meaningful guarantee for the encoded message in situations where traditional error-correction or even error-detection is impossible. Informally, non-malleable codes encode a classical message in a manner such that tampering the codeword results in decoder either outputting the original message or a message that is unrelated/independent of original message.

It is impossible to construct non-malleable codes which protect against arbitrary tampering functions. This is because an adversary could simply apply the decoder to ciphertext recovering
the original message, and output the encoding of the related message as the tampering of ciphertext. There have been several works studying non-malleable codes with some restrictions on the adversary's tampering functions. Perhaps the most well known of these tampering function families is the so called split model introduced by Liu and Lysyanskaya~\cite{LL12}, who constructed efficient constant rate non-malleable codes against computationally bounded adversaries under strong cryptographic assumptions.

In the $t$-split model, the cipher text is split into $t$-parts, say $(c_1, c_2 , \ldots, c_t)$ and the adversary is allowed to tamper 
$(c_1, c_2 , \ldots, c_t) \to (c'_1, c'_2 , \ldots, c'_t)= (f_1(c_1), f_2(c_2) , \ldots, f_t(c_t))$ (see \cref{fig:splitstate121} for a diagram of this model for $t=3$). An important parameter of interest for NMCs is its rate = $m/n$ where $m$ is the message length and $n$ is the codeword length. In the $t$-split model, \cite{CG14a} provided the upper bound on the rate being $1-1/t$. There has been a flurry of works in the past decade \cite{LL12,DKO13,CG14a,CG14b,ADL17,CGL15,li15,Li17,Li19,AO20,AKOOS22} trying to achieve constant rate non-malleable codes for various $t$-split models. For example, see \cref{table:nmcs3} for the constant rate non-malleable codes known in the literature. 
\begin{table}
\centering
\begin{tabular}{||c c c ||} 
 \hline
 \textbf{Work by} & \textbf{Rate}  & \textbf{Splits}\\ [1ex] 
 \hline\hline 
  \cite{CZ14} & $\Omega(1)$ & 10 \\[0.5ex] 
 \hline
 \cite{KOS17} & $1/3$ & 4\\[0.5ex] 
 \hline
 \cite{KOS18} & $1/3$ & 3\\[0.5ex] 
 \hline
 \cite{AO20} & $\Omega(1)$ & 2 \\[0.5ex] 
 \hline
 \cite{AKOOS22} & $1/3$ & 2\\[0.5ex] 
 \hline
\end{tabular}

\vspace{0.5cm}

\caption{ Best known explicit constructions of constant rate split non-malleable codes.}
\label{table:nmcs3}
\end{table}

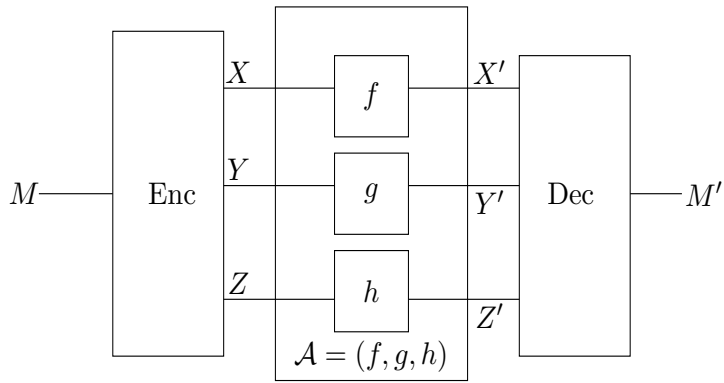
\begin{figure}
		\centering
  \resizebox{10cm}{5cm}{
	\begin{tikzpicture}
		
		\node at (1.8,1.5) {$M$};
		\node at (11,1.5) {$M'$};
		
		\draw (2,1.5) -- (3,1.5);
		\draw (3,-0.5) rectangle (4.5,3.5);
		\node at (3.8,1.5) {$\enc$};
		
		
		
		\node at (4.7,3) {$X$};
		\node at (8.1,3) {$X'$};
		\draw (4.5,2.8) -- (6,2.8);
		\draw (7,2.8) -- (8.5,2.8);
		\draw (10,1.5) -- (10.7,1.5);

            \node at (4.7,1.8) {$Y$};
		\node at (8.1,1.4) {$Y'$};
          \draw (4.5,1.6) -- (6,1.6);
		\draw (7,1.6) -- (8.5,1.6);
         \draw (6,1) rectangle (7,2);
         \node at (6.5,1.5) {$g$};
  
		\node at (4.7,0.4) {$Z$};
		\node at (8.1,0) {$Z'$};
		\draw (4.5,0.2) -- (6,0.2);
		\draw (7,0.2) -- (8.5,0.2);
		
		\draw (6,2.2) rectangle (7,3.2);
		\node at (6.5,2.7) {$f$};
		\draw (6,-0.2) rectangle (7,0.8);
		\node at (6.5,0.3) {$h$};
		
		\node at (6.5,-0.5) {$\mathcal{A}=(f,g,h)$};
		\draw (5.2,-0.8) rectangle (7.8,3.8);

		
		
		\draw (8.5,-0.5) rectangle (10,3.2);
		\node at (9.2,1.5) {$\dec$};
		
	\end{tikzpicture} }
	\caption{$3$-split model.}\label{fig:splitstate121}
\end{figure}

\begin{figure}
		\centering
   \resizebox{12cm}{6cm}{
		\begin{tikzpicture}
			

     \draw (2,4.3) -- (11,4.3);
     \node at (1.8,4.3) {$\hat{M}$};
      \node at (11.2,4.3) {$\hat{M}$};
			\node at (1.8,1.5) {$M$};
			\node at (11,1.5) {$M'$};

\draw (1.8,2.8) ellipse (0.3cm and 1.8cm);
   
			
			\draw (2,1.5) -- (3,1.5);
			\draw (3,-0.5) rectangle (4.5,3.5);
			\node at (3.8,1.5) {$\enc$};
			
			
			
			\node at (4.7,3) {$X$};
			\node at (8.25,3) {$X'$};
  
			\draw (4.5,2.8) -- (6,2.8);
			\draw (7,2.8) -- (8.6,2.8);
			\draw (10,1.5) -- (10.7,1.5);
			
			\node at (4.7,0.4) {$Z$};
			\node at (8.25,0) {$Z'$};
			\draw (4.5,0.2) -- (6,0.2);
			\draw (7,0.2) -- (8.6,0.2);

                \node at (4.7,1.8) {$Y$};
		      \node at (8.25,1.8) {$Y'$};
         \draw  (4.5,1.6) -- (5,1.6);
            \draw [dashed] (5,1.6) -- (6,1.6);
		  \draw (7,1.6) -- (8.6,1.6);
                \draw (6,1.1) rectangle (7,1.9);
                \node at (6.5,1.5) {$V$};
            \node at (7.5,1.3) {$E'_2$};
             \node at (5.5,1.4) {$E_2$};
             \draw (5.8,1.4) -- (6,1.4);
             \draw (7,1.4) -- (7.2,1.4);
   
			\draw (6,2.1) rectangle (7,2.9);
			\node at (6.5,2.5) {$U$};
			\draw (6,0.1) rectangle (7,0.9);
			\node at (6.5,0.5) {$W$};
			
			\node at (6.5,-0.4) {$\mathcal{A}=(U,V,W,\ket{\psi}_{})$};
			\draw (5,-0.8) rectangle (8,3.8);

			\draw (5.5,1.5) ellipse (0.3cm and 1.1cm);
			\node at (5.5,2.2) {$E_1$};
				\draw (5.7,2.2) -- (6,2.2);
			\node at (7.5,2.2) {$E'_1$};
				\draw (7,2.2) -- (7.2,2.2);
				\node at (5.5,0.8) {$E_3$};
				\draw (5.7,0.7) -- (6,0.7);
			\node at (7.5,0.6) {$E'_3$};
			\draw (7.0,0.7) -- (7.2,0.7);
			
			
			\draw (8.6,-0.5) rectangle (10,3.2);
			\node at (9.2,1.5) {$\dec$};
			
		\end{tikzpicture} }
		\caption{$3$-split model with shared entanglement. The shared entanglement $\psi$ is stored in registers $E_1, E_2$ and $E_3$.}\label{fig:splitstate211d}
	\end{figure}
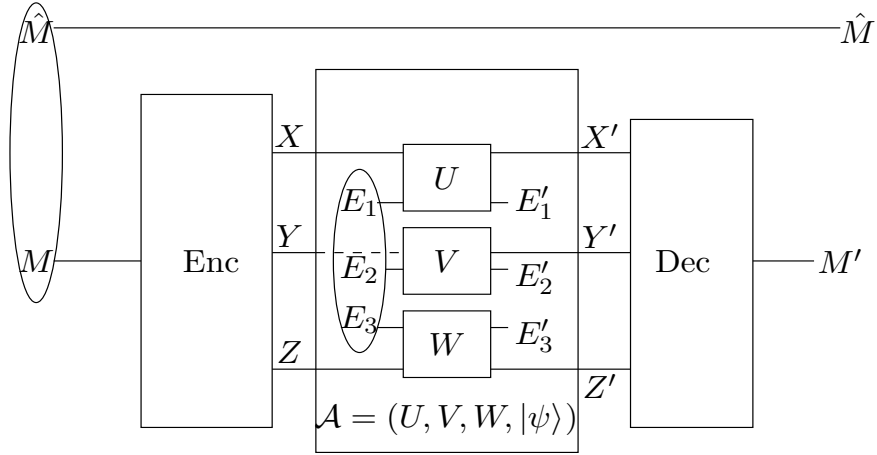
 In most of the works studied in the literature, adversaries with quantum capabilities and \emph{shared entanglement} had not been considered (see \cref{fig:splitstate211d}). It is a priori not clear whether previous coding schemes remain secure in the presence of quantum adversaries. The only two works that studied non-malleable codes against quantum adversaries are the works of Aggarwal, Boddu, and Jain~\cite{ABJ22} and Boddu, Goyal, Jain and Ribeiro~\cite{BGJR23}. \cite{ABJ22} studied non-malleable codes for classical messages secure against quantum adversaries, where the encoding and decoding procedures were classical. These are referred to as \textbf{quantum secure non-malleable codes.} \cite{BGJR23} extended the notion of non-malleability to quantum messages in the presence of a quantum adversary. These are referred to as 
 \textbf{quantum non-malleable codes.} Both the works achieve their results for exponentially small error, however with inverse polynomial (worst case) rate. That is for a message length $m$ bits/qubits, the codeword length is $\poly(m)$ bits/qubits.

The shortcomings of existing split non-malleable coding schemes in the face of adversaries with quantum capabilities raise the following natural question:
\begin{quote}
\begin{center}
Can we design constant rate efficient $t$-split quantum secure and quantum non-malleable codes?
\end{center}
\end{quote}
We resolve this question in the affirmative for $t = 3$. We also construct constant rate, $2$-split, quantum secure non-malleable code for uniformly distributed message.

\subsection*{Non-malleable randomness encoder}

In \cite{KOS18}, the notion of “Non-malleable Randomness Encoders” (NMREs) was introduced which can be thought of as a relaxation of non-malleable codes in
the following sense: NMREs output a random message along with its corresponding non-malleable encoding. This has 
applications where we only need to encode randomness
 and security is not required to hold for arbitrary, adversarially chosen
messages. For example, in applications of non-malleable codes to tamper-resilient security, the messages that are encoded are typically randomly generated secret keys.

Efficient construction of NMREs was given in \cite{KOS18} where they built a $2$-split, rate-$1/2$ NMRE. They also showed that NMREs can be used, in a black-box manner, to construct $3$-split non-malleable codes. To the best of our knowledge, NMRE constructed by \cite{KOS18} is not known to be quantum secure. Thus, it leads to us to the following question:

\begin{quote}
\begin{center}
Can we design rate $1/2$ efficient $2$-split NMRE secure against quantum adversaries with shared entanglement?
\end{center}
\end{quote}
We resolve this question in the affirmative.

\subsection*{Our results}
We prove the following results (stated informally). The encoding and decoding procedures for all the codes run in time $\poly(n)$. All the codes are $\eps$-secure for $\eps = 2^{-n^{\Omega(1)}}$.
 \begin{theorem}\label{thm:mainqnmc}
    There exists a rate $1/2$, $2$-split quantum secure non-malleable randomness encoder.
\end{theorem}

\begin{theorem}\label{thm:mainqnmc1}
    There exists a rate $1/11$, $3$-split quantum non-malleable code.
\end{theorem}

\begin{theorem}\label{thm:mainqnmc2}
    There exists a rate $1/3$, $3$-split quantum secure non-malleable code.
\end{theorem}

\begin{theorem}\label{thm:mainqnmc3}
    There exists a rate $1/5$, $2$-split quantum secure non-malleable code for uniform classical messages.
\end{theorem}

As a side result, we show something stronger:\begin{itemize}
    \item  The explicit code from \cref{thm:mainqnmc1} is a $3$-out-of-$3$ quantum non-malleable secret sharing scheme with share size at most $n$, any message of length $\approx \frac{n}{11}$.
    \item  The explicit code from \cref{thm:mainqnmc2} is a $3$-out-of-$3$ quantum secure non-malleable secret sharing scheme with share size at most $n$, any message of length $\approx \frac{n}{3}$.
    \item The explicit code from \cref{thm:mainqnmc3} is a $2$-out-of-$2$ average case quantum secure non-malleable secret sharing scheme with share size at most $n$, any message of length $\approx \frac{n}{5}$.
\end{itemize}

All the above secret sharing schemes encoding and decoding procedures run in time $\poly(n)$. All the secret sharing schemes are secure for error $\eps = 2^{-n^{\Omega(1)}}$. We refer the reader to \cref{Appendixss} for more details on quantum non-malleable secret sharing and quantum secure non-malleable secret sharing.

\begin{remark} We note to the reader that as observed in~\cite{GK18}, not every $3$-split non-malleable code (for classical messages) is a $3$-out-of-$3$ non-malleable secret sharing scheme (for classical messages). In particular, as observed in~\cite{BS19}, the  construction of $3$-split non-malleable code of~\cite{KOS18} is not a $3$-out-of-$3$ non-malleable secret sharing scheme.
    
\end{remark}
\begin{remark}
$2$-split quantum secure non-malleable code of 
\cite{ABJ22} is also a $2$-out-of-$2$ quantum secure non-malleable secret sharing scheme. Similarly, $2$-split quantum non-malleable code of \cite{BGJR23} is also a $2$-out-of-$2$ quantum  non-malleable secret sharing scheme.  

    Both the works \cite{ABJ22,BGJR23} achieve their results for exponentially small error, however with inverse polynomial (worst case) rate. On the other hand, our non-malleable secret sharing schemes (both quantum and quantum secure) are of constant rate. 

\end{remark}

\begin{table}
\centering
\begin{tabular}{||c c c c c c||} 
 \hline
 \textbf{Work by} & \textbf{Rate} & \textbf{Messages} & \textbf{Adversary} & \textbf{Shared key}  & \textbf{Splits}\\ [1ex] 
 \hline\hline 
 \cite{KOS18} & $1/2$ & classical & classical & No & {\color{blue} 2} \\[0.5ex] 
 \hline
  This work & $1/2$ & \textbf{classical} & \textbf{quantum}  & No & {\color{blue} 2}\\ [1ex] 
 \hline
\end{tabular}

\vspace{0.5cm}

\caption{Comparison between the best known non-malleable randomness encoders.}
\label{table:nmres}
\end{table}

\section{Our techniques}
\subsection*{Our NMRE construction}
We first provide a short overview of the construction from~\cite{KOS18} before stating our construction. 
The goal of NMRE is to take uniform randomness as input, say $R$ and produce as output $3$ registers, say $MXY$ such that $M$ is uniform and $(X,Y)$ correspond to the $2$-split non-malleable encoding of $M$. What this means is that: when adversary tampers $(X,Y) \to (X',Y')$ independently, we will have $M$ is either same as $M'$ or $M$ is unrelated to $M'$.

Any non-malleable code is, by default, a secure NMRE as we can simply generate randomness $R$ at random and let the codeword be the output of NMC. The main challenge is in building a rate-optimal, state-optimal NMRE. The construction in \cite{KOS18} uses information-theoretic one-time message authentication codes (MACs), randomness extractors and (poor rate) $2$-split non-malleable code as building blocks.

Randomness extractors are functions that take as input a source $W$ and an independent and uniform seed $S$ and extract $O=\Ext(W,S)$. $O$ is close to uniform randomness even when the source is not perfectly uniform but has some min-entropy guarantees (weak source). $\mac$s are functions that take as input message $m$ and uniform key $K$, such that adversary who has access to $(m, \mac(m,K))$ cannot produce $(m',\mac(m',K))$ for any $m\ne m'$ except with tiny probability.


We are now ready to state the construction of \cite{KOS18}. We provide their (slightly modified) construction which is a $3$-split NMRE which is as follows:
\begin{itemize}
    \item Consider input $R = W \vert \vert S \vert \vert K$.
    \item Output $M = \Ext(W,S)$.
    \item Output $(X,Y) =  \mathsf{NMC}(S,K)$ and $Z= W \vert \vert \mac(W,K)$.
\end{itemize}

\begin{figure}[h]
\centering

\resizebox{12cm}{6cm}{

\begin{tikzpicture}



\node at (0,2.8) {$W$};
\draw (0.7,2.8) -- (0.7,3);
\draw (2,3) -- (0.7,3);
\draw (2,2.9) rectangle (3.5,3.5);
\node at (2.7,3.2) {$\mac$};
\draw (3.5,3.2) -- (6.3,3.2);

\draw (0.2,2.8) -- (4.5,2.8);
\node at (0,0.2) {$S$};
\draw (3.5,0.2) -- (4.5,0.2);
\draw (4.5,0.2) -- (4.5,2.4);
\draw (6.3,2.4) -- (4.5,2.4);

\node at (0,-0.2) {$K$};
\draw (1.4,2.8) -- (1.4,1.8);
\draw (1.4,1.8) -- (2,1.8);
\draw (0.2,-0.2) -- (2,-0.2);
\draw (1.4,1.2) -- (2,1.2);
\draw (1.4,0.2) -- (1.4,1.2);
\draw (1,-0.2) -- (1,3.2);
\draw (2,3.2) -- (1,3.2);
\draw (0.2,0.2) -- (2,0.2);
\draw (3.5,1.5) -- (4.2,1.5);
\node at (3.7,1.7) {$M$};

\node at (0,3.7) {$M$};
\draw  (4.2,1.5) -- (4.2,3.7);
\draw  (0.2,3.7) -- (4.2,3.7);

\draw (15.2,1.5) -- (13.5,1.5);
\draw (13,1.5) -- (13.5,1.5);
\node at (13.4,1.7) {$M'$};



\draw (2,1) rectangle (3.5,2);
\node at (2.7,1.5) {$\Ext$};
\draw (2,-0.4) rectangle (3.5,0.6);
\node at (2.7,0) {$\mathsf{NMC}$};

\node at (7.8,3) {$W'$};

\node at (7.8,2.6) {$Y'$};
\draw (4,2.8) -- (6.3,2.8);
\draw (7.3,2.8) -- (8.5,2.8);
\draw (7.3,3.2) -- (8.5,3.2);
\draw (11,2.8) -- (9.7,2.8);
\draw (7.3,2.4) -- (7.7,2.4);
\draw (7.7,2.4) -- (7.7,0.5);
\draw (8.5,0.5) -- (7.7,0.5);

\node at (3.7,-0.02) {$X$};
\node at (3.7,0.4) {$Y$};
\node at (7.8,0.0) {$X'$};
\draw (3.5,-0.2) -- (5,-0.2);

\draw (5,0.2) -- (5,-0.2);
\draw (5,0.2) -- (6.3,0.2);

\draw (7.3,0.2) -- (8.5,0.2);
\draw (11,0.2) -- (9.6,0.2);
\node at (9.9,0.4) {$K'$};
\node at (9.9,0) {$S'$};
\draw (10,0.6) -- (9.6,0.6);
\draw (10,0.6) -- (10,2);
\draw (9,2) -- (10,2);
\draw (9,2) -- (9,2.5);

\draw (6.3,2) rectangle (7.3,3.5);
\node at (6.8,2.8) {$f$};
\draw (6.3,0) rectangle (7.3,1);
\node at (6.8,0.5) {$g$};
\node at (9.1,0.3) {$\mathsf{NMD}$};
\node at (9.1,3) {$\mathsf{Verify}$};
\node at (11.3,3.5) {If check passes, };
\node at (11,3) { $W$ else $\perp$ };
\node at (7,-0.4) {$\mathcal{A}=(f,g)$};
\draw (5.2,-0.8) rectangle (8.1,4);

\draw (8.5,2.5) rectangle (9.7,3.5);



\draw (0.5,-2) rectangle (4.8,5);
\draw (8.35,-2) rectangle (14.1,5);
\draw (11.4,-0.5) rectangle (13,1.9);
\draw (8.5,-0.5) rectangle (9.6,1);
\node at (11.5,-1.6) {$\dec$};
\node at (2.2,-1.6) {$\enc$};
\draw (11,1.8) -- (11,2.8);
\draw (11,1.8) -- (11.4,1.8);
\draw (11,1.2) -- (11.4,1.2);
\draw (11,0.2) -- (11,1.2);
\node at (12.2,1.5) {$\Ext$};
\node at (15.4,1.5) {$M'$};

\end{tikzpicture} }
\caption{Rate $1/2, ~2$-split non-malleable randomness encoder (slightly modified)~\cite{KOS18}.}\label{fig:kos}
\end{figure}
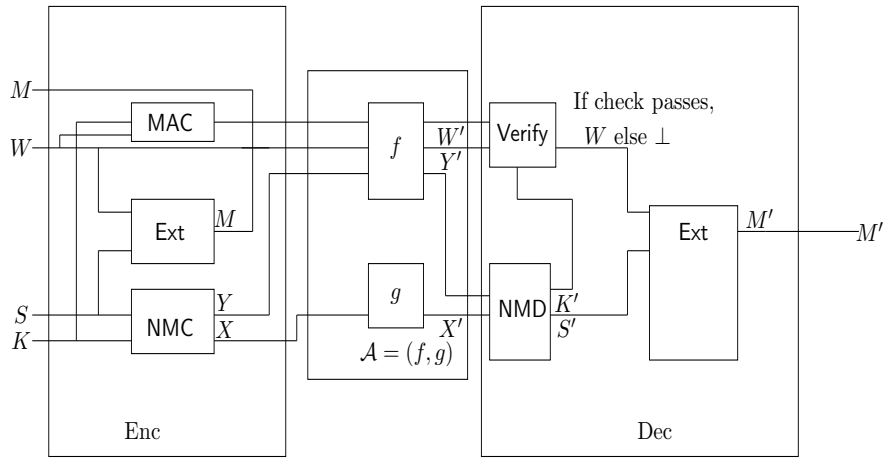

The randomness $M$ is generated as the output of an extractor.
The problem now boils down to protecting both $W,S$ non-malleably. Since size of $\vert W \vert > \vert M \vert$, we cannot protect $W$ using poor rate $\mathsf{NMC}$. This is where $\mac$ comes into the picture. Using short key $K$, one can hope to protect $W$ by considering $Z=W \vert \vert \mac(W,K)$ in one of the splits. The other two splits $(X,Y)$ are used to non-malleably protect both short seed $S$ and the key $K$ (which is used to protect $W$). In \cite{KOS18}, the authors further require $\vert W \vert > 2 \vert M\vert$ for their technical arguments to go through. This overall gives a rate $1/2$, $3$-split NMRE. They also prove that on replacing $\mathsf{NMC}$ with a stronger augmented $\mathsf{NMC}$, we can consider $(X, (YZ))$ as two splits, and the above construction still remains secure which gives a $2$-split NMRE. 

Even though \cite{KOS18} achieve a rate optimal and state optimal NMRE, their construction is not known to be quantum secure. One may hope to prove that their construction is quantum secure, given we know quantum secure extractors~\cite{DPVR09}, quantum secure NMC for classical messages~\cite{BJK21}. Unfortunately, the NMC used by~\cite{KOS18} is the construction of~\cite{Li17} which is not known to be quantum secure. Furthermore, it is a priori not clear if their construction remains quantum secure if we replace the NMC used by~\cite{KOS18} with the quantum secure NMC of~\cite{BJK21}.

We  provide a construction of rate $1/2$, $2$-split NMRE which is arguably simpler than the construction in~\cite{KOS18} and is also quantum secure (see \cref{table:nmres}). Our  NMRE is essentially obtained from a new quantum secure $2$-source non-malleable extractor $2\nmext$ that we construct. Our new  $2\nmext$ construction follows on similar lines to that of~\cite{BJK21}. In~\cite{BJK21}, both the sources were $n$ bits long. In our case, one source is $n$ bit long and the other is $\delta n$ bits long,  which requires a careful change and adjustment of parameters.  Consider the following simple approach (see Figure~\ref{fig:splitstate2adv}): 
\begin{itemize}
    \item First sample uniform registers $X$ and $Y$ as two parts of codewords.  
    \item Encoding procedure computes $M= 2\nmext(X,Y)$ and treats it as uniform random message for which $(X,Y)$ is the non-malleable encoding. 
\end{itemize}

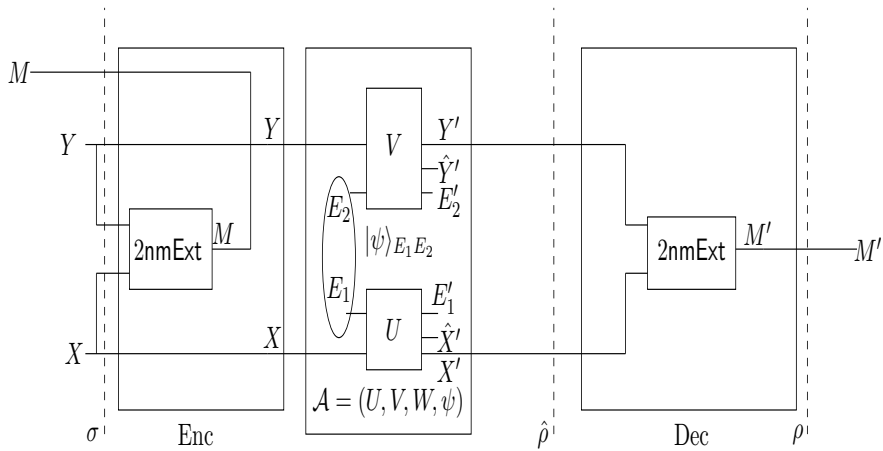
\begin{figure}[h]
\centering

\resizebox{12cm}{6cm}{

\begin{tikzpicture}



\node at (0.9,2.8) {$Y$};

\draw (1.2,2.8) -- (4.5,2.8);
\node at (1,0.2) {$X$};
\draw (1.4,2.8) -- (1.4,1.8);
\draw (1.4,1.8) -- (2,1.8);
\draw (1.4,1.2) -- (2,1.2);
\draw (1.4,0.2) -- (1.4,1.2);
\draw (1.2,0.2) -- (4.5,0.2);
\draw (3.5,1.5) -- (4.2,1.5);
\node at (3.7,1.7) {$M$};

\node at (0,3.7) {$M$};
\draw  (4.2,1.5) -- (4.2,3.7);
\draw  (0.2,3.7) -- (4.2,3.7);

\draw (15.2,1.5) -- (13.5,1.5);
\draw (13,1.5) -- (13.5,1.5);
\node at (13.4,1.7) {$M'$};


\draw (2,1) rectangle (3.5,2);
\node at (2.7,1.5) {$2\nmext$};

\draw [dashed] (1.55,-0.8) -- (1.55,4.5);
\draw [dashed] (9.7,-0.8) -- (9.7,4.5);
\draw [dashed] (14.3,-0.8) -- (14.3,4.5);

\node at (6.9,1.6) {$\ket{\psi}_{E_1E_2}$};
\node at (1.34,-0.8) {$\sigma$};
\node at (9.5,-0.8) {$\hat{\rho}$};
\node at (14.13,-0.8) {$\rho$};

\node at (4.6,3) {$Y$};
\node at (7.8,3) {$Y'$};

\draw (4.5,2.8) -- (6.3,2.8);
\draw (7.3,2.8) -- (11,2.8);

\node at (4.6,0.4) {$X$};
\node at (7.8,0.0) {$X'$};

\draw (4.5,0.2) -- (6.3,0.2);
\draw (7.3,0.2) -- (11,0.2);

\draw (6.3,2) rectangle (7.3,3.5);
\node at (6.8,2.8) {$V$};
\draw (6.3,0) rectangle (7.3,1);
\node at (6.8,0.5) {$U$};

\node at (6.7,-0.4) {$\mathcal{A}=(U,V,W,\psi)$};
\draw (5.2,-0.8) rectangle (8.2,4);

\draw (5.8,1.4) ellipse (0.3cm and 1cm);
\node at (5.8,2) {$E_2$};
\draw (6,2.2) -- (6.3,2.2);

\node at (7.8,2.1) {$E_2'$};
\draw (7.3,2.2) -- (7.5,2.2);
\node at (5.8,1) {$E_1$};
\draw (5.93,0.7) -- (6.3,0.7);
\node at (7.7,0.9) {$E_1'$};
\draw (7.3,0.7) -- (7.6,0.7);


\draw (1.8,-0.5) rectangle (4.8,4);
\draw (10.2,-0.5) rectangle (14.1,4);
\draw (11.4,1) rectangle (13,1.9);
\node at (12.2,-0.8) {$\dec$};
\node at (3.2,-0.8) {$\enc$};
\draw (11,1.8) -- (11,2.8);
\draw (11,1.8) -- (11.4,1.8);
\draw (11,1.2) -- (11.4,1.2);
\draw (11,0.2) -- (11,1.2);
\node at (12.2,1.5) {$2\nmext$};
\node at (15.4,1.5) {$M'$};
\node at (7.8,0.4) {$\hat{X}'$};
\draw (7.3,0.4) -- (7.6,0.4);
\node at (7.8,2.5) {$\hat{Y}'$};
\draw (7.3,2.5) -- (7.6,2.5);
\end{tikzpicture} }
\caption{Rate $1/2, ~2$-split quantum secure non-malleable randomness encoder.}\label{fig:splitstate2adv}
\end{figure}
Here, we need a function $2\nmext$ such that, when adversary tampers and  $(X,Y) \ne (X', Y')$, we get $M=2\nmext(X,Y)$ to be independent (thus unrelated) of $M'=2\nmext(X',Y')$. When adversary does not tamper and $(X,Y) = (X', Y')$, then we have $M=M'$. This is exactly the functionality provided by a $2$-source quantum secure non-malleable extractor. Hence, using  arguments similar to that of~\cite{ABJ22}, we can construct a quantum secure NMRE using a quantum secure $2\nmext$. 

Let $n$ be a positive integer and $\delta>0$ be any tiny constant. We construct $2 \nmext : \{0,1 \}^n \times \{ 0,1\}^{\delta n} \rightarrow \{ 0,1\}^{(0.5-\delta)n}$ which is quantum secure (details of our construction in \cref{sec:2nm}). Thus the rate of our quantum secure NMRE is $( \frac{0.5-\delta}{1+\delta}) \approx (0.5-\delta')$ for tiny $\delta'>0.$

 \begin{remark}
   \em{}We note that the $2$-source non-malleable extractor construction of {\cite{BJK21}}, viewed as an quantum secure NMRE provides a rate $\frac{1}{8}$ quantum secure NMRE. But we optimize the parameters in their construction to obtain rate $(0.5 -\delta')$ quantum secure NMRE for any tiny constant $\delta'>0$. 
 \end{remark}

 \subsection*{From quantum secure NMRE to $3$-split quantum and quantum secure non-malleable code}


We now provide the definition of non-malleable codes for quantum messages. We state the definition for $3$-split model but it can be generalized to any $t$-split model.

Let $\sigma_M$ be an arbitrary state in a message register $M$ and $\sigma_{M\hat{M}}$ be its canonical purification.
Let ($3$-split) coding scheme be given by an encoding Completely Positive Trace-Preserving (CPTP) map 
$\enc : \cL( \cH_M) \to \cL(\cH_{X}\otimes \cH_Y\otimes\cH_{Z})$ and a decoding CPTP map $\dec  : \cL(\cH_{X}\otimes\cH_{Y} \otimes\cH_{Z}) \to  \cL(\cH_{M})$, where $\cL(\cH)$ is the space of all linear operators in the Hilbert space $\cH$.
The most basic property we require of this coding scheme $(\enc,\dec)$ is correctness (which includes preserving entanglement with external systems), i.e.,
\begin{equation*}
\dec(\enc(\sigma_{M\hat{M}}))=\sigma_{M\hat{M}},
\end{equation*}
where we use the shorthand $T$ to represent the CPTP map $T \otimes \id$ whenever the action of the identity operator $\id$ is clear from context.

Before we proceed to define $3$-split non-malleability, we describe the $3$-split adversarial tampering model in the quantum setting.
Let $\rho_{XYZ}=\enc(\sigma_M)$ be the $3$-split encoding of message $\sigma_M$.
A $3$-split tampering adversary $\cA$ is specified by three tampering maps.  $U:\cL(\cH_X\otimes \cH_{E_1})\to\cL(\cH_{X'}\otimes\cH_{E'_1})$, $V:\cL(\cH_Y\otimes \cH_{E_2})\to\cL(\cH_{Y'}\otimes \cH_{E'_2})$ and $W:\cL(\cH_Z\otimes \cH_{E_3})\to\cL(\cH_{Z'}\otimes \cH_{E'_3})$ along with a quantum state $\ket\psi_{E_1 E_2E_3}$ which captures the shared entanglement between the non-communicating tampering adversaries. Finally, the decoding procedure $\dec$ is applied to the tampered codeword $\tau_{X' Y'Z'}$ and stored in register $M'$.
\cref{fig:splitstate211d} presents a diagram of this tampering model.
Let 
\begin{equation*}
    \eta = \dec( ( U \otimes V \otimes W)  \left(\enc( \sigma_{M \hat{M}}) \otimes \ketbra{\psi} \right) ( U^\dagger \otimes V^\dagger \otimes W^\dagger) )
\end{equation*}
be the final state after applying the $3$-split tampering adversary $\cA$ followed by the decoding procedure. The $3$-split non-malleability of the coding scheme $(\enc,\dec)$ is defined as follows.
\begin{definition}[$3$-split non-malleable code for  quantum messages\label{def:qnmcodes}~\cite{BGJR23}]\label{def:qnmcodesfinaldef}  
We say that the coding scheme
$(\enc, \dec)$ is an \emph{$\eps$-$3$-split non-malleable code for quantum messages} if for every $3$-split adversary $\cA=(U,V,W,\ket\psi_{E_1 E_2E_3})$ and for every quantum message $\sigma_M$ (with canonical purification $\sigma_{M\hat{M}}$)
it holds that
\begin{equation*}
    \eta_{M'\hat{M}} \approx_\eps p_{\cA} \sigma_{M \hat{M}}+(1-p_{\cA})\gamma^{\cA}_{M'}\otimes \sigma_{\hat{M}},
\end{equation*}
where $p_{\cA}\in[0,1]$ and $\gamma^{\cA}_{M'}$ depend only on the $3$-split adversary $\cA$, and $\approx_\eps$ denotes that the two states are $\eps$-close in trace distance. 
\end{definition}

Once we have a rate $1/2$ quantum secure NMRE, we use ideas similar to \cite{KOS18} to construct $3$-split non-malleable codes. Say, we have a NMRE that uses randomness $R=(X,Y)$ where $(X,Y)$ is $2$-split non-malleable encoding of $R$. We can then use the randomness $R= 2\nmext(X,Y)$ to both encrypt and authenticate a quantum (or classical) message in the third split using $R$. We first provide a $3$-split argument of~\cite{KOS18} for classical message $m$. Consider $R=R_e \vert \vert R_a$. 
\begin{itemize}
    \item Sample uniform $X \otimes Y$, and consider $(X,Y)$ correspond to $2$-splits of codeword
    \item Compute $R_e \vert \vert R_a =R=2\nmext(X,Y)$
    \item Output the third split as $Z=R_e \oplus m \vert \vert \mac(R_a,R_e \oplus m )$.
\end{itemize}
Using the $\mac$ property, we can protect the cipher text $ R_e \oplus m $, which in turn protects $m$, given the encryption property of one-time pad. Carefully analyzing the parameters of $\mac$, OTP, rate $1/2$ quantum secure NMRE, the above stated construction provides rate $1/3$, $3$-split quantum secure non malleable codes for classical messages. 

This approach does not work for quantum messages. This is because the notion of quantum message authentication and quantum encryption are not two seperate concepts unlike for classical messages. We have a way to protect quantum messages in the third split using quantum authentication schemes~\cite{BW16} which use Clifford group unitaries. Unfortunately, using Clifford group to authenticate quantum messages is expensive as it needs $\vert R\vert = \cO(\vert M\vert ^2)$, where $\vert M\vert$ is the number of qubits of a quantum message. Thus, we cannot construct a constant rate $3$-split quantum non-malleable code using this.

A recent work of~\cite{CLLW16} showed that there exists a subgroup of Clifford group ($\cSC(\cH_A)$) which forms a unitary $2$-design (a property needed to provide non-malleability for third split message~\cite{AM17,BGJR23}). Furthermore, the implementation of $\cSC(\cH_A)$ requires $\vert R \vert =5 \vert M \vert$ randomness. Using the twirling property of $\cSC(\cH_A)$ to authenticate (or non-malleably protect the quantum message), we obtain a rate $1/11$, $3$-split quantum non malleable code.

\subsection*{From quantum secure NMRE to quantum secure $2$-split non-malleable code}
One natural way to define a $2$-split quantum secure NMC would be to consider $(X,(Y \vert \vert Z))$ as two parts of the code (where $X,Y,Z$ are three parts of our $3$-split code). One may hope to use the augmented property of NMC to extend the arguments to $2$-split case. But in the quantum setting as observed in~\cite{BGJR23}, the security of the underlying NMC is compromised. To argue quantum security of $2$-split NMCs for uniform classical messages, we use arguments similar to that of~\cite{BGJR23}, combined with a classical analogue of unitary $2$-designs, i.e. pair-wise independent permutation family.  This leads us to argue rate $1/5$, $2$-split non malleable codes for uniform classical messages.

\cref{table:nmcs} summarizes the main properties of known constructions of split non-malleable codes and related constructions of keyed quantum schemes.
\begin{table}
\centering
\begin{tabular}{||c c c c c c||} 
 \hline
 \textbf{Work by} & \textbf{Rate} & \textbf{Messages} & \textbf{Adversary} & \textbf{Shared key}  & \textbf{Splits}\\ [1ex] 
 \hline\hline 
\cite{CGL15} & $\Omega\left(\frac{1}{\mathsf{poly}(n)}\right)$  & classical  & classical  & No & 2\\ [1ex] 
 \hline
 \cite{Li19} & $ \Omega\left(\frac{ \log \log (n)}{\log(n)}\right)$ & classical & classical & No & 2\\ [1ex] 
 \hline
 \cite{AO20} & $\Omega(1)$ & classical & classical & No & 2 \\[0.5ex] 
 \hline
 \cite{AKOOS22} & $1/3$ & classical & classical & No & 2\\[0.5ex] 
 \hline
 \cite{KOS18} & $1/3$ & classical & classical & No & {\color{blue} 3} \\[0.5ex] 
 \hline
 \cite{ABJ22} & $\Omega\left(\frac{1}{\mathsf{poly}(n)}\right)$ & classical & \textbf{quantum} & No & 2 \\[0.5ex] 
 \hline
\cite{AM17} & $\Omega(1)$ & \textbf{quantum} & \textbf{quantum} & {\color{red} Yes} & 2\\[0.5ex] 
 \hline
 \cite{BGJR23} & $\Omega\left(\frac{1}{\mathsf{poly}(n)}\right)$ & \textbf{quantum} & \textbf{quantum}  & No & 2\\ [1ex] 
 \hline
  This work & $1/11$ & \textbf{quantum} & \textbf{quantum}  & No & {\color{blue} 3}\\ [1ex] 
 \hline
  This work & $1/3$ & \textbf{classical} & \textbf{quantum}  & No & {\color{blue} 3}\\ [1ex] 
 \hline
  This work & $1/5$ & \textbf{(average) classical} & \textbf{quantum}  & No & {\color{blue} 2}\\ [1ex] 
 \hline
\end{tabular}

\vspace{0.5cm}

\caption{Comparison between the best known explicit constructions of split non-malleable codes. Here, $n$ denotes the codeword length.}
\label{table:nmcs}
\end{table}

\subsection*{Organization}

In \cref{sec:prelims}, we discuss some quantum information theoretic and other preliminaries. \cref{sec:usefulstuff} contains useful lemmas and facts. In \cref{sec:2nmre}, we provide and define a quantum secure rate $1/2$ non-malleable randomness encoder. We describe and analyze our construction of rate $1/11$, $3$-split quantum non-malleable code in \cref{sec:qnmc67}. We describe and analyze our construction of rate $1/3$, $3$-split quantum secure non-malleable code in \cref{sec:3qnmcclassical}. We describe and analyze our construction of rate $1/5$, $2$-split quantum secure non-malleable code for uniform classical messages in \cref{sec:2qnmcclassical}. \cref{Appendixss} contains details on our $3$-out-of-$3$ quantum non-malleable secret sharing scheme.

\subsection*{Acknowledgments}

 The work of R.B. is supported by the Ministry
of Education, Singapore, under the Research Centres of Excellence program. The work of N.G.B. was done partially while he was a research assistant at the Centre for Quantum Technologies. The work of R.J. is supported by the NRF grant NRF2021-QEP2-02-P05 and the Ministry of Education, Singapore, under the Research Centres of Excellence program. This work was done in part while R.J. was visiting the Technion-Israel Institute of Technology, Haifa, Israel and the Simons Institute for the Theory of Computing, Berkeley, USA.

\section{Preliminaries}
\label{sec:prelims}

The notation and major part of this section is Verbatim taken from~\cite{BGJR23}.

Let $n,m,d$  represent positive integers and $k, k_1, k_2,  \eps \geq 0$ represent reals.

\subsection{Basic notation}
All the logarithms are evaluated to the base $2$. 
We denote sets by uppercase calligraphic letters such as $\X$ and use uppercase roman letters such as $X$ and $Y$ for both random variables and quantum registers. The distinction will be clear from context.
We denote the uniform distribution over $\{0,1\}^d$ by $U_d$.
For a {\em random variable} $X \in \X$, we use $X$ to denote both the random variable and its distribution, whenever it is clear from context. 
We use $x \leftarrow X$ to denote that $x$ is drawn according to $X$, and, for a finite set $\X$, we use $x \leftarrow \X$ to denote that $x$ is drawn uniformly at random from $\X$. For two random variables $X,Y$ we use $X \otimes Y$ to denote their product distribution. 
We call random variables $X, Y$, {\em copies} of each other if and only if $\Pr[X=Y]=1$.  For a random variable $X \in \{0,1 \}^n$ and $d\leq n$, let $\pre(X, d)$ represent the first $d$ bits of $X$. 

\subsection{Quantum information theory}

In this section we cover some important basic prerequisites from quantum information theory alongside some useful lemmas and facts.

\subsubsection{Conventions and notation}
Consider a finite-dimensional Hilbert space $\cH$ endowed with an inner-product $\langle \cdot, \cdot \rangle$ (we only consider finite-dimensional Hilbert-spaces). A quantum state (or a density matrix or a state) is a positive semi-definite operator on $\cH$ with trace value  equal to $1$. 
It is called {\em pure} if and only if its rank is $1$.  
 Let $\ket{\psi}$ be a unit vector on $\cH$, that is $\langle \psi,\psi \rangle=1$.  
With some abuse of notation, we use $\psi$ to represent the state and also the density matrix $\ketbra{\psi}$, associated with $\ket{\psi}$. Given a quantum state $\rho$ on $\cH$, the {\em support of $\rho$}, denoted by $\text{supp}(\rho)$, is the subspace of $\cH$ spanned by all eigenvectors of $\rho$ with non-zero eigenvalues.
 
A {\em quantum register} $A$ is associated with some Hilbert space $\cH_A$. Define $\vert A \vert := \log\left(\dim(\cH_A)\right)$. Let $\mathcal{L}(\cH_A)$ represent the set of all linear operators on the Hilbert space $\cH_A$. For operator $O \in \cL(\cH_A)$, we denote $O^T$ to represent the transpose of operator $O$. For operators $O, O'\in \cL(\cH_A)$, the notation $O \leq O'$ represents the L\"{o}wner order, that is, $O'-O$ is a positive semi-definite operator. We denote by $\mathcal{D}(\cH_A)$, the set of all quantum states on the Hilbert space $\cH_A$. State $\rho$ with subscript $A$ indicates $\rho_A \in \mathcal{D}(\cH_A)$. If two registers $A,B$ are associated with the same Hilbert space, we shall represent the relation by $A\equiv B$. For two states $\rho, \sigma$, we let $\rho \equiv \sigma$ represent that they are identical as states (potentially in different registers). Composition of two registers $A$ and $B$, denoted $AB$, is associated with the Hilbert space $\cH_A \otimes \cH_B$.  For two quantum states $\rho\in \mathcal{D}(\cH_A)$ and $\sigma\in \mathcal{D}(\cH_B)$, $\rho\otimes\sigma \in \mathcal{D}(\cH_{AB})$ represents the tensor product ({\em Kronecker} product) of $\rho$ and $\sigma$. The identity operator on $\cH_A$ is denoted $\id_A$. Let $U_A$ denote maximally mixed state in $\cH_A$. Let $\rho_{AB} \in \mathcal{D}(\cH_{AB})$. Define
$$ \rho_{B} \defeq \tr_{A}{\rho_{AB}} \defeq \sum_i (\bra{i} \otimes \id_{B})
\rho_{AB} (\ket{i} \otimes \id_{B}) , $$
where $\{\ket{i}\}_i$ is an orthonormal basis for the Hilbert space $\cH_A$.
The state $\rho_B\in \mathcal{D}(\cH_B)$ is referred to as the marginal state of $\rho_{AB}$ on the register $B$. Unless otherwise stated, a missing register from subscript in a state represents partial trace over that register. Given $\rho_A\in\mathcal{D}(\cH_A)$, a {\em purification} of $\rho_A$ is a pure state $\rho_{AB}\in \mathcal{D}(\cH_{AB})$ such that $\tr_{B}{\rho_{AB}}=\rho_A$. Purification of a quantum state is not unique.
Suppose $A\equiv B$. Given $\{\ket{i}_A\}$ and $\{\ket{i}_B\}$ as orthonormal bases over $\cH_A$ and $\cH_B$ respectively, the \textit{canonical purification} of a quantum state $\rho_A$ is a pure state $\rho_{AB} \defeq (\rho_A^{\frac{1}{2}}\otimes\id_B)\left(\sum_i\ket{i}_A\ket{i}_B\right)$. 

A quantum {map} $\cE: \mathcal{L}(\cH_A)\rightarrow \mathcal{L}(\cH_B)$ is a completely positive and trace preserving (CPTP) linear map. CPTP map $\cE$ is described by Kraus operators $\{ M_i : \cH_A\rightarrow \cH_B \}_i$ such that $\cE(\rho) = \sum_i M_i \rho M^\dagger_i$ and $\sum_i M^\dagger_i M_i =\id_A. $  A {\em Hermitian} operator $H:\cH_A \rightarrow \cH_A$ is such that $H=H^{\dagger}$. A projector $\Pi \in  \mathcal{L}(\cH_A)$ is a Hermitian operator such that $\Pi^2=\Pi$. A {\em unitary} operator $V_A:\cH_A \rightarrow \cH_A$ is such that $V_A^{\dagger}V_A = V_A V_A^{\dagger} = \id_A$. The set of all unitary operators on $\cH_A$ is  denoted by $\mathcal{U}(\cH_A)$. An {\em isometry}  $V:\cH_A \rightarrow \cH_B$ is such that $V^{\dagger}V = \id_A$ and $VV^{\dagger} = \id_B$. A {\em POVM} element is an operator $0 \le M \le \id$.~We use the shorthand $\bar{M} \defeq \id - M$, where $\id$ is clear from the context. 
We use shorthand $M$ to represent $M \otimes \id$, where $\id$ is clear from the context. 

\begin{definition}[Pauli operators]\label{def:pauli}
The single-qubit Pauli operators are given by
\[ I = \begin{pmatrix} 1 & 0 \\ 0 & 1\end{pmatrix} \quad X = \begin{pmatrix} 0 & 1 \\ 1 & 0\end{pmatrix} \quad Y = \begin{pmatrix} 0 & -i \\ i & 0 \end{pmatrix} \quad Z = \begin{pmatrix} 1 & 0 \\ 0 & -1 \end{pmatrix}.\]

 An $n$-qubit Pauli operator is given by the $n$-fold tensor product of single-qubit Pauli operators. We denote the set of all $\vert A \vert$-qubit Pauli operators on $\cH_A$ by  $\cP(\cH_A)$, where $\vert \cP(\cH_A)\vert =4^{\vert A \vert}$. Any linear operator $L \in \cL(\cH_A)$ can be written as a linear combination of $\vert A \vert$-qubit Pauli operators with complex coefficients as $L = \sum_{P \in \mathcal{P}(\cH_A)} \alpha_P P$. This is called the Pauli decomposition of a linear operator.
\end{definition}
\begin{definition}[Pauli group]
The single-qubit Pauli group is given by
\begin{equation*}
    \{ +P, -P, \ iP, \ -iP : P \in \{ I, X, Y, Z\} \}.
\end{equation*}
The Pauli group on $\vert A \vert$-qubits is the group generated by the operators described above applied to each of $\vert A \vert$-qubits in the tensor product. We denote the $\vert A \vert$-qubit Pauli group on $\cH_A$ by  $\tilde{\cP}(\cH_A)$.
    
\end{definition}
\begin{definition}[Clifford group]\label{def:clifford}
The Clifford group $\mathcal{C}(\cH_A)$ is defined as the group of unitaries that normalize the Pauli group $\tilde{\cP}(\cH_A)$: 
$$\mathcal{C}(\cH_A) = \{ V \in \mathcal{U}(\cH_A) : V \tilde{\cP}(\cH_A) V^\dagger =\tilde{\cP}(\cH_A)\}.$$ The Clifford unitaries are elements in the Clifford group.

\end{definition}

\begin{definition}[\cite{CLLW16}~Sub-group of a Clifford group]\label{lem:subclifford}
Let $P ,Q \in \cP(\cH_A)$ be non-identity Pauli operators. There exists a sub-group of Clifford group $\mathcal{SC}(\cH_A)$ such that the following holds: 
\[  \vert \{ C \in \cSC(\cH_A) \vert C^\dagger P C =Q \} \vert = \frac{\vert \cSC(\cH_A)  \vert}{\vert\cP(\cH_A) \vert -1} \quad ; \quad \vert \cSC(\cH_A)  \vert = 2^{5 \vert A \vert }-2^{3 \vert A \vert}.\]
    Informally, applying a random Clifford operator from  $\mathcal{SC}(\cH_A)$ (by conjugation) maps $P$ to a Pauli operator chosen uniformly over all non-identity Pauli operators. 
Furthermore $\cP(\cH_A)\subset \cSC(\cH_A)$.

\end{definition}

\begin{definition}[Classical register in a pure state]\label{def:classicalinpurestate}Let $\X$ be a set. A {\em classical-quantum} (c-q) state $\rho_{XE}$ is of the form \[ \rho_{XE} =  \sum_{x \in \X}  p(x)\ket{x}\bra{x} \otimes \rho^x_E , \] where ${\rho^x_E}$ are states. 

Let $\rho_{XEA}$ be a pure state. We call $X$ a classical register in $\rho_{XEA}$, if $\rho_{XE}$ (or $\rho_{XA}$) is a c-q state. We identify random variable $X$ with the register $X$, with $\Pr(X=x) =p(x)$.

\end{definition}

\begin{definition}[Copy of a classical  register]\label{def:copyofaclassicalregister}
Let $\rho_{X\hat{X}E}$ be a pure state with $X$ being a classical register in $\rho_{X\hat{X}E}$ taking values in $\cX$. Similarly, let $\hat{X}$ be a classical register in $\rho_{X\hat{X}E}$ taking values in $\cX$. Let $\Pi_{\mathsf{Eq}} = \sum_{x \in \cX} \ketbra{x} \otimes \ketbra{x}$ be the \emph{equality} projector acting on the registers $X\hat{X}$. We call $X$ and $\hat{X}$ copies of each other (in the computational basis) if $\tr\left(\Pi_{\mathsf{Eq}} \rho_{X\hat{X}}\right) =1$.
\end{definition}

\begin{definition}[Conditioning] \label{def:conditioning}
Let  
\[ \rho_{XE} =  \sum_{x \in \{0,1\}^n}  p(x)\ket{x}\bra{x} \otimes \rho^x_E , \]
be a c-q state. For an event $\mathcal{S} \subseteq \{0,1\}^n$, define  $$\Pr(\mathcal{S})_\rho \defeq  \sum_{x \in \mathcal{S}} p(x) \quad ; \quad (\rho|X\in \mathcal{S})\defeq \frac{1}{\Pr(\mathcal{S})_\rho} \sum_{x \in \mathcal{S}} p(x)\ket{x}\bra{x} \otimes \rho^x_E.$$
We sometimes shorthand $(\rho|X\in \mathcal{S})$ as $(\rho|\mathcal{S})$ when the register $X$ is clear from the context. 

Let $\rho_{AB}$ be a state with $|A|=n$. We define 
$(\rho|A \in \mathcal{S}) \defeq (\sigma|\mathcal{S})$, where $\sigma_{AB}$ is the c-q state obtained by measuring the register $A$ in $\rho_{AB}$ in the computational basis. In case $\mathcal{S}=\{s\}$ is a singleton set, we shorthand $(\rho|A = s) \defeq \tr_A (\rho|A =s)$.
\end{definition}

\begin{definition}[Safe maps] \label{def:safe}
We call an isometry $V: \cH_X \otimes \cH_A \rightarrow \cH_X \otimes \cH_B$, {\em safe} on $X$ if and only if there is a collection of isometries $V_x: \cH_A\rightarrow \cH_B$ such that the following holds.  For all states $\ket{\psi}_{XA} = \sum_x \alpha_x \ket{x}_X \ket{\psi^x}_A$,
$$V  \ket{\psi}_{XA} =  \sum_x \alpha_x \ket{x}_X V_x \ket{\psi^x}_A.$$
We call a CPTP map $\Phi: \mathcal{L}( \cH_X \otimes \cH_A) \rightarrow \mathcal{L}(\cH_X \otimes \cH_B)$, {\em safe} on classical register $X$ iff there is a collection of CPTP maps $\Phi_x: \mathcal{L}(\cH_A)\rightarrow \mathcal{L}(\cH_B)$ such that the following holds.  For all c-q states $\rho_{XA} = \sum_x \Pr(X=x)_{\rho} \ketbra{x} \otimes  \rho^x_A$,
$$\Phi({\rho}_{XA}) =  \sum_x \Pr(X=x)_{\rho} \ketbra{x} \otimes \Phi_x( \rho^x_A).$$

\end{definition}

\begin{definition}[Extension] \label{def:extension} Let $$\rho_{XE}=  \sum\limits_{x \in \cX}  p(x)\ket{x}\bra{x} \otimes \rho^x_E,$$
be a c-q state. For a function $Z:\cX \rightarrow \cZ$, define the following extension of $\rho_{XE}$, 
\[ \rho_{ZXE} \defeq  \sum_{x\in \cX}  p(x) \ket{Z(x)}\bra{Z(x)} \otimes \ket{x}\bra{x} \otimes  \rho^{x}_E.\]
\end{definition}

For a pure state $\rho_{XEA}$ (with $X$ classical and $X \in \cX$) and a function $Z:\cX \rightarrow \cZ$, define $\rho_{Z\hat{Z}XEA}$ to be a pure state extension of $\rho_{XEA}$ generated via a safe isometry $V: \cH_X \rightarrow \cH_X \otimes \cH_Z \otimes \cH_{\hat{Z}}$ ($Z$ classical with copy $\hat{Z}$). We use the notation $\mathcal{M}_{A}(\rho_{AB})$ to denote measurement in the computational basis on register $A$ in state $\rho_{AB}$.

All isometries considered in this paper are safe on classical registers that they act on. Isometries  applied by adversaries can be assumed without loss of generality as safe on classical registers, by the adversary first making a (safe) copy of classical registers and then proceeding as before. This does not reduce the power of the adversary.


\begin{definition}
\label{def:infoquant}    
\begin{enumerate}
\item For $p \geq 1$ and matrix $A$,  let $\| A \|_p$ denote the {\em Schatten} $p$-norm defined as $\| A \|_p  \defeq (\tr(A^\dagger A)^{\frac{p}{2}})^{\frac{1}{p}}.$

\item For states $\rho,\sigma: \Delta(\rho,\sigma)\defeq \frac{1}{2}\Vert\rho- \sigma\Vert_1.$ We write $\rho \approx_\eps \sigma$ to denote $\Vert\rho- \sigma\Vert_1 \le \eps$. 

\item {\bf Fidelity:}  For states $\rho,\sigma: ~\F(\rho,\sigma)\defeq\|\sqrt{\rho}\sqrt{\sigma}\|_1.$ 

\item {\bf Bures metric:}  For states $\rho,\sigma: \Delta_B(\rho,\sigma)\defeq \sqrt{1-\F(\rho,\sigma)}.$ 

\item Define $d(X)_\rho \defeq \Delta_B(\rho_X,U_X)$ and  $d(X|Y )_\rho \defeq \Delta_B(\rho_{XY}, U_X \otimes \rho_Y)$. 

\item {\bf Max-divergence (\cite{Datta09}, see also~\cite{JainRS02}):}\label{dmax}  For states $\rho,\sigma$ such that $\supp(\rho) \subset \supp(\sigma)$, $$ \dmax{\rho}{\sigma} \defeq  \min\{ \lambda \in \mathbb{R} :   \rho  \leq 2^{\lambda} \sigma \}.$$ 

\item {\bf Min-entropy and conditional-min-entropy:}  For a state $\rho_{XE}$, the min-entropy of $X$ is defined as,
 $$ \hminone{X}_\rho \defeq - \dmax{\rho_{X}}{\id_X} .$$
 The conditional-min-entropy of $X$, conditioned on $E$, is defined as,
 $$ \hmin{X}{E}_\rho \defeq - \inf_{\sigma_E \in  \mathcal{D}(\cH_{E}) } \dmax{\rho_{XE}}{\id_X \otimes \sigma_E}.$$
\end{enumerate}
\end{definition}

\subsubsection{Useful facts from the literature}

For the facts stated below without citation, we refer the reader to standard textbooks~\cite{NielsenC00,WatrousQI}.

\begin{fact}[Data-processing]
\label{fact:data}
Let $\rho, \sigma$  be two states and $\cE$ be a CPTP map. Then 
\begin{itemize}
    \item $ \Delta ( \cE(\rho)  , \cE(\sigma))  \le \Delta (\rho  , \sigma).$    
     \item $ \Delta_B ( \cE(\rho)  , \cE(\sigma))  \le \Delta_B (\rho  , \sigma).$    
    \item  $\dmax{ \cE(\rho) }{ \cE(\sigma) }  \le \dmax{\rho}{ \sigma} .$    
\end{itemize}
Above are equalities if $\cE$ is a map corresponding to an isometry.
\end{fact}

\begin{fact}\label{traceavg1}
Let $\rho,\sigma$ be states. Let $\Pi$ be a projector. Then,
\begin{equation*}
    \tr(\Pi \rho)  \left\|  \frac{\Pi \rho_{} \Pi}{\tr(\Pi \rho) }-  \frac{\Pi\sigma_{} \Pi}{\tr(\Pi \sigma)} \right\|_1 \leq \| \rho_{}-\sigma_{} \|_1.
\end{equation*}
The fact also holds for $\Delta_B$, instead of $\|. \|_1$.
\end{fact}

\begin{fact}\label{fact:traceconvex} Let $\rho, \sigma$  be states such that $\rho = \sum_{x} p_x \rho^x$,  $\sigma = \sum_{x} p_x \sigma^x$, $\{\rho^x, \sigma^x\}_x$ are states and $\sum_x p_x =1$. Then, 
\begin{equation*}
    \Vert\rho- \sigma \Vert_1 \leq \sum_x p_x \Vert \rho^x -\sigma^x \Vert_1.
\end{equation*}

\end{fact}

\subsection{Useful lemmas and facts}\label{sec:usefulstuff}

\subsubsection{Pauli and Clifford gates}

\begin{fact}[Pauli twirl~\cite{DCEL09}]\label{lem:paulitwirl}
 Let $\rho \in \cD( \cH_A)$ be a state and $P , P' \in \cP(\cH_A)$ be Pauli operators such that $P \ne P'$. Then 
$$\sum_{Q \in \cP(\cH_A)}   Q^\dagger P Q \rho Q^\dagger P'^\dagger Q  =0 .$$

As an immediate corollary, we obtain that for any normal operator $M \in \cL( \cH_A)$ such that $M^\dagger M=MM^\dagger$, 
$$\sum_{Q \in \cP(\cH_A)}   Q^\dagger P Q M Q^\dagger P'^\dagger Q  =0 ,$$
since $M$ has an eigen-decomposition. 
\end{fact}

\begin{fact}[Twirl]\label{lem:cliffordtwirl}
Let $\rho \in \cD( \cH_A)$ be a state and $P , P' \in \cP(\cH_A)$ be Pauli operators such that $P \ne P'$. Let $\cSC(\cH_A)$ be the sub-group of Clifford group as defined in \cref{lem:subclifford}. Then, 
    $$ \sum_{C \in \cSC(\cH_A)}   C^\dagger P C \rho C^\dagger P' C  = 0.$$
  As an immediate corollary, we obtain that for any normal operator $M \in \cL( \cH_A)$ such that $M^\dagger M=MM^\dagger$, 
$$\sum_{C \in \cSC(\cH_A)}   C^\dagger P C M C^\dagger P'^\dagger C  =0 ,$$
since $M$ has an eigen-decomposition. 
\end{fact}
\begin{proof}
The proof follows along similar lines to the proof of Clifford twirl in \cite{BW16}.  We provide a  proof here for completeness. 

Since $\cP(\cH_A)$ is a subgroup of $\cSC(\cH_A)$, the number of left cosets of $\cP(\cH_A)$ in $\cSC(\cH_A)$ are $\frac{\vert \cSC(\cH_A) \vert}{\vert \cP(\cH_A) \vert}$. Then,

\begin{align*}
      \sum_{C \in \cSC(\cH_A)}   C^\dagger P C \rho C^\dagger P' C  &=\sum_{i=1}^{\frac{\vert \cSC(\cH_A) \vert}{\vert \cP(\cH_A) \vert}} \sum_{R\in \cP(\cH_A)} (C_iR)^\dagger P (C_iR) \rho (C_iR)^\dagger P' (C_iR)  \\
      &=\sum_{i=1}^{\frac{\vert \cSC(\cH_A) \vert}{\vert \cP(\cH_A) \vert}} \sum_{R\in \cP(\cH_A)} R^\dagger C_i^\dagger P C_iR \rho R^\dagger C_i^\dagger P' C_iR .
\end{align*}
Since $C_i^\dagger P C_i=Q_i$ for some $Q_i\in \cP(H_A)$ and if $P\neq P'$, then $Q_i\neq Q_i'$, we get:

\begin{align*}
      \sum_{C \in \cSC(\cH_A)}   C^\dagger P C \rho C^\dagger P' C  
      &=\sum_{i=1}^{\frac{\vert \cSC(\cH_A) \vert}{\vert \cP(\cH_A) \vert}} \sum_{R\in \cP(\cH_A)} R^\dagger Q_iR \rho R^\dagger Q_i'R \\
      &= \sum_{i=1}^{\frac{\vert \cSC(\cH_A) \vert}{\vert \cP(\cH_A) \vert}} 0\\
      &=0.
\end{align*}
Here second last equality follows from \cref{lem:paulitwirl}.
 
\end{proof}

\begin{fact}[Modified twirl]\label{lem:cliffordtwirl1}
 Let $\rho_{\hat{A}A}$ be the canonical purification of $\rho_A$. Let $P \in \cP(\cH_A),P' \in \cP(\cH_A)$ be Pauli operators such that $P \ne P'$. Let $\cSC(\cH_A)$ be the sub-group of Clifford group as defined in \cref{lem:subclifford}. Then, 
    $$ \sum_{C \in \cSC(\cH_A)}   (\id \otimes C^\dagger PC) \rho_{\hat{A}A} (\id \otimes C^\dagger P' C)  = 0.$$
\end{fact}
\begin{proof}
Let 
\begin{align*}
   \rho_{\hat{A}A}  & =\sum_{ Q' \in \cP(\cH_{\hat{A}}), Q \in \cP(\cH_{A})} \alpha_{Q'Q} (Q' \otimes Q)  = \sum_{Q \in \cP(\cH_{A})}  \left(\left(\sum_{Q' \in \cP(\cH_{\hat{A}})} \alpha_{Q'Q} Q'\right) \otimes Q \right)  \\
    & = \sum_{Q \in \cP(\cH_{A})}  \left(M^{Q} \otimes Q\right),
\end{align*}
where, 
$$M^Q \defeq \left(\sum_{Q' \in \cP(\cH_{\hat{A}})} \alpha_{Q'Q} Q'\right).$$
Consider, 
\begin{align*}
& \sum_{C \in \cSC(\cH_A)}   (\id \otimes C^\dagger PC) \rho_{\hat{A}A} (\id \otimes C^\dagger P' C)  \\ 
&= \sum_{C \in \cSC(\cH_A)}   (\id \otimes C^\dagger PC) \left(\sum_{R \in \cP(\cH_{A})}  \left(M^R \otimes R\right) \right) (\id \otimes C^\dagger P' C)  \\ 
     &=  \sum_{R \in \cP(\cH_{A})}   \left(M^R \otimes \sum_{C \in \cSC(\cH_A)}   ( C^\dagger PC) R( C^\dagger P' C)     \right)  \\
     & = 0. & \mbox{(\cref{lem:cliffordtwirl})}
\end{align*}
\end{proof}

\begin{fact}[Pauli $1$-design]\label{fact:bellbasis}
    Let $\rho_{AB}$ be a state. Then,
$$\frac{1}{\vert \cP(\cH_A) \vert} \sum_{Q \in \cP(\cH_A) } (Q \otimes \id)   \rho_{A B}  ( Q^\dagger \otimes \id  )  = U_{A } \otimes  \rho_B.$$
\end{fact}

\begin{fact}[$1$-design]\label{fact:notequal}     Let $\rho_{AB}$ be a state. Let $\cSC(\cH_A)$ be the sub-group of Clifford group as defined in \cref{lem:subclifford}. Then,
$$\frac{1}{\vert \cSC(\cH_A) \vert} \sum_{C \in \cSC(\cH_A) } (C \otimes \id)   \rho_{A B}  ( C^\dagger \otimes \id  )  = U_{A } \otimes  \rho_B.$$

\end{fact}

\begin{lemma}\label{lem:equalreq} Let state $\rho_{\hat{A}A}$ be the canonical purification of $\rho_A$. Let $P, Q \in \cP(\cH_A)$ be any two Pauli operators. Let $\cSC(\cH_A)$ be the sub-group of Clifford group as defined in \cref{lem:subclifford}. If $P  \ne Q$, then 
$$\frac{1}{\vert \cSC(\cH_A)\vert } \sum_{C \in \cSC(\cH_A)} (\id  \otimes C^\dagger P C)   \rho_{\hat{A}A }  ( \id \otimes C^\dagger Q^\dagger C  ) = 0,$$
else if $P=Q = \id_A$, then
$$\frac{1}{\vert \cSC(\cH_A)\vert } \sum_{C \in \cSC(\cH_A)} (\id  \otimes C^\dagger P C)   \rho_{\hat{A}A }  ( \id \otimes C^\dagger Q^\dagger C  ) = \rho_{\hat{A}A},$$
else $P =Q \ne \id_A$, then  
$$\frac{1}{\vert \cSC(\cH_A)\vert } \sum_{C \in \cSC(\cH_A)} (\id  \otimes C^\dagger P C)   \rho_{\hat{A}A }  ( \id \otimes C^\dagger Q^\dagger C  ) \approx_{\frac{2}{ \vert \cP(\cH_A) \vert-1}}  \rho_{\hat{A}} \otimes U_A.$$
\end{lemma}
\begin{proof}
When $P \neq Q$, using \cref{lem:cliffordtwirl1},
$$ \frac{1}{\vert \cSC(\cH_A)\vert } \sum_{C \in \cSC(\cH_A)} (\id  \otimes C^\dagger P C)   \rho_{\hat{A}A}  ( \id \otimes C^\dagger Q^\dagger C  ) = 0 .$$
When $P=Q = \id_A$,
\begin{align*}
    &\frac{1}{\vert \cSC(\cH_A)\vert } \sum_{C \in \cSC(\cH_A)} (\id  \otimes C^\dagger P C)   \rho_{ \hat{A}A}  ( \id \otimes C^\dagger P^\dagger C  ) \\
    & =\frac{1}{\vert \cSC(\cH_A)\vert } \sum_{C \in \cSC(\cH_A)} (\id  \otimes C^\dagger C)   \rho_{\hat{A}A}  ( \id \otimes C^\dagger C  )  \\
    &=  \rho_{ \hat{A}A}. & \mbox{($C^\dagger C = \id$)}
\end{align*} 
When $P=Q \ne \id_A$, 
\begin{align*}
    & \frac{1}{\vert \cSC(\cH_A)\vert } \sum_{C \in \cSC(\cH_A)} (\id  \otimes C^\dagger P C)   \rho_{ \hat{A}A}  ( \id \otimes C^\dagger P^\dagger C  )  \\
     & =\frac{1}{\vert \cP(\cH_A)\vert -1} \sum_{Q \in \cP(\cH_A)\setminus\id } (\id  \otimes Q)   \rho_{ \hat{A}A}  ( \id \otimes Q^\dagger   ) & \mbox{(\cref{lem:subclifford})} \\
    &= \frac{\vert \cP(\cH_A)\vert (  \rho_{\hat{A}} \otimes U_A ) -\rho_{\hat{A}A} }{\vert \cP(\cH_A)\vert-1}. & \mbox{(\cref{fact:bellbasis})}
\end{align*} 
The results follows since $ \left\|\frac{\vert \cP(\cH_A)\vert ( \rho_{\hat{A}} \otimes U_A ) -\rho_{\hat{A}A} }{\vert \cP(\cH_A)\vert-1}   - ( \rho_{\hat{A}} \otimes U_A )   \right\|_1 \leq \frac{2}{ \vert \cP(\cH_A) \vert-1}$.
\end{proof}

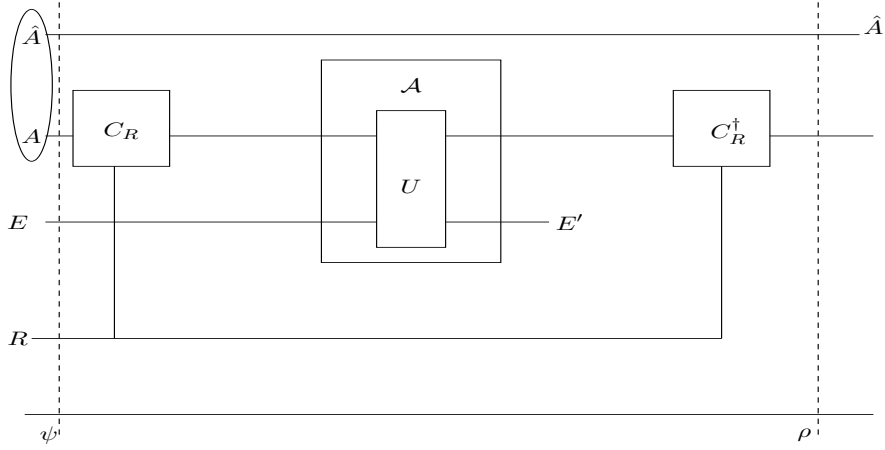
\begin{figure}
\centering
\resizebox{12cm}{6cm}{
\begin{tikzpicture}

\node at (1,6.5) {$\hat{A}$};
\node at (13.2,6.7) {$\hat{A}$};
\draw (1.2,6.5) -- (13,6.5);
\draw (1,5.5) ellipse (0.3cm and 1.5cm);
\draw (1.6,3.9) rectangle (3,5.4);
\draw (10.3,3.9) rectangle (11.7,5.4);
\node at (1,4.5) {$A$};

\node at (2.3,4.6) {$C_R$};
\node at (11.1,4.6) {$C^\dagger_{R}$};


\draw (1.2,4.5) -- (1.6,4.5);
\draw (3,4.5) -- (6,4.5);
\draw (7,4.5) -- (10.3,4.5);
\draw (11.7,4.5) -- (13.2,4.5);

\node at (0.8,0.5) {$R$};
\draw (1,0.5) -- (11,0.5);


\draw  (2.2,0.5) -- (2.2,3.9);

\draw [dashed] (1.4,-1.4) -- (1.4,7.2);
\draw [dashed] (12.4,-1.4) -- (12.4,7.2);

\draw (0.9,-1) -- (13.2,-1);

\node at (1.25,-1.4) {$\psi$};
\node at (12.2,-1.4) {$\rho$};

\node at (0.8,2.8) {$E$};

\draw (1.2,2.8) -- (6,2.8);
\draw (7,2.8) -- (8.5,2.8);

\draw (11,0.5) -- (11,3.9);


\draw (6,2.3) rectangle (7,5);
\node at (6.5,3.5) {$U$};

\node at (6.5,5.5) {$\mathcal{A}$};
\draw (5.2,2) rectangle (7.8,6);




\node at (8.8,2.8) {$E'$};
\end{tikzpicture}}

\caption{Quantum non-malleable code with shared key.}\label{fig:splitstate21}
\end{figure}

\begin{lemma}\label{lem:equal101}
 Consider \cref{fig:splitstate21}. Let state ${\psi}_{A \hat{A}}$ be the canonical purification of ${\psi}_{A}$ . Let $\psi_E$ be a state independent of ${\psi}_{A \hat{A}}$. Let  $U:(\cH_A\otimes \cH_{E})\to (\cH_{A}\otimes\cH_{E'})$ be any isometry. Let $\cSC(\cH_A)$ be the sub-group of Clifford group as defined in \cref{lem:subclifford}. Let \[ \rho_{\hat{A}AE'} = \frac{1}{\vert \cSC(\cH_A)\vert } \sum_{C \in \cSC(\cH_A)}  (   C^\dagger UC)  (\psi_{ \hat{A}A} \otimes \psi_E ) (C^\dagger  U^\dagger C) .\]Then,
 \begin{multline*}
    \rho_{\hat{A}A}   =  p \psi_{\hat{A}A} + (1- p) \left(\frac{\vert \cP(\cH_A)\vert (\psi_{\hat{A}} \otimes U_{{A}} ) -\psi_{\hat{A}A} }{\vert \cP(\cH_A)\vert-1} \right) \approx_{\frac{2}{4^{\vert A \vert}-1}}  p \psi_{\hat{A}A} + (1- p) (\psi_{\hat{A}} \otimes U_A),
 \end{multline*}
 where $p$ depends only on unitary $U $ and state $\psi_E$.
\end{lemma}
\begin{proof}Let $\Phi  : \cL(\cH_A) \to  \cL(\cH_A)$ be the CPTP map
defined by:
$\Phi(\rho_A) \defeq \tr_{E'} (U (\rho_A \otimes \psi_E)U^\dagger)$ for every $\rho_A$. Note $\Phi$ depends only on unitary $U $ and state $\psi_E$. Let $\{ M^i\}_i$ be the set of Kraus operators corresponding to $\Phi$. Let $M_i=\sum \alpha_{ij}P_j$ for $P_j\in \cP(\cH_A)$. From $\sum_iM_i^\dagger M_i=\id$, by taking trace on both sides, we get $\sum_{ij}|\alpha_{ij}|^2=1$. 

Consider,
\begin{align*}
   &\rho_{\hat{A}A}\\
   &= \frac{1}{\vert \cSC(\cH_A)\vert } \sum_{C \in \cSC(\cH_A)}     C^\dagger \Phi(C\psi_{ \hat{A}A}C^\dagger )  C  \\ 
   &=  \sum_{i}  \left(\frac{1}{\vert \cSC(\cH_A)\vert } \sum_{C \in \cSC(\cH_A)}  (   C^\dagger M_iC)  (\psi_{ \hat{A}A} ) (C^\dagger  M_i^\dagger C)  \right)  
  \\ &=  \sum_{i}  \left(\frac{1}{\vert \cSC(\cH_A)\vert }\sum_{j,j'}\alpha_{ij}\alpha^*_{ij'} \sum_{C \in \cSC(\cH_A)}   (   C^\dagger P_{ij}C)  (\psi_{ \hat{A}A} ) (C^\dagger  P_{ij'}^\dagger C)  \right)   \\ &=  \sum_{i,j}  \left(\vert\alpha_{ij}\vert^2\frac{1}{\vert \cSC(\cH_A)\vert } \sum_{C \in \cSC(\cH_A)}   (   C^\dagger P_{ij}C)  (\psi_{ \hat{A}A} ) (C^\dagger  P_{ij'}^\dagger C)  \right)  \\ &+  \sum_{i,j\neq j'}  \left(\alpha_{ij}\alpha^*_{ij'}\frac{1}{\vert \cSC(\cH_A)\vert } \sum_{C \in \cSC(\cH_A)}   (   C^\dagger P_{ij}C)  (\psi_{ \hat{A}A} ) (C^\dagger  P_{ij'}^\dagger C)  \right)  \\
&=  \sum_{i,j}  \left(\vert\alpha_{ij}\vert^2\frac{1}{\vert \cSC(\cH_A)\vert } \sum_{C \in \cSC(\cH_A)}   (   C^\dagger P_{ij}C)  (\psi_{ \hat{A}A} ) (C^\dagger  P_{ij}^\dagger C)  \right)  &\mbox{(\cref{lem:cliffordtwirl1})}\\
    \end{align*}   
\begin{align*}
    &=  \sum_{i,j,P_{ij}=\id}  \left(\vert\alpha_{ij}\vert^2\frac{1}{\vert \cSC(\cH_A)\vert } \sum_{C \in \cSC(\cH_A)}   (   C^\dagger P_{ij}C)  (\psi_{ \hat{A}A} ) (C^\dagger  P_{ij}^\dagger C)  \right)  \\ &+  \sum_{i,j,P_{ij}\neq \id}  \left(\vert\alpha_{ij}\vert^2\frac{1}{\vert \cSC(\cH_A)\vert } \sum_{C \in \cSC(\cH_A)}   (   C^\dagger P_{ij}C)  (\psi_{ \hat{A}A} ) (C^\dagger  P_{ij}^\dagger C)  \right)  \\
    &=  \sum_{i,j,P_{ij}=\id}  \left(\vert\alpha_{ij}\vert^2\psi_{ \hat{A}A}  \right)  \\ &+  \sum_{i,j,P_{ij}\neq \id}  \vert\alpha_{ij}\vert^2\left(  \frac{\vert \cP(\cH_A)\vert ( \psi_{\hat{A}} \otimes U_A ) -\psi_{\hat{A}A} }{\vert \cP(\cH_A)\vert-1}  \right) \\ 
    &=  p\psi_{ \hat{A}A}  +  (1-p)\left(  \frac{\vert \cP(\cH_A)\vert ( \psi_{\hat{A}} \otimes U_A ) -\psi_{\hat{A}A} }{\vert \cP(\cH_A)\vert-1}  \right).
\end{align*} 

Note that $p$ depends only on the adversary CPTP map $\Phi$. The approximation in the result follows from \cref{fact:traceconvex} by observing that
\begin{equation*}
\left\|\frac{\vert \cP(\cH_A)\vert ( U_{A} \otimes \psi_{\hat{A}} ) -\psi_{A\hat{A}} }{\vert \cP(\cH_A)\vert-1}   - ( U_{A} \otimes \psi_{\hat{A}} )   \right\|_1 \leq \frac{2}{ \vert \cP(\cH_A) \vert-1}.
\end{equation*}
\end{proof}

\section{A rate $1/2$, $2$-split quantum secure, non-malleable randomness encoder\label{sec:2nmre}}
In some applications of cryptography one only needs to be able to encode randomness
i.e., security is not required to hold for arbitrary, adversarially chosen
messages. For example, in applications of non-malleable codes to tamper-resilient security, the messages that are encoded are typically randomly generated secret keys. To exploit this, in the work of~\cite{KOS18}, the notion of “Nonmalleable Randomness Encoders” (NMREs) is introduced.  One can think of “Nonmalleable Randomness Encoders” as a relaxation of non-malleable codes in
the following sense: NMREs output a random message along with its corresponding non-malleable encoding.

We first provide the definition of non-malleable randomness encoder as stated in~\cite{KOS18}.  Let $r, m, n, n_1, n_2, k$ be positive integers and $\eps, \delta>0$.

\begin{definition}[\cite{KOS18}]\label{def:nmrec}
 Let $(\nmreenc, \nmredec)$ be such that $\nmreenc: \{0,1 \}^r \rightarrow  \{0,1 \}^{m}  \times \{ 0,1\}^{n_1} \times  \{0,1\}^{n_2}$ and $\nmredec : \{ 0,1\}^{n_1} \times  \{0,1\}^{n_2} \rightarrow \{ 0,1\}^m $. Denote, $$\nmreenc(.) = (\nmreenc_1(.) , \nmreenc_2(.) ),$$where $\nmreenc_1: \{0,1 \}^r \rightarrow  \{0,1 \}^{m}  $ and $\nmreenc_2: \{0,1 \}^r \rightarrow  \{ 0,1\}^{n_1} \times  \{0,1\}^{n_2}$. 
 Let $\nmreenc(U_r) =MXY$. Here, $\nmreenc_1$ specifies the procedure to generate random message $M$ and $\nmreenc_2$ specifies the procedure to generate $(X,Y)$.
 
 We say $(\nmreenc, \nmredec)$ is an $\eps$-non-malleable randomness encoder if the following holds: 
\begin{itemize}
    \item \textbf{correctness:} $\Pr_{}( \nmredec( X,Y) =M)=1$.
     \item \textbf{non-malleability:} For every $f: \{ 0,1\}^{n_1} \rightarrow \{0,1 \}^{n_1}$ and $g: \{ 0,1\}^{n_2} \rightarrow \{0,1 \}^{n_2}$, we have 
     $$\nmredec(f(X),g(Y))M \approx_{\eps} p_{(f,g)} MM + (1-p_{(f,g)}) M'_{(f,g)} \otimes M,$$ 
     where $p_{(f,g)}$ and $M'_{(f,g)}$ depend only on functions $(f,g)$.
\end{itemize}
The rate of the code is $\frac{m}{n_1+n_2}$.    
\end{definition}

We now formally define quantum secure non-malleable randomness encoder and give a construction for the same.

\begin{definition}\label{def:nmreq}
 Let $(\nmreenc, \nmredec)$ be such that $\nmreenc: \{0,1 \}^r \rightarrow  \{0,1 \}^{m}  \times \{ 0,1\}^{n_1} \times  \{0,1\}^{n_2}$ and $\nmredec : \{ 0,1\}^{n_1} \times  \{0,1\}^{n_2} \rightarrow \{ 0,1\}^m $. Denote, $$\nmreenc(.) = (\nmreenc_1(.) , \nmreenc_2(.) ),$$where $\nmreenc_1: \{0,1 \}^r \rightarrow  \{0,1 \}^{m}  $ and $\nmreenc_2: \{0,1 \}^r \rightarrow  \{ 0,1\}^{n_1} \times  \{0,1\}^{n_2}$. 
 Let $\nmreenc(U_r) =\rho_{MXY}$. Here, $\nmreenc_1$ specifies the procedure to generate random message $M$ and $\nmreenc_2$ specifies the procedure to generate $(X,Y)$.
 
 We say $(\nmreenc, \nmredec)$ is an $\eps$-quantum secure non-malleable randomness encoder if the following holds: 
\begin{itemize}
    \item \textbf{correctness:} $\Pr_{}( \nmredec( X,Y) =M)_\rho=1$.
     \item  \textbf{non-malleability:} Let $\ket{ \psi}_{E_1 E_2}$ be a pure state independent of $XY$. Let  $U:\cL(\cH_X\otimes \cH_{E_1})\to\cL(\cH_{X'}\otimes\cH_{\hat{X}'}\otimes\cH_{E'_1})$ and $V:\cL(\cH_Y\otimes \cH_{E_2})\to\cL(\cH_{Y'}\otimes\cH_{\hat{Y}'}\otimes \cH_{E'_2})$ be two isometries along with a quantum state $\ket\psi_{E_1 E_2}$ which captures the shared entanglement between the non-communicating tampering adversaries. Let $ \tilde{\sigma} = (U \otimes V) (\rho_{MXY} \otimes  \psi) (U^\dagger \otimes V^\dagger)$~\footnote{Note registers $X'$ and $Y'$ are classical with copy registers $\hat{X}'$ and $\hat{Y}'$ respectively.}. For every $\mathcal{A}=(U,V,\ketbra{\psi})$, we have 
$$ \tilde{\sigma}_{\nmredec(X',Y')M } \approx_{\eps} p_{\mathcal{A}} \rho_{MM} + (1-p_{\mathcal{A}})   \eta_{M'}^{\mathcal{A}} \otimes \rho_M,$$where 
     $p_{\mathcal{A}}$ and $\eta_{M'}^{\mathcal{A}}$ depend only on adversary $\mathcal{A}$. We used $\rho_{MM}$ to denote two copies of $\rho_M$.

\end{itemize}
 
 The rate of the code is $\frac{m}{n_1+n_2}$.
    
\end{definition}

In this work, we present a scheme $(\nmreenc, \nmredec)$ such that $r=n_1+n_2$, $\nmreenc_2(.)$ is identity function and functions $\nmreenc_1, \nmredec$ are both the same as the  $2\nmext :  \{ 0,1\}^{n} \times  \{0,1\}^{\delta n} \rightarrow \{ 0,1\}^{(0.5-\delta)n} $ function from \cref{alg:2nmExt} and \cref{thm:2nmext}. The rate of our scheme is close to $\frac{1}{2}$ (for tiny constant $\delta$).

\begin{lemma}[Rate $1/2$, $2$-split, quantum secure, non-malleable randomness encoder]\label{lem:qnmcodesfromnmext} 
    Let $2\nmext: \{0,1 \}^n \times \{0,1 \}^{\delta n} \rightarrow \{ 0,1\}^{ (1/2 -\delta )n}$ be the $(n-k, \delta n-k,\epsilon)$-quantum secure non-malleable extractor from \cref{alg:2nmExt} for $k=O(n^{1/4})$ and $\epsilon=2^{-n^{\Omega(1)}}$. Consider~\cref{fig:splitstate6}. Let $\rho_{X Y }$ be a state such that 
    \[  \rho_{X Y } = \rho_{X  } \otimes \rho_{ Y  } \quad ; \quad \rho_X=U_n \quad ;\quad \rho_Y = U_{\delta n}.  \]Let $\ket{ \psi}_{E_1 E_2}$ be a pure state independent of $\rho$.  Let  $U:\cL(\cH_X\otimes \cH_{E_1})\to\cL(\cH_{X'}\otimes\cH_{\hat{X}'}\otimes\cH_{E'_1})$ and $V:\cL(\cH_Y\otimes  \cH_{E_2})\to\cL(\cH_{Y'}\otimes\cH_{\hat{Y}'}\otimes \cH_{E'_2})$ be two isometries along with a quantum state $\ket\psi_{E_1 E_2}$ which captures the shared entanglement between the non-communicating tampering adversaries. Let $ \tilde{\sigma} = (U \otimes V) (\rho \otimes  \psi ) (U^\dagger \otimes V^\dagger)$~\footnote{Note registers $X'$ and $Y'$ are classical with copy registers $\hat{X}'$ and $\hat{Y}'$ respectively.}. Then,

\begin{itemize}
  \item $ \Vert  2\nmext(X,Y)X - U_{r} \otimes U_n  \Vert_1 \leq \eps$ and $\Vert  2\nmext(X,Y)Y - U_{r} \otimes U_{\delta n}  \Vert_1 \leq \eps,$
    \item $\sigma_{RR'E_2'} = p_{\mathcal{A}} \sigma^1_{RR'E_2'} + (1-p_{\mathcal{A}}) \sigma^2_{RR'E_2'}$,
    \item $p_{\mathcal{A}}\Vert \sigma^1_{RE_2'}-  U_{(1/2-\delta)n} \otimes \sigma^1_{E_2'} \Vert_1 +  (1-p_{\mathcal{A}})  \Vert \sigma^2_{RR'E_2'}-  U_{(1/2-\delta)n} \otimes \sigma^2_{R'E_2'}  \Vert_1  \leq 4(2^{-k}+ \eps)$,
\end{itemize}
where $\sigma^1=  \tilde{\sigma}  \vert ( XY =  {X}' {Y}' ) $, $\sigma^2= \tilde{\sigma} \vert  ( XY \ne {X}' {Y}'  )$ and $p_{\mathcal{A}}= \Pr(    XY = {X}' {Y}')_{\tilde{\sigma}}$. Furthermore, $\Pr(R=R')_{\sigma^1}=1$. 
\end{lemma}
\begin{proof}
    The proof of the above lemma follows using similar arguments as that of Theorem~$1$ in~\cite{ABJ22} and using \cref{thm:2nmext}.
\end{proof}

\begin{figure}[h]
\centering
\resizebox{12cm}{6cm}{
\begin{tikzpicture}


\node at (3.5,2.8) {$Y$};
\draw (3.8,2.8) -- (5,2.8);
\draw (4.67,4.5) -- (5,4.5);
\draw (4.67,2.8) -- (4.67,4.5);
\draw (4.8,4.5) -- (5,4.5);
\draw (5,4.5) -- (11,4.5);

\node at (14.5,4.7) {$R$};
\draw (13,4.5) -- (14.5,4.5);

\draw [dashed] (4.98,-2.5) -- (4.98,6.5);
\draw [dashed] (10,-2.5) -- (10,6.5);
\draw [dashed] (14,-2.5) -- (14,6.5);

\node at (6.8,1.6) {$\ket{\psi}_{E_1E_2}$};
\node at (4.79,-2) {$\rho$};
\node at (9.8,-2) {$\tilde{\sigma}$};
\node at (13.8,-2) {${\sigma}$};

\node at (7.8,3) {$Y'$};

\draw (4.65,2.8) -- (6.3,2.8);

\draw (7.3,2.8) -- (11,2.8);

\node at (14.5,2.6) {$R'$};
\draw (13,2.4) -- (14.5,2.4);

\node at (3.5,0.2) {$X$};

\draw (4.67,-1.2) -- (9.5,-1.2);
\draw (4.67,-1.2) -- (4.67,0.2);
\draw (9.5,-1.2) -- (9.5,4);
\draw (9.5,4) -- (11,4);

\node at (7.8,0.0) {$X'$};
\draw (3.8,0.2) -- (6.3,0.2);
\draw (7.3,0.2) -- (11,0.2);

\draw (6.3,2) rectangle (7.3,3);
\node at (6.8,2.5) {$V$};
\draw (6.3,0) rectangle (7.3,1);
\node at (6.8,0.5) {$U$};

\node at (6.8,-0.4) {$\mathcal{A}=(U,V,\psi)$};
\draw (5.2,-0.8) rectangle (8.1,3.8);

\draw (5.8,1.5) ellipse (0.3cm and 1cm);
\node at (5.8,2) {$E_2$};
\draw (6,2.2) -- (6.3,2.2);
\node at (7.7,2) {$E'_2$};
\draw (7.3,2.2) -- (8.8,2.2);
\draw (8.8,2.2) -- (8.8,5.3);
\draw (8.8,5.3) -- (14.5,5.3);
\node at (14.5,5.5) {$E'_2$};

\node at (5.8,1) {$E_1$};
\draw (6,0.7) -- (6.3,0.7);
\node at (7.7,0.9) {$E'_1$};
\draw (7.3,0.7) -- (7.5,0.7);


\draw (11,-0.5) rectangle (13,3.2);
\draw (11,3.5) rectangle (13,5);
\node at (12,4.3) {${2\nmext}$};
\node at (12,1.5) {${2\nmext}$};
\node at (7.8,2.5) {$\hat{Y}'$};
\node at (7.8,0.5) {$\hat{X}'$};

\draw (7.3,2.6) -- (7.6,2.6);
\draw (7.3,0.4) -- (7.6,0.4);


\end{tikzpicture}}
\caption{Rate $1/2$, $2$-split quantum secure non-malleable randomness encoder.}\label{fig:splitstate6}
\end{figure}
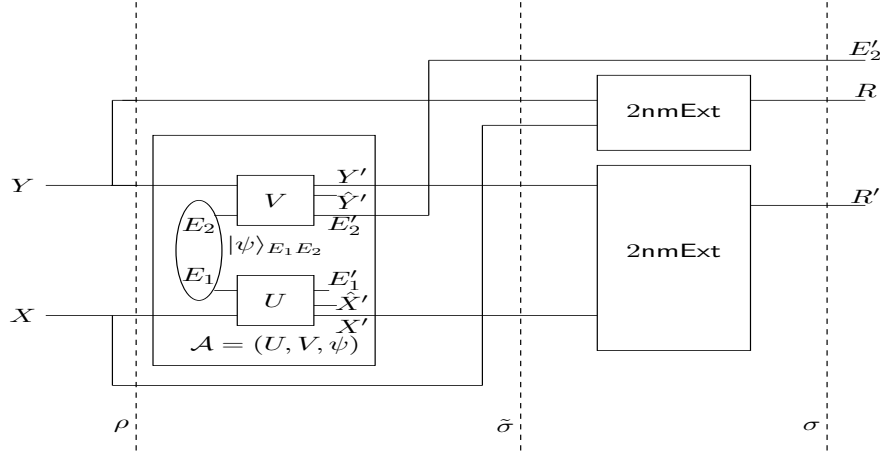

\section{A rate $1/11$, $3$-split, quantum non-malleable code}\label{sec:qnmc67}

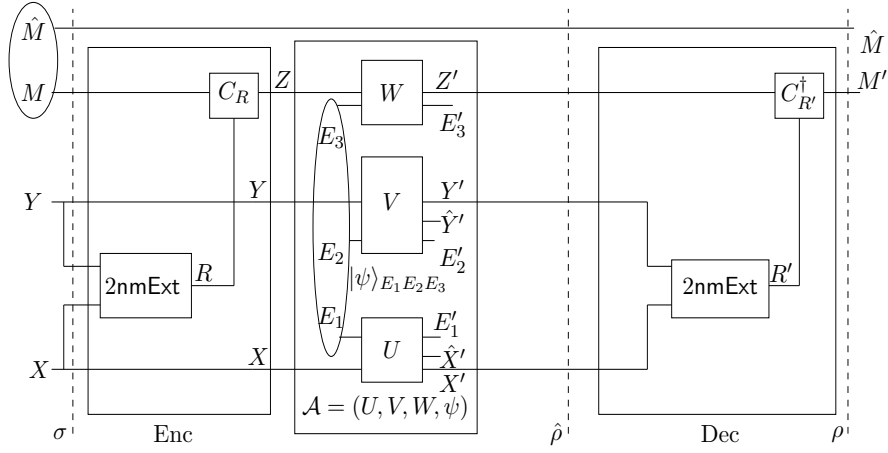
\begin{figure}[h]
\centering

\resizebox{12cm}{6cm}{

\begin{tikzpicture}

\node at (0.9,4.5) {$M$};
\node at (0.9,5.5) {$\hat{M}$};
\draw (0.9,5) ellipse (0.4cm and 0.9cm);
\draw (1.25,5.5) -- (14.4,5.5);
\node at (0.9,2.8) {$Y$};
\draw (1.2,4.5) -- (3.8,4.5);
\node at (5.8,3.8) {$E_3$};
\node at (7.8,4) {$E_3'$};
\draw (7.3,4.3) -- (7.8,4.3);

\draw (1.2,2.8) -- (4.5,2.8);
\node at (1,0.2) {$X$};
\draw (1.4,2.8) -- (1.4,1.8);
\draw (1.4,1.8) -- (2,1.8);
\draw (1.4,1.2) -- (2,1.2);
\draw (1.4,0.2) -- (1.4,1.2);
\draw (1.2,0.2) -- (4.5,0.2);
\draw (3.5,1.5) -- (4.2,1.5);
\node at (3.7,1.7) {$R$};
\node at (4.2,4.5) {$C_R$};
\node at (5,4.7) {$Z$};
\node at (7.7,4.7) {$Z'$};
\draw  (4.2,1.5) -- (4.2,4.1);
\draw (3.8,4.1) rectangle (4.6,4.8);
\draw (4.6,4.5) -- (6.3,4.5);

\draw (13.1,4.1) rectangle (13.9,4.8);
\draw (13.5,4.1) -- (13.5,1.5);
\draw (13,1.5) -- (13.5,1.5);
\node at (13.2,1.7) {$R'$};
\node at (14.7,4.7) {$M'$};
\node at (14.7,5.3) {$\hat{M}$};


\draw (2,1) rectangle (3.5,2);
\node at (2.7,1.5) {$2\nmext$};
\draw (6.3,4) rectangle (7.3,5);
\node at (6.8,4.5) {$W$};

\draw [dashed] (1.55,-0.8) -- (1.55,5.8);
\draw [dashed] (9.7,-0.8) -- (9.7,5.8);
\draw [dashed] (14.3,-0.8) -- (14.3,5.8);

\node at (6.9,1.6) {$\ket{\psi}_{E_1E_2E_3}$};
\node at (1.34,-0.8) {$\sigma$};
\node at (9.5,-0.8) {$\hat{\rho}$};
\node at (14.13,-0.8) {$\rho$};

\node at (4.6,3) {$Y$};
\node at (7.8,3) {$Y'$};

\draw (4.5,2.8) -- (6.3,2.8);
\draw (7.3,2.8) -- (11,2.8);

\node at (4.6,0.4) {$X$};
\node at (7.8,0.0) {$X'$};

\draw (4.5,0.2) -- (6.3,0.2);
\draw (7.3,0.2) -- (11,0.2);

\draw (6.3,2) rectangle (7.3,3.5);
\node at (6.8,2.8) {$V$};
\draw (6.3,0) rectangle (7.3,1);
\node at (6.8,0.5) {$U$};

\node at (6.7,-0.4) {$\mathcal{A}=(U,V,W,\psi)$};
\draw (5.2,-0.8) rectangle (8.2,5.3);

\draw (5.8,2.4) ellipse (0.3cm and 2cm);
\node at (5.8,2) {$E_2$};
\draw (6.1,2.2) -- (6.3,2.2);
\draw (5.9,4.3) -- (6.3,4.3);
\draw (7.3,4.5) -- (13.1,4.5);
\draw (13.9,4.5) -- (14.5,4.5);
\node at (7.8,1.9) {$E_2'$};
\draw (7.3,2.2) -- (7.5,2.2);
\node at (5.8,1) {$E_1$};
\draw (5.93,0.7) -- (6.3,0.7);
\node at (7.7,0.9) {$E_1'$};
\draw (7.3,0.7) -- (7.6,0.7);


\draw (1.8,-0.5) rectangle (4.8,5.2);
\draw (10.2,-0.5) rectangle (14.1,5.2);
\draw (11.4,1) rectangle (13,1.9);
\node at (12.2,-0.8) {$\dec$};
\node at (3.2,-0.8) {$\enc$};
\draw (11,1.8) -- (11,2.8);
\draw (11,1.8) -- (11.4,1.8);
\draw (11,1.2) -- (11.4,1.2);
\draw (11,0.2) -- (11,1.2);
\node at (12.2,1.5) {$2\nmext$};
\node at (13.5,4.5) {$C^\dagger_{R'}$};
\node at (7.8,0.4) {$\hat{X}'$};
\draw (7.3,0.4) -- (7.6,0.4);
\node at (7.8,2.5) {$\hat{Y}'$};
\draw (7.3,2.5) -- (7.6,2.5);
\end{tikzpicture} }
\caption{Rate $1/11, ~3$-split quantum non-malleable code.}\label{fig:splitstate2}
\end{figure}

\begin{figure}[h]
\centering

\resizebox{12cm}{6cm}{

\begin{tikzpicture}

\node at (0.9,4.5) {$M$};
\node at (0.9,5.5) {$\hat{M}$};
\draw (0.9,5) ellipse (0.4cm and 0.9cm);
\draw (1.25,5.5) -- (14.4,5.5);
\node at (0.9,2.8) {$Y$};
\draw (1.2,4.5) -- (4.8,4.5); 
\node at (5.8,3.8) {$E_3$};
\node at (12.8,4) {$E_3'$};
\draw (12.3,4.3) -- (12.8,4.3);

\draw (1.2,2.8) -- (4.5,2.8);
\node at (1,0.2) {$X$};
\draw (1.5,1.8) -- (4.5,1.8);
\draw [dashed] (4.5,1.8) -- (6.5,1.8);
\draw  (6.5,1.8) -- (8.5,1.8);
\draw (1.5,1.2) -- (4.5,1.2);
\draw (1.5,1.2) -- (1.5,0.2);
\draw (1.5,2.8) -- (1.5,1.8);
\draw [dashed] (4.5,1.2) -- (6.5,1.2);
\draw  (6.5,1.2) -- (8.5,1.2);
\draw (1.2,0.2) -- (4.5,0.2);
\draw (9.8,1.5) -- (10.5,1.5);
\node at (10.7,1.7) {${R}$};
\node at (10.5,4.5) {$C_{{R}}$};
\node at (11.1,4.7) {$Z$};
\node at (12.7,4.7) {$Z'$};
\draw  (10.5,1.5) -- (10.5,4.1);
\draw (10.1,4.1) rectangle (10.9,4.8);
\draw (4.6,4.5) -- (6.3,4.5);

\draw (13.1,4.1) rectangle (13.9,4.8);
\draw (13.5,4.1) -- (13.5,1.5);

\node at (13.2,1.7) {$R'$};
\node at (14.7,4.7) {$M'$};
\node at (14.7,5.3) {$\hat{M}$};


\draw (8.5,1) rectangle (9.8,2);
\node at (9.15,1.5) {$2\nmext$};

\draw (11.3,4) rectangle (12.3,5);
\node at (11.8,4.5) {$W$};

\draw [dashed] (3.78,-0.8) -- (3.78,6);
\draw [dashed] (9.9,-0.8) -- (9.9,6);
\draw [dashed] (8.4,-0.8) -- (8.4,6);
\draw [dashed] (14.3,-0.8) -- (14.3,6);

\node at (3.64,-0.8) {$\sigma$};
\node at (8.24,-0.8) {$\tilde{\tau}$};
\node at (9.74,-0.8) {${\tau}$};
\node at (14.13,-0.8) {$\rho$};

\node at (4.6,3) {$Y$};
\node at (7.8,3) {$Y'$};
\node at (7.8,2.5) {$\hat{Y}'$};
\draw (7.3,2.5) -- (7.6,2.5);
\draw (4.5,2.8) -- (6.3,2.8);
\draw (7.3,2.8) -- (8.2,2.8);

\draw  (8.2,0.6) -- (8.2,0.9);
\draw [dashed] (8.2,0.9) -- (8.2,2.8);
\draw  (8.2,2.5) -- (8.2,2.8);
\draw (8.2,0.6) -- (8.5,0.6);

\node at (4.6,0.4) {$X$};
\node at (7.8,0.0) {$X'$};
\node at (7.8,0.4) {$\hat{X}'$};
\draw (7.3,0.4) -- (7.6,0.4);
\draw (4.5,0.2) -- (6.3,0.2);
\draw (7.3,0.2) -- (8.5,0.2);

\draw (6.3,2) rectangle (7.3,3.5);
\node at (6.8,2.8) {$V$};
\draw (6.3,0) rectangle (7.3,1);
\node at (6.8,0.5) {$U$};

\node at (6.7,-0.4) {$\mathcal{A}=(U,V,W,\psi)$};
\draw (5.2,-0.8) rectangle (8.1,5.3);

\draw (5.8,2.4) ellipse (0.3cm and 2cm);
\node at (5.8,2) {$E_2$};
\draw (6.1,2.2) -- (6.3,2.2);
\draw (5.9,4.3) -- (9.8,4.3);
\draw [dashed] (9.8,4.3) -- (11.1,4.3);
\draw (11.1,4.3) -- (11.3,4.3);
\draw (4.8,4.5) -- (10.1,4.5);
\draw (10.9,4.5) -- (11.3,4.5);
\draw (12.3,4.5) -- (13.1,4.5);
\draw (13.9,4.5) -- (14.5,4.5);
\node at (7.8,2) {$E_2'$};
\draw (7.3,2.2) -- (7.5,2.2);
\node at (5.8,1) {$E_1$};
\draw (5.93,0.7) -- (6.3,0.7);
\node at (7.7,0.9) {$E_1'$};
\draw (7.3,0.7) -- (7.6,0.7);



\draw (8.5,0) rectangle (9.8,0.9);
\draw (9.8,0.4) -- (13.5,0.4);
\draw (13.5,0.4) -- (13.5,1.5);
\node at (9.2,0.5) {$2\nmext$};
\node at (13.5,4.5) {$C^\dagger_{R'}$};

\end{tikzpicture} }
\caption{Rate $1/11, ~3$-split quantum non-malleable code in the modified picture.}\label{fig:splitstate4}
\end{figure}
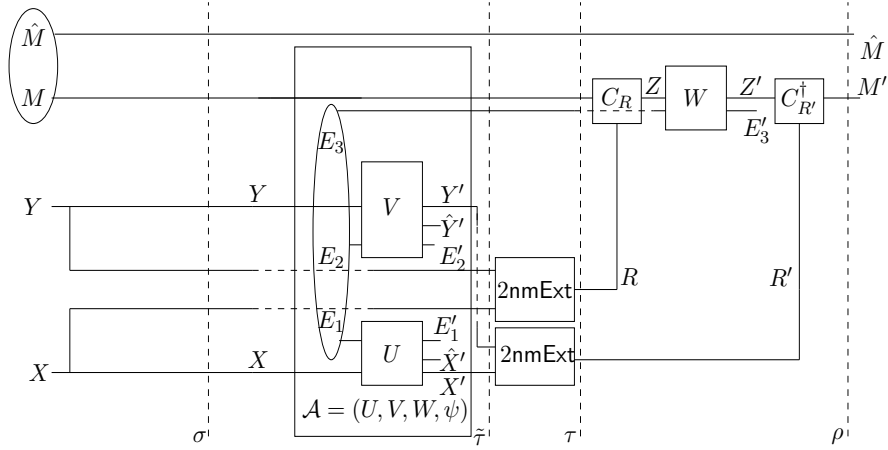
We describe our candidate $3$-split non-malleable code for quantum messages in \cref{fig:splitstate2}. Our result is as follows. 
\begin{theorem}\label{thm:main}
Let $n,k$ be positive integers and $\eps,  \delta>0$, where $k=O(n^{1/4})$ and $\epsilon=2^{-n^{\Omega(1)}}$. Let $2\nmext$ in \cref{fig:splitstate2} be from \cref{lem:qnmcodesfromnmext}. Let $\sigma_{M \hat{M} }$ be the canonical purification of state $\sigma_M$. Let classical registers $\sigma_{XY}$ be independent of $\sigma_{M \hat{M} }$ such that 
\[ \sigma_{XY} = \sigma_X \otimes \sigma_Y \quad ; \quad \sigma_X= U_n \quad ;\quad  \sigma_Y= U_{\delta n}.\]
Let $R=2\nmext(X,Y)$. Let $ \vert M\vert=  \frac{(1/2-\delta)n}{5} $ and $C_R$ be the random Clifford unitary (on $\vert M \vert$-qubits) picked using randomness $\sigma_R$ drawn from $\cSC(\cH_M)$~(see \cref{lem:subclifford}).

Let $\enc : \cL( \cH_M) \to \cL(\cH_Z \otimes \cH_Y \otimes \cH_X)$ be the encoding CPTP map and $\dec  : \cL(\cH_{Z'} \otimes \cH_{Y'}\otimes \cH_{X'}) \to  \cL(\cH_{M'})$ be the decoding CPTP map as shown in \cref{fig:splitstate2}. 
Note that $Z$ is the quantum part of ciphertext and $(X,Y)$ are the classical parts of ciphertext. Then, $(\enc,\dec)$ as specified above is an $\eps'$-$3$-split non-malleable code for $\sigma_{\hat{M}M}$, where the registers $(Z,Y,X)$ correspond to $3$ parts of the codeword, and $\eps' = 2( 4(2^{-k}+ \eps)+{\frac{1}{4^{\vert M \vert}-1}})$.
\end{theorem}
\begin{proof}
To show that $(\enc,\dec)$ is an $\eps'$-$3$-split non-malleable code for quantum messages, it suffices to show that for every $\mathcal{A}=(U,V,W,\psi)$ it holds that (in \cref{fig:splitstate2})
\begin{equation}\label{eq:finalgoal}
    (\rho)_{\hat{M}M'} \approx_{\eps'} p_{\mathcal{A}} \sigma_{\hat{M}M}  + (1-p_\mathcal{A}) (\sigma_{\hat{M}} \otimes  \gamma^{\mathcal{A}}_{M'}),
\end{equation}where $(p_{\mathcal{A}}, \gamma^{\mathcal{A}}_{M'})$ depend only on the $3$-split adversary $\cA$. Note that~\cref{fig:splitstate2},~\cref{fig:splitstate4} are equivalent, except for the delayed action of unitary $W$. We show that \cref{eq:finalgoal} holds in \cref{fig:splitstate4} which completes the proof. 

Consider the state $\tilde{\tau}$ in~\cref{fig:splitstate4}. Note $(\tilde{\tau})_{\hat{M}M}$ is a pure state (thus independent of other registers in $\tilde{\tau}$) and $\tilde{\tau} = (U \otimes V)  (\sigma \otimes \ketbra{\psi}_{E_1E_2E_3}) (U \otimes V)^\dagger.$ Let,
\begin{equation}\label{eq:7thm1}
   \tau^0= \tilde{\tau} \vert  ({X}'{Y}' \ne XY) ;\quad \tau^1= \tilde{\tau}\vert  ({X}'{Y}' = XY)   ;\quad  p_{\sm}= \Pr( {X}'{Y}' = XY )_{\tilde{\tau}}.
\end{equation}
 Using \cref{lem:qnmcodesfromnmext}, we have
\begin{equation}\label{eq:2classical198}
     (\tau)_{{R}R'E_2'E_3M\hat{M}}= p_{\sm} (\tau^1)_{{R}R'E_2'E_3M\hat{M}} + (1-p_{\sm}) (\tau^0)_{{R}R'E_2'E_3M\hat{M}}
    \end{equation}
    and
\begin{multline}\label{eq:2classical1}
        p_{\sm} \Vert (\tau^1)_{{R}E_2'E_3M\hat{M}} -  U_{\vert R \vert} \otimes (\tau^1)_{E_2'E_3M\hat{M}} \Vert_1 + \\ (1-p_{\sm})\Vert (\tau^0)_{{R}R'E_2'E_3M\hat{M}} -  U_{\vert R \vert} \otimes (\tau^0)_{R'E_2'E_3M\hat{M}} ) \Vert_1  \leq 4(2^{-k}+\eps).
    \end{multline}Further, we have
    \begin{equation}\label{Eq:file123}
         \Pr({R} =R')_{\tau^1}=1.
    \end{equation}
 Let $\Phi$ be a CPTP map from registers ${R}R'E_3M\hat{M}$ to $M'\hat{M}$ (i.e. $\Phi$ maps state $\tau$ to $\rho$) in~\cref{fig:splitstate4}. One can note, 
\begin{equation}\label{eq:2classical19845}
    \Phi(U_{\vert R \vert} \otimes \tau^0_{R'E_3M\hat{M}}) =  \eta^{\mathcal{A}}_{M'}   \otimes\sigma_{\hat{M}}
    \end{equation}follows from using~\cref{fact:notequal} followed by~\cref{fact:data}. Further state $\eta^{\mathcal{A}}_{M'}$ depends only on adversary $\mathcal{A}$. Let $\tilde{ \tau}^1$ be the state such that 
\[ \tilde{ \tau}^1_{{R}R'E_3M\hat{M}} =\tilde{ \tau}^1_{{R}R'} \otimes \tau^0_{E_3M\hat{M}} \quad ; \quad   \Pr({R} =R')_{\tilde{ \tau}^1}=1 \quad ; \quad \tilde{ \tau}^1_{R'}=U_{\vert R \vert}. \]
Note, we have\begin{equation}\label{eq:2classical19845r4}
   \Phi( \tilde{ \tau}^1_{ {R}R'E_3M\hat{M}}) \approx_{\frac{2}{4^{\vert M \vert}-1}} p \sigma_{M'\hat{M}} + (1-p) U_{M'} \otimes \sigma_{\hat{M}} 
    \end{equation}follows from using~\cref{lem:equal101}. Note $\sigma_{M'\hat{M}} \equiv \sigma_{M\hat{M}}$. Consider, 
\begin{align}
  \rho_{M'\hat{M}} 
   &=\Phi(\tau_{{R}R'E_3M\hat{M}})\nonumber \\
   &= p_{\sm}  \Phi(\tau^1_{{R}R'E_3M\hat{M}}) \nonumber \\
   & \quad \quad + (1-p_{\sm})  \Phi(\tau^0_{{R}R'E_3M\hat{M}})\label{eq:11fromeq20} \\
     & \approx_{ 4(2^{-k}+ \eps) +  {\frac{2}{4^{\vert M \vert}-1}}} p_{\sm}  \left( p \sigma_{M'\hat{M}} + (1-p) (U_{{M'}} \otimes \sigma_{\hat{M}}) \right) \nonumber \\
     & \quad \quad + (1-p_{\sm})  \Phi(\tau^0_{{R}R'E_3M\hat{M}})\label{eq:11fromclaimnotequal}\\
   &=  p_{\sm}\cdot p\cdot \sigma_{M'\hat{M}}+  p_{\sm}  (1-p) (U_{{M'}} \otimes \sigma_{\hat{M}})  \nonumber\\
     & \quad \quad + (1-p_{\sm})  \Phi(\tau^0_{{R}R'E_3M\hat{M}}) \nonumber  \\
       &\approx_{4(2^{-k}+ \eps) }  p_{\sm}\cdot p\cdot \sigma_{M'\hat{M}}+  p_{\sm}  (1-p) (U_{{M'}} \otimes \sigma_{\hat{M}})  \nonumber\\
     & \quad \quad + (1-p_{\sm}) \eta^{\mathcal{A}}_{M'} \otimes \sigma_{\hat{M}}\label{eq:111fromclaimnotequal} \\
     & = p_{\sm}\cdot p\cdot  \sigma_{M'\hat{M}}  \nonumber\\
     & \quad\quad + 
 \left(p_{\sm}  (1-p)  U_{M'}  + (1-p_{\sm})  \eta^{\mathcal{A}}_{M'}  \right)\otimes \sigma_{\hat{M}}  . \nonumber
\end{align}

The equality in \eqref{eq:11fromeq20} follows from \cref{eq:2classical198}. The approximation in~\eqref{eq:11fromclaimnotequal} uses \cref{eq:2classical19845r4} and~\cref{eq:2classical1}. The approximation in~\eqref{eq:111fromclaimnotequal} uses \cref{eq:2classical19845} and~\cref{eq:2classical1}. Noting that
\[ p_{\sm} \quad ; \quad p_{} \quad ; \quad    \eta^{\mathcal{A}}_{M'}  \]depends only on the adversary $\mathcal{A} = (U,V,W, \ket{\psi}_{E_1E_2E_3})$ completes the proof. 
\end{proof}

Using \cref{thm:main}, we get the following corollary:

\begin{corollary}
    There exists constant rate $3$-split quantum non-malleable code with efficient encoding and decoding procedures with rate $\frac{1}{11+\delta}$ for arbitrarily small constant $\delta>0$. 
\end{corollary}

\section{A rate $1/3, 3$-split quantum secure non-malleable code}\label{sec:3qnmcclassical}
\begin{figure}
\centering
\resizebox{12cm}{6cm}{
\begin{tikzpicture}

\node at (1,6.5) {$\hat{A}$};
\node at (16.2,6.7) {$\hat{A}$};
\node at (16.2,4.7) {$O$};
\draw (1.2,6.5) -- (16,6.5);
\draw (1.1,6.2) -- (1.1, 4.8);
\draw (1.6,3.9) rectangle (3.7,5.4);
\draw (10.3,3.9) rectangle (14,5.4);
\node at (1,4.5) {$A$};

\node at (2.7,4.9) {$ T=$};
\node at (2.7,4.5) {$\mac(R,A)$};
\node at (4.7,4.8) {$Z=(A,T)$};
\node at (8.6,4.8) {$ Z'=(A',T') $};
\node at (12.2,5.1) {$\mathsf{Verify: 
}$};
\node at (12.2,4.6) {If $
\mac(R,A')=T',$ };
\node at (12.2,4.1) {Output $A'$ else $\perp$  };


\draw (1.2,4.5) -- (1.6,4.5);
\draw (3.7,4.5) -- (6,4.5);
\draw (7,4.5) -- (10.3,4.5);
\draw (14,4.5) -- (16,4.5);

\node at (0.8,0.5) {$R$};
\draw (1,0.5) -- (11,0.5);


\draw  (2.2,0.5) -- (2.2,3.9);

\draw [dashed] (1.4,-1.4) -- (1.4,7.2);
\draw [dashed] (15.4,-1.4) -- (15.4,7.2);

\node at (1.25,-1.4) {$\psi$};
\node at (15.2,-1.4) {$\rho$};

\node at (0.8,2.8) {$E$};

\draw (1.2,2.8) -- (6,2.8);
\draw (7,2.8) -- (8.5,2.8);

\draw (11,0.5) -- (11,3.9);

\node at (8.8,3.3) {$\hat{Z}'$};
\draw (7,3.3) -- (8.5,3.3);
\node at (8.8,2.8) {$E'$};
\draw (6,2.3) rectangle (7,5);
\node at (6.5,3.5) {$U$};

\node at (6.5,5.5) {$\mathcal{A}$};
\draw (5.7,1) rectangle (7.3,6);





\end{tikzpicture}}

\caption{Classical message authentication code.}\label{fig:splitstate211}
\end{figure}
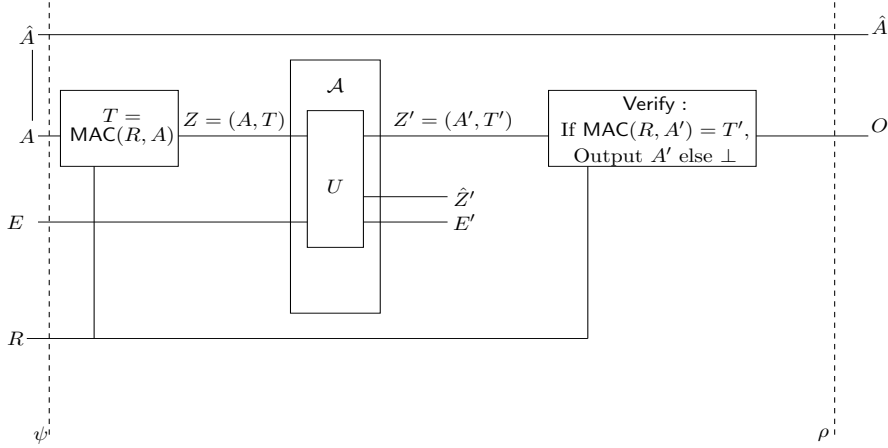

\begin{figure}[h]
\centering

\resizebox{12cm}{6cm}{

\begin{tikzpicture}

\node at (0,4.5) {$M$};
\node at (0,5.5) {$\hat{M}$};


\draw (0.3,5.5) -- (1.3,5.5);
\draw (0.3,4.5) -- (1.3,4.5);
\draw (0.3,0.2) -- (1.3,0.2);
\draw (0.3,2.8) -- (1.3,2.8);
\draw (1.2,1.8) -- (1.2,2.8);
\draw (1.2,0.2) -- (1.2,1.2);
\draw (1.25,5.5) -- (14.4,5.5);
\node at (0,2.8) {$Y$};
\draw (1.2,4.5) -- (4.8,4.5); 
\node at (5.8,3.8) {$E_3$};
\node at (12.8,4) {$E_3'$};
\draw (12.3,4.3) -- (12.8,4.3);

\draw (1.2,2.8) -- (4.5,2.8);
\draw (0,5.25) -- (0,4.7);
\node at (0,0.2) {$X$};
\draw (1.2,1.8) -- (4.5,1.8);
\draw [dashed] (4.5,1.8) -- (6.5,1.8);
\draw  (6.5,1.8) -- (8.5,1.8);
\draw (1.2,1.2) -- (4.5,1.2);
\draw [dashed] (4.5,1.2) -- (6.5,1.2);
\draw  (6.5,1.2) -- (8.5,1.2);
\draw (1.2,0.2) -- (4.5,0.2);
\draw (9.8,1.5) -- (10.5,1.5);
\node at (11.3,1.2) {${R}=R_{e}||R_{a}$};
\node at (11.3,0.65) {${R'}=R_{e}'||R_{a}'$};
\node at (11.8,6.7) {$V_R(M)=\left(R_e\oplus M,\mac(R_a,R_e\oplus M)\right)= (Z_1,Z_2)$.};
\node at (12.2,6.3) {$C_{R'}(Z')=\text{If $\mac(R'_a,Z'_1)=Z'_2$, output $Z_1'\oplus R'_e$ else $\perp$  }  $.};
\node at (10.5,4.5) {$V_{{R}}$};
\node at (11.1,4.7) {$Z$};
\node at (12.7,4.7) {$Z'$};
\draw  (10.5,1.5) -- (10.5,4.1);
\draw (10.1,4.1) rectangle (10.9,4.8);
\draw (4.6,4.5) -- (6.3,4.5);

\draw (13.1,4.1) rectangle (13.9,4.8);
\draw (13.5,4.1) -- (13.5,1.5);


\node at (14.7,4.7) {$M'$};
\node at (14.7,5.3) {$\hat{M}$};


\draw (8.5,1) rectangle (9.8,2);
\node at (9.15,1.5) {$2\nmext$};

\draw (11.3,4) rectangle (12.3,5);
\node at (11.8,4.5) {$W$};

\draw [dashed] (3.78,-0.8) -- (3.78,6);
\draw [dashed] (9.9,-0.8) -- (9.9,6);
\draw [dashed] (8.4,-0.8) -- (8.4,6);
\draw [dashed] (14.3,-0.8) -- (14.3,6);

\node at (3.64,-0.8) {$\sigma$};
\node at (8.24,-0.8) {$\tilde{\tau}$};
\node at (9.74,-0.8) {${\tau}$};
\node at (14.13,-0.8) {$\rho$};

\node at (4.6,3) {$Y$};
\node at (7.8,3) {$Y'$};

\draw (4.5,2.8) -- (6.3,2.8);
\draw (7.3,2.8) -- (8.2,2.8);

\draw  (8.2,0.6) -- (8.2,0.9);
\draw [dashed] (8.2,0.9) -- (8.2,2.8);
\draw  (8.2,2.5) -- (8.2,2.8);
\draw (8.2,0.6) -- (8.5,0.6);

\node at (4.6,0.4) {$X$};
\node at (7.8,0.0) {$X'$};

\draw (4.5,0.2) -- (6.3,0.2);
\draw (7.3,0.2) -- (8.5,0.2);

\draw (6.3,2) rectangle (7.3,3.5);
\node at (6.8,2.8) {$V$};
\draw (6.3,0) rectangle (7.3,1);
\node at (6.8,0.5) {$U$};

\node at (6.7,-0.4) {$\mathcal{A}=(U,V,W,\psi)$};
\draw (5.2,-0.8) rectangle (8.1,5.3);

\draw (5.8,2.4) ellipse (0.3cm and 2cm);
\node at (5.8,2) {$E_2$};
\draw (6.1,2.2) -- (6.3,2.2);
\draw (5.9,4.3) -- (9.8,4.3);
\draw [dashed] (9.8,4.3) -- (11.1,4.3);
\draw (11.1,4.3) -- (11.3,4.3);
\draw (4.8,4.5) -- (10.1,4.5);
\draw (10.9,4.5) -- (11.3,4.5);
\draw (12.3,4.5) -- (13.1,4.5);
\draw (13.9,4.5) -- (14.5,4.5);
\node at (7.8,2) {$E_2'$};
\draw (7.3,2.2) -- (7.5,2.2);
\node at (5.8,1) {$E_1$};
\draw (5.93,0.7) -- (6.3,0.7);

\node at (7.8,2.5) {$\hat{Y}'$};

\node at (7.7,0.44) {$\hat{X}'$};

\draw (7.3,0.44) -- (7.5,0.44);
\draw (7.3,2.5) -- (7.5,2.5);
\node at (7.7,0.9) {$E_1'$};
\draw (7.3,0.9) -- (7.5,0.9);


\draw (8.5,0) rectangle (9.8,0.9);
\draw (9.8,0.4) -- (13.5,0.4);
\draw (13.5,0.4) -- (13.5,1.5);
\node at (9.2,0.5) {$2\nmext$};
\node at (13.5,4.5) {$C_{R'}$};

\end{tikzpicture} }
\caption{Rate $1/3, ~3$-split quantum secure non-malleable code }\label{fig:splitstate4'}
\end{figure}
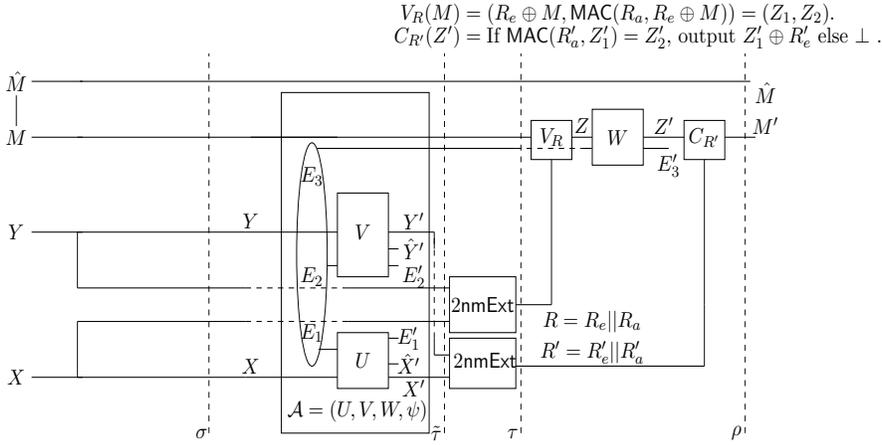

Let $r, m, t, n, k$ be positive integers and $\eps >0$.
\begin{definition}\label{def:mac}
	A function $\mac:\{0,1\}^{r} \times\{0,1\}^m \to \{0,1\}^t$ is an \emph{$\eps$-information-theoretically secure one-time message authentication code} if for any function $\mathcal{A}:\{0,1\}^m \times \{0,1\}^t \to \{0,1\}^m\times \{0,1\}^t$ it holds that for all $\mu \in \{0,1\}^m$
	$$\Pr_{k\leftarrow \{0,1\}^{r}}\big[ (\mac(k,\mu') = \sigma' ) \, \wedge \, (\mu'\neq \mu) : (\mu',\sigma') \leftarrow \mathcal{A}(\mu,\mac(k,\mu))\big] \,\leq\,\eps.$$

\end{definition}
Efficient constructions of $\mac$ satisfying the conditions of \cref{def:mac} are known. The following fact summarizes some parameters that are achievable using a construction based on polynomial evaluation.
\begin{fact}[Proposition 1 in~\cite{KR09}]\label{prop:mac}
	For any integer $m > 0$ and $\eps_{\mac}>0$, there exists an efficient family of  $\eps$-information-theoretically secure one-time message authentication codes $$\mac:\{0,1\}^{r} \times\{0,1\}^m \to \{0,1\}^t,$$ for parameters $t= \log(m) + \log(1/\eps_\mac)$ and $r=2t$.
\end{fact}

\begin{fact}[\cite{KR09}]\label{prop:mac2}
	Let $m$ be a natural number and $\epsilon = 2^{-m}$. Let $K$ be uniformly distributed in  $\{0,1 \}^{2m}$. There exists an efficient family  $P:\{0,1\}^{2m} \times\{0,1\}^m \to \{0,1\}^m,$ with the following properties. 
 \begin{enumerate}
     \item For all $k$: $P_k :\{0,1\}^m \to \{0,1\}^m$ is a bijective function, where $P_k(\mu) \defeq P(k,\mu)$.
     \item For any $\mu \in \{0,1 \}^{m}:$  $P(K,\mu) = U_m$.
 \item The collection $\{P(K,\mu) ~| ~ \mu \in \{0,1 \}^{m}\}$ is $\epsilon$-pairwise independent.

 \end{enumerate}

\end{fact}
\begin{lemma}[\cite{BJK21}]\label{lem:equal10112}
 Consider \cref{fig:splitstate211}. Let state ${\psi}_{A \hat{A}RE}$ be a c-q state (with registers $A\hat{A}R$ classical) such that 
 \[ {\psi}_{A \hat{A}RE} ={\psi}_{A \hat{A}} \otimes U_R \otimes {\psi}_{E}  \quad ; \quad  \Pr(\hat{A}=A)_\psi =1. \]Let  $U:(\cH_Z\otimes \cH_{E})\to (\cH_{Z'}\otimes \cH_{\hat{Z}'} \otimes \cH_{E'})$ be any unitary. Then,
 \[ \rho_{\hat{A}O}   \approx_{\eps}  p \psi_{\hat{A}A} + (1- p)  (\psi_{\hat{A}} \otimes \eta^{}_{O}), \]
 where $p$ depends only on unitary $U $ and state $\psi_E$ and $\eta_O$ is $\perp$.
\end{lemma}
\begin{proof}
    The proof follows using the arguments of Claims~$12$ and $13$ of~\cite{BJK21}. 
\end{proof}

\begin{theorem}\label{thm:m1}
Consider \cref{fig:splitstate4'}.  Let $2\nmext$ in \cref{fig:splitstate4'} be from \cref{lem:qnmcodesfromnmext} for $k=O(n^{1/4})$ and $\epsilon=2^{-n^{\Omega(1)}}$. Let $\sigma_{M\hat{M} XY}$ be a state (all registers are classical) such that 
\[\sigma_{M\hat{M} XY} =\sigma_{M\hat{M} } \otimes U_n \otimes U_{\delta n} \quad ; \quad \Pr(M=\hat{M})_\sigma=1 \]
Let $R=2\nmext(X,Y)$. Let $ \vert M\vert= ({\frac{1}{2}-1.1\delta})n$ and $\mac$ be the one-time message-authentication code from \cref{prop:mac} for parameters $\eps_\mac = 2^{-\Omega(n)}$. Let $R=R_e \vert \vert R_a$, $\vert R_e \vert = \vert M\vert$ and $\vert R_a \vert =0.1 \delta n$.

Let $\enc : \cL( \cH_M) \to \cL(\cH_Z \otimes \cH_Y \otimes \cH_X)$ be the classical encoding map and $\dec  : \cL(\cH_{Z'} \otimes \cH_{Y'}\otimes \cH_{X'}) \to  \cL(\cH_{M'})$ be the classical decoding map. The procedure for $\enc$ on input $(\sigma_{M})$ is as follows:
\begin{itemize}
    \item Sample ${XY} = U_n \otimes U_{\delta n}$ (independent of $\sigma_M$)
    \item Compute $R =R_e \vert \vert R_a= 2 \nmext(X,Y)$
    \item Set $Z =  (R_e \oplus M) \vert \vert \mac(R_a,(R_e \oplus M) )$
\end{itemize}
The procedure for $\dec$ on input $(Z',Y',X')$ is as follows:
\begin{itemize}
    \item Compute $R' =R'_e \vert\vert R'_a= 2 \nmext(X',Y')$
    \item Check if $\mac(R_a', Z_1') =Z_2'$, output $Z_1' \oplus R_e'$, else output $\perp$.
\end{itemize}Note that all parts $(X,Y,Z)$ of the ciphertext are classical. Then, $(\enc,\dec)$ as specified above is an $\eps'$-$3$-split non-malleable code for $\sigma_{\hat{M}M}$, where the registers $(Z,Y,X)$ correspond to $3$ parts of the codeword, and $\eps' =  8(2^{-k}+ \eps)+\eps_\mac$.
\end{theorem}
\begin{proof}
To show that $(\enc,\dec)$ is an $\eps'$-$3$-split non-malleable code for classical messages, it suffices to show that for every $\mathcal{A}=(U,V,W,\psi)$ it holds that (in \cref{fig:splitstate4'})
\begin{equation}\label{eq:finalgoal1'}
    \rho_{\hat{M}M'} \approx_{\eps'} p_{\mathcal{A}} \sigma_{\hat{M}M}  + (1-p_\mathcal{A}) (\sigma_{\hat{M}} \otimes  \gamma^{\mathcal{A}}_{M'}),
\end{equation}where $(p_{\mathcal{A}}, \gamma^{\mathcal{A}}_{M'})$ depend only on the $3$-split adversary $\cA$. 

Consider the state $\tilde{\tau}$ in~\cref{fig:splitstate4'}.
Using arguments as in the proof of \cref{thm:main} using \cref{lem:qnmcodesfromnmext}, we have the following:
\begin{equation}\label{eq:2classical198'}
     (\tau)_{RR'E_2'E_3M\hat{M}}= p_{\sm} (\tau^1)_{RR'E_2'E_3M\hat{M}} + (1-p_{\sm}) (\tau^0)_{RR'E_2'E_3M\hat{M}}
    \end{equation}
    and
\begin{multline}\label{eq:2classical1'}
        p_{\sm} \Vert (\tau^1)_{RE_2'E_3M\hat{M}} -  U_{\vert R \vert} \otimes (\tau^1)_{E_2'E_3M\hat{M}} \Vert_1 + \\ (1-p_{\sm})\Vert (\tau^0)_{RR'E_2'E_3M\hat{M}} -  U_{\vert R \vert} \otimes (\tau^0)_{R'E_2'E_3M\hat{M}} ) \Vert_1  \leq 4(2^{-k}+\eps).
    \end{multline}Further, we have
    \begin{equation}\label{Eq:file123'}
         \Pr(R =R')_{\tau^1}=1.
    \end{equation}
 Let $\Phi$ be a CPTP map from registers $RR'E_3M\hat{M}$ to $M'\hat{M}$ (i.e. $\Phi$ maps state $\tau$ to $\rho$) in~\cref{fig:splitstate4'}. One can note, 
\begin{equation}\label{eq:2classical19845'}
    \Phi(U_{\vert R \vert} \otimes \tau^0_{R'E_3M\hat{M}}) = \eta^{\mathcal{A}}_{M'} \otimes \sigma_{\hat{M}}
    \end{equation} follows using encryption property of one-time pad followed by~\cref{fact:data}. Further state $\eta^{\mathcal{A}}_{M'}$ depends only on adversary $\mathcal{A}$. Let $\tilde{ \tau}^1$ be the state such that 
\[ \tilde{ \tau}^1_{RR'E_3M\hat{M}} =\tilde{ \tau}^1_{RR'} \otimes \tau^0_{E_3M\hat{M}} \quad ; \quad   \Pr(R =R')_{\tilde{ \tau}^1}=1 \quad ; \quad \tilde{ \tau}^1_{R'}=U_{\vert R\vert}. \]

Note, \begin{equation}\label{eq:2classical19845r4'}
   \Phi( \tilde{ \tau}^1_{ RR'E_3M\hat{M}})  \approx_{\eps_\mac}  p \sigma_{M\hat{M}} + (1- p)  (\perp \otimes \sigma_{\hat{M}} ), 
    \end{equation} follows from \cref{lem:equal10112}. From here on using arguments as in the proof of \cref{thm:main} after \cref{eq:2classical19845r4}, we  get, 
\begin{align}
  \rho_{M' \hat{M}} 
     & \approx_{ 8(2^{-k}+ \eps) + 2\eps_\mac} p_{\sm}\cdot p\cdot  \sigma_{M\hat{M}}  \nonumber\\
     & \quad\quad + 
 \left(p_{\sm}  (1-p) \perp  + (1-p_{\sm})  \eta^{\mathcal{A}}_{M'}  \right) \otimes \sigma_{\hat{M}} , \nonumber
\end{align}Noting that
\[ p_{\sm} \quad ; \quad p_{} \quad ; \quad    \eta^{\mathcal{A}}_{M'}  \]depends only on the adversary $\mathcal{A} = (U,V,W, \ket{\psi}_{E_1E_2E_3})$ completes the proof.

\end{proof}

Using \cref{thm:m1}, we get the following corollary:

\begin{corollary}
    There exists constant rate $3$-split quantum secure non-malleable code with efficient classical encoding and decoding procedures with rate $\frac{1}{3+\delta}$ for arbitrarily tiny constant $\delta>0$. 
\end{corollary}

\section{A rate $1/5$, $2$-split, quantum secure, average case,  non-malleable code}\label{sec:2qnmcclassical}

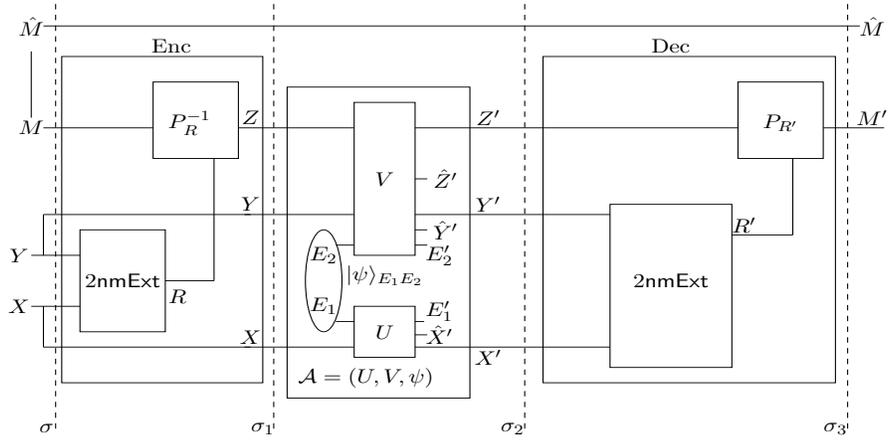
\begin{figure}
\centering
\resizebox{12cm}{6cm}{
\begin{tikzpicture}

\draw (1,4.7) -- (1,6);
\draw (3,3.9) rectangle (4.4,5.4);
\draw (12.6,3.9) rectangle (14,5.4);
\node at (1,4.5) {$M$};

\node at (3.6,4.6) {$P^{-1}_{R}$};
\node at (13.3,4.6) {$P_{R'}$};
\node at (1,6.5) {$\hat{M}$};
\node at (14.8,6.5) {$\hat{M}$};

\draw (1.2,6.5) -- (14.6,6.5);

\draw (1.2,4.5) -- (3,4.5);
\draw (4.4,4.5) -- (6.3,4.5);
\draw (7.3,4.5) -- (12.6,4.5);
\draw (14,4.5) -- (15,4.5);
\node at (14.8,4.7) {$M'$};
\node at (0.8,2) {$Y$};
\node at (0.8,1) {$X$};
\node at (3.4,1.2) {$R$};
\draw (1,1) -- (1.8,1);
\draw (1,2) -- (1.8,2);
\draw (3.2,1.5) -- (4,1.5);
\draw (4,1.5) -- (4,3.9);
\draw (1.2,0.2) -- (4.6,0.2);
\draw (1.2,2.8) -- (4.6,2.8);
\draw (1.2,0.2) -- (1.2,1);
\draw (1.2,2) -- (1.2,2.8);
\draw (1.8,0.5) rectangle (3.2,2.5);
\draw (1.5,-0.5) rectangle (4.8,5.9);
\node at (3.3,6.1) {$\enc$};
\node at (2.5,1.5) {$2\nmext$};

\draw (9.4,-0.5) rectangle (14.2,5.9);
\node at (11.5,6.1) {$\dec$};


\draw [dashed] (1.4,-1.4) -- (1.4,7.2);
\draw [dashed] (4.98,-1.4) -- (4.98,7.2);
\draw [dashed] (9.1,-1.4) -- (9.1,7.2);
\draw [dashed] (14.4,-1.4) -- (14.4,7.2);


\node at (6.8,1.6) {$\ket{\psi}_{E_1E_2}$};
\node at (1.25,-1.4) {$\sigma$};
\node at (4.8,-1.4) {$\sigma_1$};
\node at (8.9,-1.4) {$\sigma_2$};
\node at (14.2,-1.4) {$\sigma_3$};

\node at (4.6,4.7) {$Z$};
\node at (8.5,4.7) {$Z'$};
\node at (4.6,3) {$Y$};
\node at (8.5,3) {$Y'$};
\draw (4.5,2.8) -- (6.3,2.8);
\draw (7.3,2.8) -- (10.5,2.8);

\node at (12.7,2.6) {$R'$};
\draw (12.5,2.4) -- (13.5,2.4);
\draw (13.5,2.4) -- (13.5,3.9);

\node at (4.6,0.4) {$X$};
\node at (8.5,0.0) {$X'$};
\draw (4.5,0.2) -- (6.3,0.2);
\draw (7.3,0.2) -- (10.5,0.2);

\draw (6.3,2) rectangle (7.3,5);
\node at (6.8,3.5) {$V$};
\draw (6.3,0) rectangle (7.3,1);
\node at (6.8,0.5) {$U$};

\node at (6.5,-0.4) {$\mathcal{A}=(U,V,\psi)$};
\draw (5.2,-0.8) rectangle (8.2,5.3);

\draw (5.8,1.5) ellipse (0.3cm and 1cm);
\node at (5.8,2) {$E_2$};
\draw (6,2.2) -- (6.3,2.2);
\node at (7.7,2) {$E'_2$};
\draw (7.3,2.2) -- (7.5,2.2);
\node at (5.8,1) {$E_1$};
\draw (6,0.7) -- (6.3,0.7);
\node at (7.7,0.9) {$E'_1$};
\draw (7.3,0.7) -- (7.45,0.7);



\draw (10.5,-0.2) rectangle (12.5,3.0);
\node at (11.5,1.5) {${2\nmext}$};
\node at (7.8,2.5) {$\hat{Y}'$};

\node at (7.8,3.5) {$\hat{Z}'$};
\draw (7.3,3.5) -- (7.5,3.5);

\node at (7.7,0.44) {$\hat{X}'$};

\draw (7.3,0.44) -- (7.5,0.44);
\draw (7.3,2.5) -- (7.5,2.5);
\end{tikzpicture}}

\caption{Rate $1/5, 2$-split, quantum secure, average case, non-malleable code.}\label{fig:splitstate22}
\end{figure}

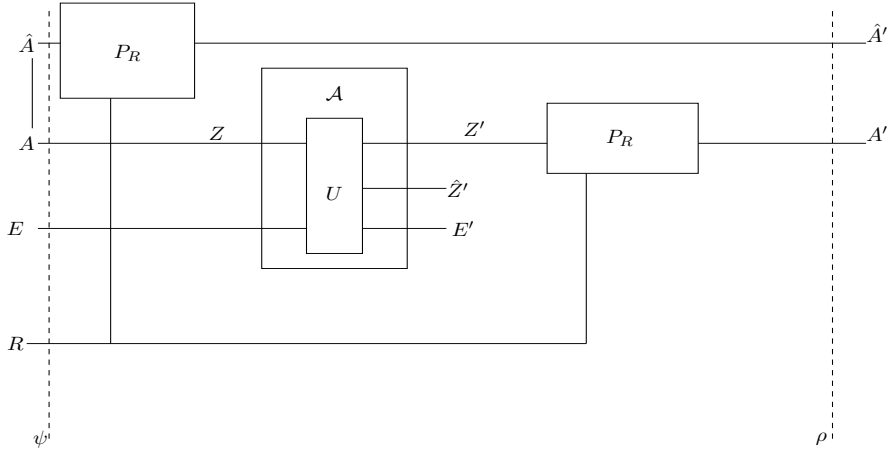
\begin{figure}
\centering
\resizebox{12cm}{6cm}{
\begin{tikzpicture}

\node at (1,6.5) {$\hat{A}$};
\node at (16.2,6.7) {$\hat{A}'$};
\node at (16.2,4.7) {$A'$};
\draw (1.2,6.5) -- (1.6,6.5);
\draw (4,6.5) -- (16,6.5);
\draw (1.1,6.2) -- (1.1, 4.8);
\draw (1.6,5.4) rectangle (4,7.3);
\draw (10.3,3.9) rectangle (13,5.3);
\node at (1,4.5) {$A$};

\node at (2.8,6.3) {$P_R$};
\node at (4.4,4.7) {$ Z $};
\node at (9,4.8) {$ Z' $};
\node at (11.6,4.6) {$P_R$ };


\draw (1.2,4.5) -- (6,4.5);
\draw (7,4.5) -- (10.3,4.5);
\draw (13,4.5) -- (16,4.5);

\node at (0.8,0.5) {$R$};
\draw (1,0.5) -- (11,0.5);


\draw  (2.5,0.5) -- (2.5,5.4);

\draw [dashed] (1.4,-1.4) -- (1.4,7.2);
\draw [dashed] (15.4,-1.4) -- (15.4,7.2);

\node at (1.25,-1.4) {$\psi$};
\node at (15.2,-1.4) {$\rho$};

\node at (0.8,2.8) {$E$};
\node at (8.7,3.6) {$\hat{Z}'$};
\draw (7,3.6) -- (8.5,3.6);

\draw (1.2,2.8) -- (6,2.8);
\draw (7,2.8) -- (8.5,2.8);

\draw (11,0.5) -- (11,3.9);


\draw (6,2.3) rectangle (7,5);
\node at (6.5,3.5) {$U$};

\node at (6.5,5.5) {$\mathcal{A}$};
\draw (5.2,2) rectangle (7.8,6);




\node at (8.8,2.8) {$E'$};
\end{tikzpicture}}

\caption{Classical non-malleable code with shared key.}\label{fig:splitstate2113}
\end{figure}
\begin{lemma}[]\label{lem:equal101125''}
 Consider \cref{fig:splitstate2113}. Let $P_R$  be $P(R,\cdot)$ from  \cref{prop:mac2} such that $P_R(M) = P(R,M)$ for $\eps_P= 2^{-\vert M\vert}$. Let $P_R^{-1}(\cdot)$ be the inverse function of $P_R^{}(\cdot)$. Let $\sigma_{\hat{A}A}$ (with both registers classical) be a state such that 
 \[ \sigma_{A} =U_A \quad ; \quad \Pr(\hat{A}=A)_\sigma =1.\]
 Let state ${\psi}_{A \hat{A}RE}$ be a c-q state (with registers $A\hat{A}R$ classical) such that 
 \[ {\psi}_{A \hat{A}RE} ={\psi}_{A \hat{A}} \otimes U_R \otimes {\psi}_{E}  \quad ; \quad  \Pr(\hat{A}=A)_\psi =1 \quad ; \quad \vert R \vert =2 \vert A\vert \quad ; \quad \vert A \vert =m. \]Let  $U:(\cH_Z\otimes \cH_{E})\to (\cH_{Z'}\otimes \cH_{\hat{Z}'}\otimes \cH_{E'})$ be a unitary. Then,
 \[ \rho_{\hat{A}'A'}   \approx_{\frac{2}{2^m-1}}  p \sigma_{\hat{A}A} + (1- p)  (\sigma_{\hat{A}} \otimes U_{A'}), \]
 where  $p$ depends only on unitary $U $ and state $\psi_E$.
\end{lemma}
\begin{proof}
Note 
${\psi}_{A \hat{A}}$ is independent of other registers in ${\psi}_{A \hat{A}RE}$. Let $\tilde{\psi} = \tr_{E'\hat{Z}'} \left(U  (\psi_{\hat{A} A} \otimes \psi_E) U^\dagger \right).$ Let,
\begin{equation}\label{eq:7thm1}
   \psi^0= \tilde{\psi} \vert  (\hat{A} = Z' );\quad \psi^1= \tilde{\psi}\vert  (\hat{A} \not= Z')   ;\quad  p= \Pr( \hat{A}=Z')_{\tilde{\psi}}.
\end{equation}
Let $\Phi$ be the CPTP map from registers $\hat{A}Z'R$ to $\hat{A}'A'$ (i.e. $\Phi$ applies $P_R$ on registers $Z',\hat{A}$). Then 
 \begin{align*}
 \rho_{\hat{A}'A'}&=p\Phi(\psi^0)+(1-p)\Phi(\psi^1)\\
 &=p\sigma_{\hat{A}A}+(1-p)\Phi(\psi^1)\\
 &\approx_\frac{2}{2^m-1} p\sigma_{\hat{A}A}+(1-p)(\sigma_{\hat{A}}\otimes U_{A'}).
   \end{align*}
   Last approximation follows because of $\epsilon$-pairwise independence property of the family $P$ from \cref{prop:mac2}.
\end{proof}

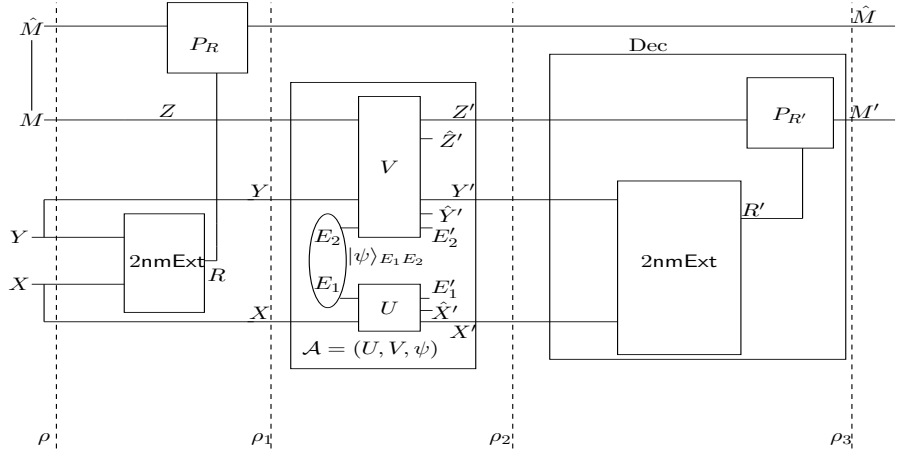
\begin{figure}
\centering
\resizebox{12cm}{6cm}{
\begin{tikzpicture}

\draw (1,4.7) -- (1,6.2);
\draw (3.2,5.5) rectangle (4.5,7);
\draw (12.6,3.9) rectangle (14,5.4);
\node at (1,4.5) {$M$};
\node at (3.8,6.1) {$P_R$};
\node at (13.3,4.6) {$P_{R'}$};
\node at (1,6.5) {$\hat{M}$};
\node at (14.5,6.75) {$\hat{M}$};
\draw (1.2,6.5) -- (3.2,6.5);
\draw (4.5,6.5) -- (15,6.5);

\draw (2.5,0.4) rectangle (3.8,2.5);
\node at (0.8,2) {$Y$};
\node at (0.8,1) {$X$};
\node at (4,1.2) {$R$};
\draw (1,1) -- (2.5,1);
\draw (1,2) -- (2.5,2);
\draw (3.8,1.5) -- (4,1.5);
\draw (4,1.5) -- (4,1.8);
\draw  (4,1.8) -- (4,2.3);
\draw (4,2.3) -- (4,5.5);

\draw (1.2,0.2) -- (4.6,0.2);
\draw (1.2,2.8) -- (4.6,2.8);
\draw (1.2,0.2) -- (1.2,1);
\draw (1.2,2) -- (1.2,2.8);

\draw (1.2,4.5) -- (1.7,4.5);
\draw (1.7,4.5) -- (2.5,4.5);
\draw (2.5,4.5) -- (6.3,4.5);
\draw (7.3,4.5) -- (12.6,4.5);
\draw (14,4.5) -- (15,4.5);
\node at (14.5,4.7) {$M'$};
\node at (3.2,1.5) {$2\nmext$};


\draw [dashed] (1.4,-2.5) -- (1.4,7);
\draw [dashed] (4.88,-2.5) -- (4.88,7);
\draw [dashed] (8.8,-2.5) -- (8.8,7);
\draw [dashed] (14.3,-2.5) -- (14.3,7);

\node at (6.75,1.6) {$\ket{\psi}_{E_1E_2}$};
\node at (1.2,-2.3) {$\rho$};
\node at (4.74,-2.3) {$\rho_1$};
\node at (8.6,-2.3) {$\rho_2$};
\node at (14.13,-2.3) {$\rho_3$};

\node at (3.2,4.7) {$Z$};
\node at (8,4.7) {$Z'$};
\node at (4.7,3) {$Y$};
\node at (8,3) {$Y'$};
\draw (4.5,2.8) -- (6.3,2.8);
\draw (7.3,2.8) -- (10.5,2.8);

\node at (12.7,2.6) {$R'$};
\draw (12.5,2.4) -- (13.5,2.4);
\draw (13.5,2.4) -- (13.5,3.9);

\node at (4.7,0.4) {$X$};
\node at (8,0.0) {$X'$};
\draw (4.5,0.2) -- (6.3,0.2);
\draw (7.3,0.2) -- (10.5,0.2);


\draw (6.3,2) rectangle (7.3,5);
\node at (6.8,3.5) {$V$};
\draw (6.3,0) rectangle (7.3,1);
\node at (6.8,0.5) {$U$};

\node at (6.5,-0.4) {$\mathcal{A}=(U,V,\psi)$};
\draw (5.2,-0.8) rectangle (8.2,5.3);

\draw (5.8,1.5) ellipse (0.3cm and 1cm);
\node at (5.8,2) {$E_2$};
\draw (6,2.2) -- (6.3,2.2);
\node at (7.7,2) {$E'_2$};
\draw (7.3,2.2) -- (7.5,2.2);
\node at (5.8,1) {$E_1$};
\draw (6,0.7) -- (6.3,0.7);
\node at (7.7,0.9) {$E'_1$};
\draw (7.3,0.7) -- (7.45,0.7);



\draw (10.5,-0.5) rectangle (12.5,3.2);
\node at (11.5,1.5) {${2\nmext}$};



\draw (9.4,-0.6) rectangle (14.2,5.9);
\node at (11,6.1) {$\dec$};
\node at (7.8,2.5) {$\hat{Y}'$};

\node at (7.7,0.44) {$\hat{X}'$};
\node at (7.8,4.1) {$\hat{Z}'$};
\draw (7.3,0.44) -- (7.5,0.44);
\draw (7.3,2.5) -- (7.5,2.5);
\draw (7.3,4.1) -- (7.5,4.1);
\end{tikzpicture}}
\caption{Rate $1/5, ~2$-split, quantum secure, average case, non-malleable code in the modified process.}\label{fig:splitstate33}
\end{figure}

\begin{figure}
\centering
\resizebox{12cm}{6cm}{
\begin{tikzpicture}




\draw (1,4.7) -- (1,6);
\draw (14.6,5.7) rectangle (16,7);
\draw (12.9,3.9) rectangle (14.3,5.4);
\node at (1,4.5) {$M$};
\node at (15.3,6.1) {$P_{R}^{}$};
\node at (13.5,4.6) {$P_{R'}^{}$};
\node at (1,6.5) {$\hat{M}$};
\draw (1.2,6.5) -- (14.6,6.5);
\draw (16,6.5) -- (16.5,6.5);

\draw (1.2,4.5) -- (1.7,4.5);
\draw (1.7,4.5) -- (2.5,4.5);
\draw (2.5,4.5) -- (6,4.5);
\draw (7,4.5) -- (12.9,4.5);
\draw (14.3,4.5) -- (16,4.5);
\node at (16.5,4.7) {$M'$};
\node at (16.5,6.2) {$\hat{M}$};
\draw (4.5,-1) -- (10.5,-1);
\draw (4.5,-1.6) -- (10.5,-1.6);
\draw (15.5,-1.6) -- (12.4,-1.6);

\draw  (15.5,-1.6) -- (15.5,4.1);
\draw [dashed] (15.5,4.1) -- (15.5,4.7);
\draw  (15.5,4.7) -- (15.5,5.7);

\draw [dashed] (1.4,-2.5) -- (1.4,7);
\draw [dashed] (4.88,-2.5) -- (4.88,7);
\draw [dashed] (9,-2.5) -- (9,7);
\draw [dashed] (12.8,-2.5) -- (12.8,7);
\draw [dashed] (16.1,-2.5) -- (16.1,7);

\node at (6.45,1.6) {$\ket{\psi}_{E_1E_2}$};
\node at (1.2,-2.3) {$\theta$};
\node at (4.74,-2.3) {$\theta_1$};
\node at (8.8,-2.3) {$\theta_2$};
\node at (13,-2.3) {$\theta_3$};
\node at (15.9,-2.3) {$\theta_4$};

\node at (3.2,4.7) {$Z$};
\node at (8,4.7) {$Z'$};
\node at (0.8,3) {$Y$};
\draw (4.5,2.8) -- (1,2.8);
\draw (4.5,2.8) -- (4.5,0.5);
\draw [dashed] (4.5,0.5) -- (4.5,-0.5);
\draw (4.5,-0.5) -- (4.5,-1);

\draw (1,0.2) -- (4.5,0.2);
\draw (3.5,0.2) -- (3.5,-1.6);
\draw (3.5,-1.6) -- (4.5,-1.6);
\node at (8,3) {$Y'$};
\draw (4.5,2.8) -- (6,2.8);
\draw (7,2.8) -- (10.5,2.8);

\node at (12.6,2.6) {$R'$};
\draw (12.4,2.4) -- (13.5,2.4);
\draw (13.5,2.4) -- (13.5,3.9);

\node at (0.8,0.4) {$X$};
\node at (8,0.0) {$X'$};
\draw (4.5,0.2) -- (6,0.2);
\draw (7,0.2) -- (10.5,0.2);

\draw (6,2) rectangle (7,5);
\node at (6.5,3.5) {$V$};
\draw (6,0) rectangle (7,1);
\node at (6.5,0.5) {$U$};

\node at (6.5,-0.4) {$\mathcal{A}=(U,V,\psi)$};
\draw (5.2,-0.8) rectangle (7.8,5.3);

\draw (5.5,1.5) ellipse (0.25cm and 1.1cm);
\node at (5.5,2) {$E_2$};
\draw (5.7,2.2) -- (6,2.2);

\node at (7.4,2) {$E'_2$};
\draw (7,2.2) -- (8.5,2.2);
\draw (8.5,2.2) -- (8.5,6.3);
\draw (8.5,6.3) -- (13.8,6.3);
\node at (14.1,6.3) {$E'_2$};

\node at (5.5,1) {$E_1$};
\draw (5.7,0.7) -- (6,0.7);
\node at (7.4,0.9) {$E'_1$};
\draw (7.0,0.7) -- (7.2,0.7);


\draw (10.5,-2) rectangle (12.4,-0.8);
\draw (10.5,-0.4) rectangle (12.4,3.1);
\node at (11.5,1.5) {${2\nmext}$};
\node at (11.5,-1.4) {${2\nmext}$};

\node at (14.3,-1.4) {$R$};

\draw (9.4,-0.6) rectangle (14.5,5.9);
\node at (11,6.1) {$\dec$};
\node at (7.4,2.5) {$\hat{Y}'$};

\node at (7.4,3.5) {$\hat{Z}'$};
\draw (7.2,3.5) -- (7,3.5);
\node at (7.4,0.44) {$\hat{X}'$};

\draw (7.2,0.44) -- (7,0.44);
\draw (7.2,2.5) -- (7,2.5);
\end{tikzpicture}}
\caption{Rate $1/5, ~2$-split, quantum secure, average case, non-malleable code in the modified process.}\label{fig:splitstate48}
\end{figure}
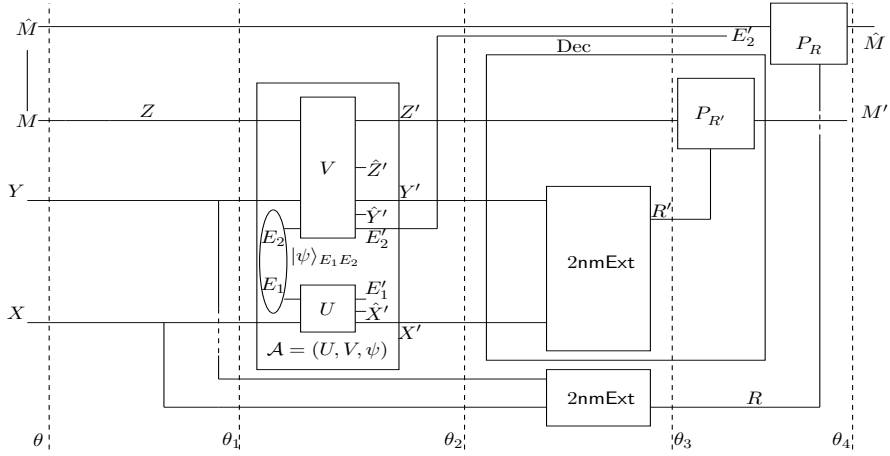

\begin{theorem}\label{thm:m''}
Consider \cref{fig:splitstate22}.  Let $2\nmext$ in \cref{fig:splitstate22} be from \cref{lem:qnmcodesfromnmext} for $k=O(n^{1/4})$ and $\epsilon=2^{-n^{\Omega(1)}}$. Let $P_R$ in \cref{fig:splitstate22} be $P(R,\cdot)$ from  \cref{prop:mac2} such that $P_R(M) = P(R,M)$ for $\eps_P= 2^{-\vert M\vert}$. Note for every $r$, $P^{}(r,\cdot)$ is a bijective function. Let $P_R^{-1}(\cdot)$ be the inverse function of $P_R^{}(\cdot)$. Let $\sigma_{M\hat{M} XY}$ be a state (all registers are classical) such that 
\[\sigma_{M\hat{M} XY} =\sigma_{M\hat{M} } \otimes U_n \otimes U_{\delta n} \quad ; \quad \Pr(M=\hat{M})_\sigma=1 \quad ;\quad \sigma_M=U_{M} .\]
Let $R=2\nmext(X,Y)$. Let $ \vert M\vert=m= ({\frac{1}{2}-\delta})n$.

Let $\enc : \cL( \cH_M) \to \cL(\cH_Z \otimes \cH_Y \otimes \cH_X)$ be the classical encoding map and $\dec  : \cL(\cH_{Z'} \otimes \cH_{Y'}\otimes \cH_{X'}) \to  \cL(\cH_{M'})$ be the classical decoding maps as shown in \cref{fig:splitstate22}. Note that all parts $(X,Y,Z)$ of the ciphertext are classical. Then, $(\enc,\dec)$ as specified above is an $\eps'$-$2$-split non-malleable code for $\sigma_{\hat{M}M}$, where the registers $((Z,Y),X)$ correspond to $2$ parts of the codeword, and $\eps' =  8(2^{-k}+ \eps)+3\eps_P$.

\end{theorem}
\begin{proof}

First we note that since $\rho_M=\sigma_M=U_{\vert M\vert}$ and $P_r(\cdot)$ is a bijective function, the states $\rho_1, \rho_2, \rho_3$ in \cref{fig:splitstate33}
and states $\sigma_1, \sigma_2, \sigma_3$ in \cref{fig:splitstate22} are the same. Further, note \cref{fig:splitstate33} and \cref{fig:splitstate48} are equivalent except for the delayed generation of register $R$ and its controlled operation on register $\hat{M}$ in \cref{fig:splitstate48}. Note this can be done as the adversary operations and generation of register $R$ commute. 

Thus, to show that $(\enc,\dec)$ is an $\eps'$-$2$-split non-malleable code for average classical messages, it suffices to show that for every $\mathcal{A}=(U,V,\psi)$ it holds that (in \cref{fig:splitstate48})
\begin{equation}\label{eq:finalgoal''}
    (\theta_4)_{\hat{M}M'} \approx_{\eps'} p_{\mathcal{A}} \sigma_{\hat{M}M}  + (1-p_\mathcal{A}) (\sigma_{\hat{M}} \otimes  \gamma^{\mathcal{A}}_{M'}),
\end{equation}where $(p_{\mathcal{A}}, \gamma^{\mathcal{A}}_{M'})$ depend only on the $2$-split adversary $\cA$. 

Consider the state $\theta_3$ in~\cref{fig:splitstate48}. Using arguments as in the proof of \cref{thm:main} using \cref{lem:qnmcodesfromnmext}, we have the following:
\begin{equation}\label{eq:2classical198''}
     (\theta_3)_{RR'E_2'Z'\hat{M}}= p_{\sm} (\theta_3^1)_{RR'E_2'Z'\hat{M}} + (1-p_{\sm}) (\theta_3^0)_{RR'E_2'Z'\hat{M}}
    \end{equation}
    and
\begin{multline}\label{eq:2classical1''}
        p_{\sm} \Vert (\theta^1_3)_{RE_2'Z'\hat{M}} -  U_{\vert R \vert} \otimes (\theta_3^1)_{E_2'Z'\hat{M}} \Vert_1 + \\ (1-p_{\sm})\Vert (\theta_3^0)_{RR'E_2'Z'\hat{M}} -  U_{\vert R \vert} \otimes (\theta_3^0)_{R'E_2'Z'\hat{M}} ) \Vert_1  \leq 4(2^{-k}+\eps).
    \end{multline}Further, we have
    \begin{equation}\label{Eq:file123''}
         \Pr(R =R')_{\theta_3^1}=1.
    \end{equation}
 Let $\Phi$ be a CPTP map from registers $RR'E'_2Z'\hat{M}$ to $M'\hat{M}$ (i.e. $\Phi$ maps state $\theta_3$ to $\theta_4$) in~\cref{fig:splitstate48}. One can note, 
\begin{equation}\label{eq:2classical19845''}
    \Phi(U_{\vert R \vert} \otimes (\theta^0_3)_{R'Z'\hat{M}}) =  \eta^{\mathcal{A}}_{M'}   \otimes\sigma_{\hat{M}}
    \end{equation}follows from using~\cref{prop:mac2} (noting $P_R$ is a perfect encryption scheme) followed by~\cref{fact:data}. Further state $\eta^{\mathcal{A}}_{M'}$ depends only on adversary $\mathcal{A}$. Let $\tilde{ \theta^1_3}$ be the state such that 
\[ \tilde{ (\theta^1_3)}_{RR'Z'\hat{M}} =\tilde{ (\theta^1_3)}_{RR'} \otimes (\theta^1_3)_{Z'\hat{M}} \quad ; \quad   \Pr(R=R')_{\tilde{( \theta^1_3)}}=1 \quad ; \quad \tilde{ (\theta^1_3)}_{R'}=U_{\vert R \vert}. \]
Note, we have\begin{equation}\label{eq:2classical19845r4''}
   \Phi( \tilde{ (\theta^1_3)}_{ RR'Z'\hat{M}}) \approx_{3\eps_P} p \sigma_{M'\hat{M}} + (1-p) U_{M'} \otimes \sigma_{\hat{M}} 
    \end{equation}follows from using~\cref{lem:equal101125'',fact:data}. Note $\sigma_{M'\hat{M}} \equiv \sigma_{M\hat{M}}$. Consider, 
\begin{align}
  (\theta_4)_{M'\hat{M}} 
   &=\Phi( (\theta_3)_{RR'Z'\hat{M}})\nonumber \\
   &= p_{\sm}  \Phi((\theta^1_3)_{RR'Z'\hat{M}}) \nonumber \\
   & \quad \quad + (1-p_{\sm})  \Phi((\theta^0_3)_{RR'Z'\hat{M}})\label{eq:11fromeq20''} \\
     & \approx_{ 4(2^{-k}+ \eps) +  3\eps_P} p_{\sm}  \left( p \sigma_{M'\hat{M}} + (1-p) (U_{{M'}} \otimes \sigma_{\hat{M}}) \right) \nonumber \\
     & \quad \quad + (1-p_{\sm})  \Phi((\theta_3^0)_{RR'Z'\hat{M}})\label{eq:11fromclaimnotequal''}\\
   &=  p_{\sm}\cdot p\cdot \sigma_{M'\hat{M}}+  p_{\sm}  (1-p) (U_{{M'}} \otimes \sigma_{\hat{M}})  \nonumber\\
     & \quad \quad + (1-p_{\sm})  \Phi((\theta^0_3)_{RR'Z'\hat{M}}) \nonumber  \\
       &\approx_{4(2^{-k}+ \eps) }  p_{\sm}\cdot p\cdot \sigma_{M\hat{M}}+  p_{\sm}  (1-p) (U_{{M'}} \otimes \sigma_{\hat{M}})  \nonumber\\
     & \quad \quad + (1-p_{\sm}) \eta^{\mathcal{A}}_{M'} \otimes \sigma_{\hat{M}}\label{eq:111fromclaimnotequal''} \\
     & = p_{\sm}\cdot p\cdot  \sigma_{M\hat{M}}  \nonumber\\
     & \quad\quad + 
 \left(p_{\sm}  (1-p)  U_{M'}  + (1-p_{\sm})  \eta^{\mathcal{A}}_{M'}  \right)\otimes \sigma_{\hat{M}}  . \nonumber
\end{align}
The equality in \eqref{eq:11fromeq20''} follows from \cref{eq:2classical198''}. The approximation in~\eqref{eq:11fromclaimnotequal''} uses \cref{eq:2classical19845r4''} and~\cref{eq:2classical1''}. The approximation in~\eqref{eq:111fromclaimnotequal''} uses \cref{eq:2classical19845''} and~\cref{eq:2classical1''}. Noting that
\[ p_{\sm} \quad ; \quad p_{} \quad ; \quad    \eta^{\mathcal{A}}_{M'}  \]depends only on the adversary $\mathcal{A} = (U,V, \ket{\psi}_{E_1E_2})$ completes the proof. 
\end{proof}
Using \cref{thm:m''}, we get the following corollary:
\begin{corollary}
    There exists constant rate $2$-split quantum  secure non-malleable code for uniform classical messages with efficient classical encoding and decoding procedures with rate $\frac{1}{5+\delta}$ for arbitrarily tiny constant $\delta>0$. 
\end{corollary}

\section{Quantum and quantum secure non-malleable secret sharing\label{Appendixss}}

Let $n,t,m$ be positive integers and $\eps, \eps', \eps_1, \eps_2>0$. 

Secret sharing~\cite{Sha79,Bla79} is a fundamental cryptographic primitive where a dealer encodes a secret into $n$ shares and distributes them among $n$ parties. 
Each secret sharing scheme has an associated \emph{monotone}\footnote{A set $\Gamma\subseteq 2^{[n]}$ is \emph{monotone} if $A\in\Gamma$ and $A\subseteq B$ implies that $B\in\Gamma$.} set $\Gamma\subseteq 2^{[n]}$, usually called an \emph{access structure}, whereby any set of parties $T\in\Gamma$, called \emph{authorized} sets, can reconstruct the secret from their shares, but any \emph{unauthorized} set of parties $T\not\in\Gamma$ gains essentially no information about the secret.
One of the most natural and well-studied types of access structures are \emph{threshold} access structures, where a set of parties $T$ is authorized if and only if $|T|\geq t$ for some threshold $t$.

The notion of \emph{non-malleable} secret sharing, generalizing non-malleable codes, was introduced by Goyal and Kumar~\cite{GK16} and has received significant interest in the past few years in the classical setting (e.g., see~\cite{GK18,ADNOP19,BS19,FV19,BFOSV20,BFV21,GSZ21,CKOS22}).
Non-malleable secret sharing schemes additionally guarantee that an adversary who is allowed to tamper all the shares (according to some restricted tampering model, such as tampering each share independently) cannot make an authorized set of parties reconstruct a different but related secret.


\cite{BGJR23} extended the notion of non-malleable secret sharing to quantum messages. In particular they constructed the $2$-out-of-$2$ non-malleable secret sharing scheme for quantum  messages. These are referred to as 
 \textbf{quantum non-malleable secret sharing schemes.}

Let $\sigma_M$ be an arbitrary state in a message register $M$ and $\sigma_{M\hat{M}}$ be its canonical purification.
Let encoding scheme be given by an encoding Completely Positive Trace-Preserving (CPTP) map 
$\enc : \cL( \cH_M) \to \cL(\cH_{X}\otimes \cH_Y\otimes\cH_{Z})$ and a decoding CPTP map $\dec  : \cL(\cH_{X}\otimes\cH_{Y} \otimes\cH_{Z}) \to  \cL(\cH_{M})$, where $\cL(\cH)$ is the space of all linear operators in the Hilbert space $\cH$.
The most basic property we require of $(\enc,\dec)$ scheme is correctness (which includes preserving entanglement with external systems), i.e.,
\begin{equation*}
\dec(\enc(\sigma_{M\hat{M}}))=\sigma_{M\hat{M}}.
\end{equation*}
Let registers $X, Y, Z$ correspond to three shares of message register $M$.

\begin{definition}[$3$-out-of-$3$ quantum non-malleable secret sharing scheme\label{def:nmssqcodes}]$(\enc, \dec)$ is an $(\eps_1, \eps_2)$-$3$-out-of-$3$ quantum non-malleable secret sharing scheme, if
\begin{itemize}
    \item \textbf{statistical privacy:} for any  quantum message ${\sigma_M} \in \cD(\cH_M)$ (with canonical purification $\sigma_{M\hat{M}}$), it holds that
    \[   (\enc( \sigma))_{\hat{M}XY}  \approx_{\eps_1} \sigma_{\hat{M}} \otimes  \zeta_{XY}  \quad ; \quad  (\enc( \sigma))_{\hat{M}YZ}  \approx_{\eps_1} \sigma_{\hat{M}} \otimes  \zeta_{YZ} ) \quad ; \]
    \[ (\enc( \sigma))_{\hat{M}XZ}  \approx_{\eps_1} \sigma_{\hat{M}} \otimes  \zeta_{XZ} , 
 \]where $\zeta_{XY}, \zeta_{YZ}, \zeta_{XZ}$ are fixed states independent of $\sigma_{M\hat{M}}$.
\item \textbf{non-malleability:} for every $\mathcal{A}=(U,V, W, \ket{\psi}_{E_1E_2E_3})$ in the split-state model such that
\begin{align*}
& U : \cL(\cH_X) \otimes \cL(\cH_{E_1})\to  \cL(\cH_{X'}) \otimes \cL(\cH_{E_1'})
;\\ & V : \cL(\cH_Y) \otimes \cL(\cH_{E_2})\to  \cL(\cH_{Y'}) \otimes \cL(\cH_{E_2'})
;\\ & W : \cL(\cH_Z) \otimes \cL(\cH_{E_3})\to  \cL(\cH_{Z'}) \otimes \cL(\cH_{E_3'})
\end{align*}
and for every quantum message ${\sigma_M} \in \cD(\cH_M)$ (with canonical purification $\sigma_{M\hat{M}}$), it holds that
$$    \dec( ( U \otimes V \otimes W) (\enc( {\sigma} ) \otimes \ketbra{\psi}) ( U^\dagger \otimes V^\dagger \otimes W^\dagger))_{M\hat{M}}  \approx_{\eps_2}  p_{\mathcal{A}}{\sigma_{M\hat{M}}} + (1-p_{\mathcal{A}}) \gamma^{\mathcal{A}}_M \otimes \sigma_{\hat{M}},$$ where $(p_{\mathcal{A}}, \gamma^{\mathcal{A}}_M)$ depend only on adversary $\mathcal{A}$.

\end{itemize}

\end{definition}

We also study the notion of non-malleable secret sharing for classical messages secure against quantum adversaries in which the encoding and decoding procedures are classical. These are referred to as \textbf{quantum secure non-malleable secret sharing schemes.} The definition follows along the lines of~\cref{def:nmssqcodes} but with the following changes: the encoding scheme, decoding scheme and messages are all classical and in $\sigma$, $\hat{M}$ is a copy of $M$ instead of its canonical purification. 

In this work, we study quantum non-malleable secret sharing schemes in the local tampering model, where the adversary is allowed to tamper independently with each share. Additionally, again as in the coding setting, we allow the non-communicating adversaries tampering each share to have access to arbitrary entangled quantum states. We show that our $3$-split, quantum  non-malleable code gives rise to a constant-rate $3$-out-of-$3$ quantum non-malleable secret sharing scheme.
We also construct the first constant-rate 3-out-of-3 quantum secure non-malleable secret sharing scheme  and also the first constant-rate 2-out-of-2 quantum secure non-malleable secret sharing scheme for uniform classical messages.

\begin{corollary}\label{corr:2nmssq}
Let $(\enc,\dec)$ be an $\eps'$-$3$-split quantum non-malleable code from \cref{thm:main} (see \cref{fig:splitstate2}). Then, $(\enc,\dec)$ is also an $(2^{-n^{\Omega(1)}}, \eps')$-$3$-out-of-$3$ quantum non-malleable secret sharing scheme.
\end{corollary}
\begin{proof}

Since $(\enc,\dec)$ is already known to be $\eps'$-non-malleable by \cref{thm:main}, it remains to show that it also satisfies statistical privacy with error $2^{-n^{\Omega(1)}}$ according to \cref{def:nmssqcodes}.

Note from \cref{thm:main}, we have $\eps'= 2^{\vert M \vert} 2(\eps+ 4^{-\vert M\vert})$, where $\eps =2^{-n^{\Omega(1)}}$ is from~\cref{lem:qnmcodesfromnmext}. Note $\enc(.)$ procedure in \cref{fig:splitstate2} first samples an independent $(X,Y)$ and generate register $R= 2\nmext(X,Y)$. It follows from properties of $2\nmext$ (see \cref{lem:qnmcodesfromnmext}) that
\[ RX \approx_\eps U_R \otimes U_X \quad ; \quad RY \approx_\eps U_R \otimes U_Y .\]
Thus for every ${\sigma_{M\hat{M}}}$, using~\cref{fact:notequal}, 
we have
\begin{equation}\label{eqfd}
     (\enc( \sigma))_{\hat{M}ZX}  \approx_{\eps} \sigma_{\hat{M}} \otimes  U_{Z} \otimes U_X \quad ; \quad  (\enc( \sigma))_{\hat{M}ZY}  \approx_{\eps} \sigma_{\hat{M}} \otimes  U_{Z} \otimes U_Y . 
\end{equation}Since $(X,Y)$ are sampled independently, we also have
\begin{equation}\label{eqfdsfvf}
    (\enc( \sigma))_{\hat{M}XY}  =\sigma_{\hat{M}} \otimes  U_{X} \otimes U_Y . 
\end{equation}
Noting \cref{eqfd}, \cref{eqfdsfvf} and   $\eps =2^{-n^{\Omega(1)}}$ completes the proof. 
\end{proof}

\begin{corollary}
Let $(\enc,\dec)$ be an $\eps'$-$3$-split quantum secure non-malleable code from \cref{thm:m1} (see \cref{fig:splitstate4'}).  Then, $(\enc,\dec)$ is also an $(2^{-n^{\Omega(1)}}, \eps')$-$3$-out-of-$3$ quantum secure non-malleable secret sharing scheme.
\end{corollary}
\begin{proof}
The proof follows using similar arguments as in~\cref{corr:2nmssq}. 
\end{proof}

\begin{corollary}
Let $(\enc,\dec)$ be an $\eps'$-$2$-split average case quantum secure non-malleable code 
from \cref{thm:m''} (see \cref{fig:splitstate33}). (i.e, for uniform input classical message). Then, $(\enc,\dec)$ is also an $(2^{-n^{\Omega(1)}}, \eps')$-$2$-out-of-$2$ average case quantum secure non-malleable secret sharing scheme (i.e, for uniform input classical message).
\end{corollary}
\begin{proof}
The proof follows using similar arguments as in~\cref{corr:2nmssq}. 
\end{proof}

\bibliography{References}

\newcommand{\etalchar}[1]{$^{#1}$}
\begin{thebibliography}{AKO{\etalchar{+}}22}

\bibitem[ABJ22]{ABJ22}
Divesh Aggarwal, Naresh~Goud Boddu, and Rahul Jain.
\newblock Quantum secure non-malleable-codes in the split-state model, 2022.
\newblock ArXiv:2202.13354.

\bibitem[ABJO21]{ABJO21}
Divesh Aggarwal, Naresh~Goud Boddu, Rahul Jain, and Maciej Obremski.
\newblock Quantum measurement adversary, 2021.
\newblock ArXiv:2106.02766.

\bibitem[ADL18]{ADL17}
Divesh Aggarwal, Yevgeniy Dodis, and Shachar Lovett.
\newblock Non-malleable codes from additive combinatorics.
\newblock {\em SIAM Journal on Computing}, 47(2):524--546, 2018.
\newblock Preliminary version in STOC 2014.

\bibitem[ADN{\etalchar{+}}19]{ADNOP19}
Divesh Aggarwal, Ivan Damg{\aa}rd, Jesper~Buus Nielsen, Maciej Obremski, Erick
  Purwanto, Jo{\~a}o Ribeiro, and Mark Simkin.
\newblock Stronger leakage-resilient and non-malleable secret sharing schemes
  for general access structures.
\newblock In Alexandra Boldyreva and Daniele Micciancio, editors, {\em Advances
  in Cryptology -- CRYPTO 2019}, pages 510--539, Cham, 2019. Springer
  International Publishing.

\bibitem[AKO{\etalchar{+}}22]{AKOOS22}
Divesh Aggarwal, Bhavana Kanukurthi, Sai Lakshmi~Bhavana Obbattu, Maciej
  Obremski, and Sruthi Sekar.
\newblock Rate one-third non-malleable codes.
\newblock In {\em Proceedings of the 54th Annual ACM SIGACT Symposium on Theory
  of Computing}, STOC 2022, page 1364–1377, New York, NY, USA, 2022.
  Association for Computing Machinery.

\bibitem[AM17]{AM17}
Gorjan Alagic and Christian Majenz.
\newblock Quantum non-malleability and authentication.
\newblock In Jonathan Katz and Hovav Shacham, editors, {\em Advances in
  Cryptology -- CRYPTO 2017}, pages 310--341, Cham, 2017. Springer
  International Publishing.

\bibitem[AO20]{AO20}
Divesh Aggarwal and Maciej Obremski.
\newblock A constant rate non-malleable code in the split-state model.
\newblock In {\em 2020 IEEE 61st Annual Symposium on Foundations of Computer
  Science (FOCS)}, pages 1285--1294, Los Alamitos, CA, USA, nov 2020. IEEE
  Computer Society.

\bibitem[BFO{\etalchar{+}}20]{BFOSV20}
Gianluca Brian, Antonio Faonio, Maciej Obremski, Mark Simkin, and Daniele
  Venturi.
\newblock Non-malleable secret sharing against bounded joint-tampering attacks
  in the plain model.
\newblock In Daniele Micciancio and Thomas Ristenpart, editors, {\em Advances
  in Cryptology -- CRYPTO 2020}, pages 127--155, Cham, 2020. Springer
  International Publishing.

\bibitem[BFV21]{BFV21}
Gianluca Brian, Antonio Faonio, and Daniele Venturi.
\newblock Continuously non-malleable secret sharing: Joint tampering, plain
  model and capacity.
\newblock In Kobbi Nissim and Brent Waters, editors, {\em Theory of
  Cryptography}, pages 333--364, Cham, 2021. Springer International Publishing.

\bibitem[BGJR23]{BGJR23}
Naresh~Goud Boddu, Vipul Goyal, Rahul Jain, and Joao Ribeiro.
\newblock Split-state non-malleable codes for quantum messages, 2023.

\bibitem[BJK21]{BJK21}
Naresh~Goud Boddu, Rahul Jain, and Upendra Kapshikar.
\newblock Quantum secure non-malleable-extractors, 2021.

\bibitem[Bla79]{Bla79}
G.~R. Blakley.
\newblock Safeguarding cryptographic keys.
\newblock In {\em 1979 International Workshop on Managing Requirements
  Knowledge (MARK)}, pages 313--318, 1979.

\bibitem[BS19]{BS19}
Saikrishna Badrinarayanan and Akshayaram Srinivasan.
\newblock Revisiting non-malleable secret sharing.
\newblock In Yuval Ishai and Vincent Rijmen, editors, {\em Advances in
  Cryptology -- EUROCRYPT 2019}, pages 593--622, Cham, 2019. Springer
  International Publishing.

\bibitem[BW16]{BW16}
Anne Broadbent and Evelyn Wainewright.
\newblock Efficient simulation for quantum message authentication.
\newblock In Anderson~C.A. Nascimento and Paulo Barreto, editors, {\em
  Information Theoretic Security}, pages 72--91, Cham, 2016. Springer
  International Publishing.

\bibitem[CG14]{CG14b}
Mahdi Cheraghchi and Venkatesan Guruswami.
\newblock Non-malleable coding against bit-wise and split-state tampering.
\newblock In {\em Proceedings of Theory of Cryptography Conference {(TCC)}},
  pages 440--464, 2014.
\newblock Extended version in {Journal of Cryptology.}

\bibitem[CG16]{CG14a}
Mahdi Cheraghchi and Venkatesan Guruswami.
\newblock Capacity of non-malleable codes.
\newblock {\em IEEE Transactions on Information Theory}, 62(3):1097--1118,
  2016.

\bibitem[CGL20]{CGL15}
Eshan Chattopadhyay, Vipul Goyal, and Xin Li.
\newblock Nonmalleable extractors and codes, with their many tampered
  extensions.
\newblock {\em SIAM Journal on Computing}, 49(5):999--1040, 2020.
\newblock Preliminary version in STOC 2016.

\bibitem[CKOS22]{CKOS22}
Nishanth Chandran, Bhavana Kanukurthi, Sai Lakshmi~Bhavana Obbattu, and Sruthi
  Sekar.
\newblock Short leakage resilient and non-malleable secret sharing schemes.
\newblock In Yevgeniy Dodis and Thomas Shrimpton, editors, {\em Advances in
  Cryptology - {CRYPTO} 2022}, pages 178--207. Springer, 2022.

\bibitem[CLLW16]{CLLW16}
Richard Cleve, Debbie Leung, Li~Liu, and Chunhao Wang.
\newblock Near-linear constructions of exact unitary 2-designs.
\newblock {\em Quantum Info. Comput.}, 16(9–10):721–756, jul 2016.

\bibitem[CLW14]{CLW14}
Kai{-}Min Chung, Xin Li, and Xiaodi Wu.
\newblock Multi-source randomness extractors against quantum side information,
  and their applications.
\newblock {\em CoRR}, abs/1411.2315, 2014.

\bibitem[CZ14]{CZ14}
Eshan Chattopadhyay and David Zuckerman.
\newblock Non-malleable codes against constant split-state tampering.
\newblock In {\em 2014 IEEE 55th Annual Symposium on Foundations of Computer
  Science}, pages 306--315, 2014.

\bibitem[Dat09]{Datta09}
Nilanjana Datta.
\newblock Min- and max- relative entropies and a new entanglement monotone.
\newblock {\em IEEE Transactions on Information Theory}, 55:2816--2826, 2009.

\bibitem[DCEL09]{DCEL09}
Christoph Dankert, Richard Cleve, Joseph Emerson, and Etera Livine.
\newblock Exact and approximate unitary 2-designs and their application to
  fidelity estimation.
\newblock {\em Phys. Rev. A}, 80:012304, Jul 2009.

\bibitem[DKO13]{DKO13}
Stefan Dziembowski, Tomasz Kazana, and Maciej Obremski.
\newblock Non-malleable codes from two-source extractors.
\newblock In Ran Canetti and Juan~A. Garay, editors, {\em Advances in
  Cryptology -- CRYPTO 2013}, pages 239--257, Berlin, Heidelberg, 2013.
  Springer Berlin Heidelberg.

\bibitem[DPVR12]{DPVR09}
Anindya De, Christopher Portmann, Thomas Vidick, and Renato Renner.
\newblock Trevisan's extractor in the presence of quantum side information.
\newblock {\em SIAM Journal on Computing}, 41(4):915--940, 2012.

\bibitem[DPW18]{DPW10}
Stefan Dziembowski, Krzysztof Pietrzak, and Daniel Wichs.
\newblock Non-malleable codes.
\newblock {\em J. ACM}, 65(4), April 2018.
\newblock Preliminary version in ICS 2010.

\bibitem[FV19]{FV19}
Antonio Faonio and Daniele Venturi.
\newblock Non-malleable secret sharing in the computational setting: Adaptive
  tampering, noisy-leakage resilience, and improved rate.
\newblock In Alexandra Boldyreva and Daniele Micciancio, editors, {\em Advances
  in Cryptology -- CRYPTO 2019}, pages 448--479, Cham, 2019. Springer
  International Publishing.

\bibitem[FvdG06]{FvdG06}
C.~A. Fuchs and J.~van~de Graaf.
\newblock Cryptographic distinguishability measures for quantum-mechanical
  states.
\newblock {\em IEEE Trans. Inf. Theor.}, 45(4):1216–1227, September 2006.

\bibitem[GK18a]{GK16}
Vipul Goyal and Ashutosh Kumar.
\newblock Non-malleable secret sharing.
\newblock In {\em Proceedings of the 50th Annual ACM SIGACT Symposium on Theory
  of Computing}, STOC 2018, page 685–698, New York, NY, USA, 2018.
  Association for Computing Machinery.

\bibitem[GK18b]{GK18}
Vipul Goyal and Ashutosh Kumar.
\newblock Non-malleable secret sharing for general access structures.
\newblock In Hovav Shacham and Alexandra Boldyreva, editors, {\em Advances in
  Cryptology -- CRYPTO 2018}, pages 501--530, Cham, 2018. Springer
  International Publishing.

\bibitem[GS95]{GS95}
Arnaldo Garcia and Henning Stichtenoth.
\newblock A tower of artin-schreier extensions of function fields attaining the
  drinfeld-vladut bound.
\newblock {\em Inventiones mathematicae}, 121(1):211--222, 1995.

\bibitem[GSZ21]{GSZ21}
Vipul Goyal, Akshayaram Srinivasan, and Chenzhi Zhu.
\newblock Multi-source non-malleable extractors and applications.
\newblock In Anne Canteaut and Fran{\c{c}}ois-Xavier Standaert, editors, {\em
  Advances in Cryptology -- EUROCRYPT 2021}, pages 468--497, Cham, 2021.
  Springer International Publishing.

\bibitem[JRS02]{JainRS02}
R.~Jain, J.~Radhakrishnan, and P.~Sen.
\newblock Privacy and interaction in quantum communication complexity and a
  theorem about the relative entropy of quantum states.
\newblock In {\em The 43rd Annual IEEE Symposium on Foundations of Computer
  Science, 2002. Proceedings.}, pages 429--438, 2002.

\bibitem[KOS17]{KOS17}
Bhavana Kanukurthi, Sai Lakshmi~Bhavana Obbattu, and Sruthi Sekar.
\newblock Four-state non-malleable codes with explicit constant rate.
\newblock In {\em Theory of Cryptography: 15th International Conference, TCC
  2017, Baltimore, MD, USA, November 12-15, 2017, Proceedings, Part II}, page
  344–375, Berlin, Heidelberg, 2017. Springer-Verlag.

\bibitem[KOS18]{KOS18}
Bhavana Kanukurthi, Sai Lakshmi~Bhavana Obbattu, and Sruthi Sekar.
\newblock Non-malleable randomness encoders and their applications.
\newblock In Jesper~Buus Nielsen and Vincent Rijmen, editors, {\em Advances in
  Cryptology -- EUROCRYPT 2018}, pages 589--617, Cham, 2018. Springer
  International Publishing.

\bibitem[KR09]{KR09}
Bhavana Kanukurthi and Leonid Reyzin.
\newblock Key agreement from close secrets over unsecured channels.
\newblock In {\em Proceedings of the 28th Annual International Conference on
  Advances in Cryptology - EUROCRYPT 2009 - Volume 5479}, page 206–223,
  Berlin, Heidelberg, 2009. Springer-Verlag.

\bibitem[{Li}15]{li15}
Xin {Li}.
\newblock Three-source extractors for polylogarithmic min-entropy.
\newblock In {\em 2015 IEEE 56th Annual Symposium on Foundations of Computer
  Science}, pages 863--882, 2015.

\bibitem[Li17]{Li17}
Xin Li.
\newblock Improved non-malleable extractors, non-malleable codes and
  independent source extractors.
\newblock In {\em Proceedings of the 49th Annual ACM SIGACT Symposium on Theory
  of Computing}, STOC 2017, page 1144–1156, New York, NY, USA, 2017.
  Association for Computing Machinery.

\bibitem[Li19]{Li19}
Xin Li.
\newblock Non-malleable extractors and non-malleable codes: Partially optimal
  constructions.
\newblock In {\em Proceedings of the 34th Computational Complexity Conference},
  CCC '19, Dagstuhl, DEU, 2019. Schloss Dagstuhl--Leibniz-Zentrum fuer
  Informatik.

\bibitem[LL12]{LL12}
Feng-Hao Liu and Anna Lysyanskaya.
\newblock Tamper and leakage resilience in the split-state model.
\newblock In Reihaneh Safavi-Naini and Ran Canetti, editors, {\em Advances in
  Cryptology -- CRYPTO 2012}, pages 517--532, Berlin, Heidelberg, 2012.
  Springer Berlin Heidelberg.

\bibitem[NC00]{NielsenC00}
Michael~A. Nielsen and Isaac~L. Chuang.
\newblock {\em Quantum computation and quantum information}.
\newblock Cambridge University Press, Cambridge, UK, 2000.

\bibitem[Sha79]{Sha79}
Adi Shamir.
\newblock How to share a secret.
\newblock {\em Commun. ACM}, 22(11):612–613, nov 1979.

\bibitem[Uhl76]{uhlmann76}
A.~Uhlmann.
\newblock The "transition probability" in the state space of a *-algebra.
\newblock {\em Rep. Math. Phys.}, 9:273--279, 1976.

\bibitem[Wat11]{WatrousQI}
John Watrous.
\newblock Chapter 5, {Naimark}’s theorem; characterizations of channels.
\newblock Lecture notes on Theory of Quantum Information, 2011.

\bibitem[Wil13]{book_Wilde}
Mark~M Wilde.
\newblock {\em Quantum information theory}.
\newblock Cambridge University Press, 2013.

\end{thebibliography}
\bibliographystyle{alpha}

\newpage 
\appendix

\section{A quantum secure $2$-source non-malleable extractor\label{sec:2nm}}

\subsection{Useful lemmas and facts}\label{sec:usefulstuff1}

We state here few useful lemmas and facts for this section.

\begin{fact}[\cite{BJK21}] \label{claim:traingle_rho_rho_prime}Let $\rho_{ZA}, \rho'_{ZA}$ be states such that $\Delta_B(\rho, \rho') \leq \eps'$. If $d(Z\vert A)_{\rho'} \leq \eps$ then $d(Z\vert A)_{\rho} \leq 2\eps' + \eps$.~\footnote{Claim holds even when $\Delta_B()$ is replaced with $\Delta()$.}
\end{fact}

\begin{fact}[\cite{BJK21}]\label{fact:prefixminentropyfact}
     Let $\rho_{XE} \in \mathcal{D}(\cH_X \otimes \cH_E)$ be a c-q state such that $\vert X \vert =n$ and $\hmin{X}{E}_\rho \geq n-k.$ Let $X_d = \pre(X,d)$ for some integer $k \leq d \leq n$. Then $\hmin{X_d}{E}_\rho \geq d-k.$
\end{fact}
\begin{fact}[\cite{BJK21}]\label{claim:100}
Let $\rho_{XAYB}$ be a pure state. Let $d = |X|$. There exists a pure state $\hat{\rho}_{XAYB}$ such that,
\[ \Delta_B(\hat{\rho}_{XAYB},\rho_{XAYB}) = d(X|YB)_\rho \quad ; \quad  \hmin{Y}{XA}_{\hat{\rho}} = \hmin{Y}{XA}_\rho \quad ; \quad \hat{\rho}_{XYB} = U_{d} \otimes \hat{\rho}_{YB}.\]
\end{fact}

\begin{fact}[Alternating extraction~\cite{BJK21}]\label{lem:2}
Let $\theta_{XASB}$ be a pure state with $(XS)$ classical, $\vert X \vert =n, \vert S \vert =d$ and
\[ \hmin{X}{SB}_\theta \geq k \quad ; \quad \Delta_B( \theta_{X A S} , \theta_{X A} \otimes U_d ) \leq \eps^{\prime}. \]
Let $T \defeq \Ext(X,S)$ where $\Ext$ is a $(k,\eps)$-quantum secure strong $(n,d,m)$-extractor. Then, 
\[\Delta_B( \theta_{T B} , U_m \otimes \theta_{B} ) \leq  2\eps' + \sqrt{\eps}.\]
\end{fact}
\begin{fact}[Min-entropy loss under classical interactive communication~\cite{BJK21}]\label{lem:minentropy}
Let $\rho_{XNM}$ be a pure state where Alice holds registers $(XN)$ and Bob holds register $M$, such that register $X$ is classical and
\[ \hmin{X}{M}_\rho \geq k.\]Let Alice and Bob proceed for $t$-rounds, where in each round Alice generates a classical register $R_i$ and sends it to Bob, followed by Bob generating a classical register $S_i$ and sending it to Alice. Alice applies a (safe on $X$) isometry $V^{i}: \cH_X \otimes \cH_{N_{i-1}} \rightarrow \cH_X \otimes \cH_{N'_{i-1}} \otimes \cH_{R_{i}}$ (in round $i$) to generate~\footnote{The isometries in the communication protocols in later sections act as $V^i: \cH_X \rightarrow \cH_X  \otimes \cH_{R_{i}} \otimes  \cH_{\hat{R}_{i}}$.} $R_{i}$. Let 
 $\theta^i_{XN_iM_i}$ be the state at the end of round-$i$, where Alice holds registers $XN_i$ and Bob holds register $M_i$. Then,
 \[ \hmin{X}{M_t}_{\theta^t} \geq k-\sum_{j=1}^{t} \vert R_j\vert .\]
\end{fact}
\begin{fact}[\cite{BJK21}]\label{claim:minentropydecrease}
Let $\rho_{ABC} \in \mathcal{D}(\cH_A \otimes \cH_B \otimes \cH_C)$ be a state and $M \in \cL(\cH_C)$ such that $M^\dagger M \leq \id_C$. Let $\hat{\rho}_{ABC}= \frac{M \rho_{ABC} M^\dagger}{\tr{M \rho_{ABC} M^\dagger}}$. Then, 
\[ \hmin{A}{B}_{\hat{\rho}} \geq   \hmin{A}{B}_{\rho} - \log \left(\frac{1}{\tr{M \rho_{ABC} M^\dagger}}\right). \]
\end{fact}

\begin{fact}\label{fact:Conjugation} Let $M,A \in \mathcal{L}(\cH)$. 
If $A \geq 0$ then $M^{\dagger} A M \geq 0$. 
\end{fact}

\begin{fact}
\label{measuredmax}
Let $\rho_{AB} \in \mathcal{D}(\cH_A \otimes \cH_B)$ be a state and $M \in \cL(\cH_B)$ such that $M^\dagger M \leq \id_B$. Let $\hat{\rho}_{AB}= \frac{M \rho_{AB} M^\dagger}{\tr{M \rho_{AB} M^\dagger}}$. Then, 
$$\dmax{\hat{\rho}_A}{\rho_A} \leq \log \left(\frac{1}{\tr{M \rho_{AB} M^\dagger}}\right).$$
\end{fact}
\begin{fact}
\label{fact:close}
Let $\rho_{}, \sigma \in \mathcal{D}(\cH_A)$ be two states  and $M \in \cL(\cH_A)$ such that  $M^\dagger M \leq \id_A$. Then, 
\[ \vert \tr{M \rho_{} M^\dagger} - \tr{M \sigma_{} M^\dagger} \vert  \leq  \frac{\Vert \rho -\sigma \Vert_1}{2}.\]
\end{fact}

\begin{fact}[Gentle Measurement Lemma~\cite{book_Wilde}] \label{fact:gentle_measurement}

Let $\rho \in \mathcal{D}(\cH_A)$ be a state and $M \in \cL(\cH_A)$ such that $M^\dagger M \leq \id_A$ and $\tr(M \rho M^\dagger)
\geq 1-\eps$. Let $\hat{\rho}= \frac{M \rho M^\dagger}{\tr{M \rho M^\dagger}}$. Then, 
$\Delta_B\left(\rho , \hat{\rho}\right) \leq {\sqrt{\eps}}$.

\end{fact}

\begin{definition} Let $M=2^m$. The inner-product function, $\IP^{n}_{M}: \mathbb{F}_{M}^{n} \times \mathbb{F}_{M}^n \rightarrow \mathbb{F}_{M}$ is defined as follows:\[\IP^{n}_{M}(x,y)=\sum_{i=1}^{n} x_i y_i,\]
where the operations are over the field $\mathbb{F}_{M}.$
\end{definition}

\begin{fact}[Uhlmann's Theorem~\cite{uhlmann76}]
\label{uhlmann}
Let $\rho_A,\sigma_A\in \mathcal{D}(\cH_A)$. Let $\rho_{AB}\in \mathcal{D}(\cH_{AB})$ be a purification of $\rho_A$ and $\sigma_{AC}\in\mathcal{D}(\cH_{AC})$ be a purification of $\sigma_A$. 
There exists an isometry $V$ (from a subspace of $\cH_C$ to a subspace of $\cH_B$) such that,
\[ \Delta_B\left( \ketbra{\rho}_{AB}, \ketbra{\theta}_{AB}) =  \Delta_B(\rho_A,\sigma_A\right) ,\]
 where $\ket{\theta}_{AB} = (\id_A \otimes V) \ket{\sigma}_{AC}$.
\end{fact}

\begin{fact}[\protect{Stinespring isometry extension~[Theorem $5.1$ in~\cite{WatrousQI}]}]\label{fact:stinespring}
		 Let $\Phi :    \mathcal{L} (\cH_X ) \rightarrow   \mathcal{L}(\cH_Y )$ be a CPTP map. 
		 Then, there exists an isometry $V :  \cH_{X} \rightarrow   \cH_{Y} \otimes \cH_{Z}$ (Stinespring isometry extension of $\Phi$) such that $\Phi(\rho_X)= \tr_{Z}(V \rho_X V^\dagger)$ for every state $\rho_X$.
\end{fact}

\begin{fact}[\cite{CLW14}]
	\label{fact102}  
	 Let $\cE :    \mathcal{L} (\cH_M ) \rightarrow   \mathcal{L}(\cH_{M'} )$ be a CPTP map and let $\sigma_{XM'} =(\id \otimes \cE)(\rho_{XM}) $. Then,  $$ \hmin{X}{M'}_\sigma  \geq \hmin{X}{M}_\rho  .$$
Above is equality if $\cE$ is a map corresponding to an isometry.
\end{fact}

\begin{fact}[Lemma B.3. in~\cite{DPVR09}]
\label{fact2}  
For a c-q state $\rho_{ABC}$ (with $C$ classical),
$$\hmin{A}{BC}_\rho \geq \hmin{A}{B}_\rho - \vert C \vert.$$

\end{fact}

\begin{fact}[]
\label{traceavg}
Let $\rho_{XE},\sigma_{XE}$ be two c-q states. Then,
\begin{itemize}
    \item $ \| \rho_{XE}-\sigma_{XE} \|_1 \geq   \E_{x \leftarrow \rho_X } \| \rho^x_{E}-\sigma^x_{E} \|_1. $
     \item $ \Delta_B( \rho_{XE},\sigma_{XE} ) \geq   \E_{x \leftarrow \rho_X } \Delta_B( \rho^x_{E}, \sigma^x_{E} ). $
\end{itemize}
The above inequalities are equalities iff $\rho_X = \sigma_X$.
\end{fact}

\begin{fact}[\cite{FvdG06}]
\label{fidelty_trace}
Let $\rho,\sigma$ be two states. Then,
\[  1-\F(\rho,\sigma) \leq \Delta(\rho , \sigma) \leq \sqrt{ 1-\F^2(\rho,\sigma)} \quad ; \quad \Delta_B^2(\rho,\sigma) \leq \Delta(\rho , \sigma) \leq  \sqrt{2}\Delta_B(\rho,\sigma).  \]

\end{fact}

\subsection*{Extractors}
Throughout this section we use extractor to mean seeded extractor unless stated otherwise. 
\begin{definition}[Quantum secure extractor]
\label{qseeded}
	An $(n,d,m)$-extractor $\Ext : \{0,1\}^n \times \{0,1\}^d \to \{0,1\}^m$  is said to be $(k,\eps)$-quantum secure if for every state $\rho_{XES}$, such that $\Hmin(X|E)_\rho \geq k$ and $\rho_{XES} = \rho_{XE} \otimes U_d$, we have 
	$$  \| \rho_{\Ext(X,S)E} - U_m \otimes \rho_{E} \|_1 \leq \eps.$$
	In addition, the extractor is called strong if $$  \| \rho_{\Ext(X,S)SE} - U_m \otimes U_d \otimes \rho_{E} \|_1 \leq \eps .$$
	$S$ is referred to as the {\em seed} for the extractor.
	\end{definition}

\begin{fact}[Corollary 5.2 in~\cite{DPVR09}]
    \label{fact:extractor} There exists an explicit $(m+5\log{(\frac{1}{\eps})},\eps)$-quantum secure strong $(n,d,m)$-extractor $\Ext : \{ 0,1\}^n \times  \{ 0,1\}^d \to  \{ 0,1\}^m$ for parameters  $d = \cO( \log^2(\frac{n}{\eps}) \log m )$.
\end{fact}

\begin{definition}[$(k_1,k_2)\mhyphen\qpas$~\cite{BJK21}]\label{qmadvk1k2}
We call a pure state $\sigma_{X\hat{X}NMY\hat{Y}}$, with $(XY)$ classical and $(\hat{X}\hat{Y})$ copy of $(XY)$,  a  $(k_1,k_2)\mhyphen\qpas$ iff 
\[ \hmin{X}{MY\hat{Y}}_\sigma \geq k_1 \quad ; \quad \hmin{Y}{NX\hat{X}}_\sigma \geq k_2.\]
\end{definition}
\begin{definition}[$(k_1,k_2)\mhyphen\nmas$~\cite{BJK21}]\label{def:2source-qnmadversarydef}
     Let $\sigma_{X\hat{X}NMY\hat{Y}}$ be a $(k_1,k_2)\mhyphen\qpas$. Let $U: \cH_X \otimes \cH_N \rightarrow \cH_X \otimes \cH_{X^\prime} \otimes  \cH_{\hat{X}'} \otimes \cH_{N^\prime}$ and $V: \cH_Y \otimes \cH_M \rightarrow \cH_Y \otimes \cH_{Y'} \otimes  \cH_{\hat{Y}'} \otimes \cH_{M'}$ be isometries  such that for $\rho = (U \otimes V)\sigma(U \otimes V)^\dagger,$ we have $(X'Y')$ classical (with copy $\hat{X}'\hat{Y}'$) and, 
       $$\Pr(Y \ne Y^\prime)_\rho =1 \quad or \quad \Pr(X \ne X^\prime)_\rho =1.$$ 
       We call state $\rho$ a $(k_1,k_2)\mhyphen\nmas$.
\end{definition}
\begin{definition}[Quantum secure $2$-source non-malleable extractor~\cite{BJK21}]\label{def:2nme}
		An $(n_1,n_2,m)$-non-malleable extractor $2\nmext : \{0,1\}^{n_1} \times \{0,1\}^{n_2} \to \{0,1\}^m$ is $(k_1,k_2,\eps)$-secure against $\nma$ if for every $(k_1,k_2)\mhyphen\nmas$ $\rho$ (chosen by the adversary $\nma$),
	$$  \Vert \rho_{ 2\nmext(X,Y)2\nmext(X^\prime,Y^\prime) Y  Y^\prime M^\prime} - U_m \otimes \rho_{ 2\nmext(X^\prime ,Y^\prime) Y  Y^\prime M^\prime} \Vert_1 \leq \eps. $$

\end{definition}

\begin{fact}[$\IP$ security against $(k_1,k_2)$-$\qpas$~\cite{ABJO21}] \label{l-qma-needed-fact1} Let $n=\frac{n_1}{m}$ and $k_1+k_2 \geq n_1+m+40+8 \log\left(\frac{1}{\eps}\right)$. Let $\sigma_{X \hat{X} N' Y \hat{Y} M'}$ be a $(k_1,k_2)\mhyphen\qpas$ with $\vert X \vert = \vert Y\vert = n_1$. Let $Z=\IP^n_{2^m}(X,Y)$. Then
\[\Vert \sigma_{ZXN'} - U_{m} \otimes \sigma_{XN'}  \Vert_1 \leq \eps \quad ; \quad \Vert \sigma_{ZYM'} - U_{m} \otimes \sigma_{YM'}  \Vert_1 \leq \eps.\]
\end{fact}

\begin{remark} The proof of~\cref{l-qma-needed-fact1} is provided in the updated version of~\cite{ABJO21}.
\end{remark}

\subsection*{Error correcting codes} 
 
 \begin{definition}\label{def:ecc}
     Let $\Sigma$ be a finite set. A mapping $\ecc: \Sigma^k \to \Sigma^n$ is called an error correcting code with relative distance $\gamma$ if for any $x,y \in \Sigma^k$ such that $x \ne y,$ the Hamming distance between $\ecc(x)$ and $\ecc(y)$ is at least $\gamma n.$ The rate of the code denoted by $\delta$, is defined as $\delta \defeq \frac{k}{n}$. The alphabet size of the code is the number of elements in $\Sigma.$
 \end{definition}
 \begin{fact}[\cite{GS95}]\label{fact:ecc}
      Let $p$ be a prime number and $m$ be an even integer. Set $q=p^m$. For every $\delta \in (0,1)$ and for any large enough integer $n$ there exists an efficiently computable linear error correcting code $\ecc: \F^{\delta n}_q \to \F^{ n}_q $ with rate $\delta$ and relative distance $1-\gamma$ such that $$  \delta + \frac{1}{\sqrt{q}-1}\geq \gamma.$$
 \end{fact}

The following parameters hold throughout this section.
 \subsection*{Parameters}\label{sec:parameters_2nm}
Let $q$ be a prime power and $\delta, \delta_1, \delta_2 >0$ be small enough constants.   Let  $n,a,v,s,b,h$ be positive integers and
 $k, \eps', \gamma, \eps > 0$ such that:

  \[  v= \frac{n}{\eps \log q} \quad ; \quad q= \cO\left(\frac{1}{\eps^2}\right) \quad ; \quad \eps= 2^{- \cO(n^{\delta_1})} \quad ;\]
 
\[a=6k+2 \log q = \mathcal{O}(k) \quad ; \quad \gamma = \cO(\eps) \quad ; \quad  2^{\cO(a)}\sqrt{\eps'} = \eps \quad ;       \]

\[s = \cO\left(\log^2\left(\frac{n}{\eps'}\right)\log n \right)  \quad  ; \quad b = \cO\left( \log^2\left(\frac{n}{\eps'}\right) \log n \right) \quad ; \quad h = 10s \quad ;\quad k = \cO (n^{1/4}) \]
\begin{itemize}\label{sec:extparameters_2nm}
    \item $\IP_1$ be $\IP^{3k/\log v}_{v}$,
    \item $\Ext_1$ be $(b+5 \log(\frac{1}{\eps'}), \eps')$-quantum secure $(\delta n,s,b)$-extractor,
    \item $\Ext_2$ be $(s+5 \log(\frac{1}{\eps'}), \eps')$-quantum secure $(h,b,s)$-extractor,
    \item $\Ext_3$ be $(h+5 \log(\frac{1}{\eps'}), \eps')$-quantum secure $(n,b,h)$-extractor,
    \item $\Ext_4$ be $(n^{\delta_2}+5 \log(\frac{1}{\eps^2}), \eps^2)$-quantum secure $(\delta n,h,n^{\delta_2})$-extractor,
    \item $\IP_2$ be $\IP^{3k^3/h}_{2^h}$,
    \item  $\Ext_6$ be $((\frac{1}{2}-\delta) n+5 \log(\frac{1}{\eps
    ^2}), \eps^2)$-quantum secure $(n,n^{\delta_2},(\frac{1}{2}-\delta) n)$-extractor. 
 \end{itemize}
 
\subsection*{Definition of $2$-source non-malleable extractor}
 Let $\ecc : \F^d_q \to \F^v_q$ be an error correcting code with relative distance $1-\gamma$ and rate $\eps$ (which exists from~\cref{fact:ecc} for our choice of parameters). We identify $R$ as an element from $\{1, \ldots, v \}$. By $\ecc(Y)_{R}$, we mean the $R$-th entry of the codeword $\ecc(Y)$, interpreted as a bit string of length $\log q$.
\begin{algorithm}
\caption{: $2\nmre: \lbrace 0,1 \rbrace ^n\times \lbrace 0,1 \rbrace^{\delta n}   \rightarrow \lbrace 0,1 \rbrace^{(\frac{1}{2}-\delta) n}$}\label{alg:2nmExt}
\begin{algorithmic}
\State{}

\noindent \textbf{ Input:}  $X, Y$\\


\begin{enumerate}
    \item Advice generator: \[X_1=\pre(X,3k) \quad  ; \quad Y_1 = \pre(Y,3k) \quad ;
    \quad R= \IP_1(X_1,Y_1) \quad ; \quad  \]
    \[G=X_1 \circ Y_1 \circ \ecc(X)_{R} \circ \ecc(Y||0^{(1-\delta) n})_{R} \]

    \item $X_2=\pre(X,3k^3) \quad ; \quad Y_2=\pre(Y,3k^3)\quad; \quad Z_0=\IP_2(X_2,Y_2)$
    \item Correlation breaker with advice:\quad $S=2\advcb(Y,X,Z_0, G)$
    \item $L=\Ext_6(X,S)$
\end{enumerate}

 \noindent \textbf{ Output:} $L$ 
\end{algorithmic}
\end{algorithm}


\begin{algorithm}
\caption{: $2\advcb: \lbrace 0,1 \rbrace^{\delta n} \times \lbrace 0,1 \rbrace^n \times \lbrace 0,1\rbrace^h \times \lbrace 0,1 \rbrace^a \rightarrow \lbrace 0,1 \rbrace^{n^{\delta_2}}$}\label{alg:2AdvCB}
\begin{algorithmic}
\State{}

\noindent \textbf{ Input: }$Y, X, Z_0, G$

\begin{enumerate}
    \item For $i=1,2,\ldots,a:$ 
    
     \hspace{1cm}Flip flop: \quad $Z_i=2\ff(Y,X,Z_{i-1},G_i)$ 
     \item $S=\Ext_4(Y,Z_a)$

\end{enumerate}

 \noindent \textbf{ Output:} $S$ 
\end{algorithmic}
\end{algorithm} 


\begin{algorithm}
\caption{: $2\ff : \lbrace 0,1 \rbrace^{\delta n} \times \lbrace 0,1 \rbrace^n \times \lbrace 0,1 \rbrace^h \times \lbrace 0,1 \rbrace^a \rightarrow \lbrace 0,1\rbrace^h$}\label{alg:2FF}
\begin{algorithmic}
\State{}

\noindent\textbf{Input:}  $Y, X,  Z,  G$ 
\begin{enumerate}
     
    \item $Z_s=$Prefix$(Z,s)$, $A= \Ext_1 (Y,Z_s),\ C= \Ext_2(Z,A),\ B= \Ext_{1}(Y,C)$
    \item If $G=0$ then $\overline{Z}= \Ext_3(X,A)$ and if $G=1$ then $\overline{Z}= \Ext_3(X,B)$
     \item $\overline{Z}_s=$Prefix$(\overline{Z},s)$, $\overline{A}= \Ext_1 ({Y},\overline{Z}_s),\ \overline{C}= \Ext_2(\overline{Z},\overline{A}),\ \overline{B}= \Ext_{1}({Y},\overline{C})$
    \item If $G=0$, then $O=\Ext_3(X,\overline{B})$ and if $G=1$, then $O=\Ext_3(X,\overline{A})$
\end{enumerate} \noindent \textbf{Output:} $O$
\end{algorithmic}
\end{algorithm}

\subsection*{Result}

In Protocol~\ref{prot:2nmExt_full}, Alice and Bob generate new classical registers using safe isometries on old classical registers. At any stage of Protocol~\ref{prot:2nmExt_full}, we use $N$ to represent all the registers held by Alice other than the specified registers at that point. Similarly $M$ represents all the registers held by Bob other than the specified registers. At any stage of the protocol, we use $\tilde{A}, \tilde{B}$ to represent all the registers held by Alice and Bob respectively. We use the same convention for communication protocols in later sections as well.

The following theorem shows that the function $2\nmre$ as defined in \cref{alg:2nmExt} is $(n-k,\delta n-k,\cO(\eps))$-secure against $\nma$ by noting that  $L= 2\nmre(X,Y)$ and $L'=2\nmre(X^\prime,Y^\prime)$.
\begin{theorem}[Security of $2\nmre$] \label{thm:2nmext}

Let $\rho_{X\hat{X}X^\prime \hat{X}' NYY'\hat{Y}\hat{Y}'M}$ be a $(n-k,\delta n-k)\mhyphen\nmas$ with $\vert X \vert =n, \vert Y \vert = \delta n$. Let Protocol~\ref{prot:2nmExt_full} start with $\rho$. Let $\Lambda$ be the state at the end of the protocol. Then,
\[ \Vert \rho_{ L L^{\prime} Y  Y^\prime M} - U_{(\frac{1}{2}-\delta) n} \otimes \rho_{ L^\prime Y  Y^\prime M} \Vert_1 \leq d(L \vert \tilde{B})_\Lambda \leq \cO( \eps).\]
	
\end{theorem}
\begin{proof} The first inequality follows from \cref{fact:data}. Most of our arguments here are similar to the case of quantum secure 2-source non-malleable extractors \cite{BJK21} with modified parameters; so we provide the proof considering modified parameters following the proof outline of \cite{BJK21}.

First note that using \cref{lem:minentropy}, throughout Protocol~\ref{prot:2nmExt_full}, $X$ and $Y$ have enough conditional min-entropy left for necessary extractions since the total communication from Alice to Bob is at most (from our choice of parameters) 
$$6k+2 \log q+ 6ah + h+ \left(\frac{1}{2}-\delta\right)n \leq   \left(\frac{1}{2}-0.99\delta\right)n$$
 and the total communication from Bob to Alice is at most $$6k +6k^3 + 2a + 6ab + 2n^{\delta_2} \leq  \cO(n^{\frac{3}{4}}).$$Thus, at any state $\varrho$ in Protocol~\ref{prot:2nmExt_full},
$\hmin{X}{\tilde{B}}_\varrho \geq n-k-\left(\frac{1}{2}-0.99\delta\right)n \geq \left(\frac{1}{2}+0.99\delta\right)n -k\geq \left(\frac{1}{2}+0.9\delta\right)n$. Similarly, $\hmin{Y}{\tilde{A}}_\varrho \geq \delta n -k-\cO(n^\frac{3}{4})\geq 0.9\delta n$.

We start with a state $\rho_{X X' N Y Y' M}$ such that $\hmin{X}{\tilde{B}}_{\rho} \geq n-k$ and $\hmin{Y}{\tilde{A}}_{\rho} \geq \delta n-k$. From \cref{fact:prefixminentropyfact}, we have,
\[\hmin{X_1}{\tilde{B}}_{\rho} \geq 3k- k =2k \quad ; \quad \hmin{Y_1}{\tilde{A}}_{\rho} \geq 3k-k=2k .\]

Now from \cref{l-qma-needed-fact1} with the below assignment of registers (and noting registers $(XX',YY' )$ are included in $(\tilde{A},\tilde{B})$ respectively),
\[(Z,X,Y, \sigma) \leftarrow (R,X_1, Y_1, \rho) \quad ; \quad (k_1,k_2,m,n_1,\eps ) \leftarrow (2k,2k, \log\left(\frac{n}{\eps \log q}\right), 3k, \eps^2) \]we have,
\[ \Delta( \rho_{RYY'}, U_R \otimes \rho_{YY'} ) \leq \cO(\eps^2) \quad ; \quad \Delta( \rho_{RXX'}, U_R \otimes \rho_{XX'} ) \leq  \cO(\eps^2). \]
Using \cref{fidelty_trace}, we get 
\[ \Delta_B( \rho_{RYY'}, U_R \otimes \rho_{YY'} ) \leq  \cO(\eps) \quad ; \quad \Delta_B( \rho_{RXX'}, U_R \otimes \rho_{XX'} ) \leq  \cO(\eps). \] \\

Let $\kappa$ be the state just before Bob sends $Y_2$. Note that till then, communication from Alice to Bob and Bob to Alice is at most $7k$ each. Hence, by \cref{lem:minentropy}, $\hmin{X}{\tilde{A}}_{\kappa} \geq n-8k$ and $\hmin{Y}{\tilde{B}}_\kappa \geq \delta n-8k$; which implies (from \cref{fact:prefixminentropyfact}), $\hmin{X_2}{\tilde{A}}_{\kappa} \geq 3k^3-8k \geq 2k^3$ and $\hmin{Y_2}{\tilde{B}}_\kappa \geq 2k^3$ respectively using \cref{fact:prefixminentropyfact}. Let $\eta$ be the state just after Alice generates $Z_0$.
Using similar argument as before involving \cref{l-qma-needed-fact1}, we can get $d(Z_0 \vert \tilde{B})_{\eta} \leq  2^{-\cO(a)}\eps$.

Let $\Phi$ be the state obtained in Protocol~\ref{prot:2nmExt_full}, just before Protocol~\ref{prot:block2} starts as a subroutine. From \cref{lemma:block12t}, we have, 
\begin{equation}\label{lemma:block1:point3}\Pr(G=G')_\Phi = \cO(\eps) .
\end{equation}
Since, $d(Z_0 \vert \tilde{B})_{\eta} \leq 2^{-\cO(a)}\eps$, using \cref{fact102}, we also have, 
\begin{equation}\label{lemma:block1:point4}
d(Z_0|\tilde{B})_ \Phi \leq 2^{-\cO(a)}\eps.
\end{equation}
 Let $\mathcal{S}_1 \defeq \{(\alpha,\alpha') ~:~ \alpha = \alpha'\}.$ 
 Note $\Pr((\alpha,\alpha') \in\mathcal{S}_1)_{ {\Phi}} \leq \cO(\eps).$ Let 
$${\Phi}^{(\alpha,\alpha')}= {\Phi}| ((G,G')=(\alpha,\alpha')),$$ and $\mathcal{S}_2 \defeq \{(\alpha,\alpha') ~ :~ \Pr((G,G')=(\alpha,\alpha'))_{{\Phi}} \leq \frac{\eps}{2^{2 \vert G \vert}}\}$.
Note $\Pr((\alpha,\alpha') \in\mathcal{S}_2)_{{\Phi}} \leq \eps$. For every $(\alpha,\alpha') \notin \mathcal{S}_2$, we have (using \cref{measuredmax} and noting that, in ${\Phi}$, a copy of $(G,G')$ is part of $\tilde{B}$),
 \begin{equation}\label{eq:phi2}
{\Phi}^{(\alpha,\alpha')}_{Y\tilde{A}} \leq   \frac{{\Phi}_{Y\tilde{A}}  }{\Pr((G,G')=(\alpha,\alpha'))_{{\Phi}}}  \leq  2^{2 \vert G \vert+ \log( \frac{1}{\eps})} \cdot {\Phi}_{Y\tilde{A}}  \enspace.
 \end{equation}
 \cref{eq:phi2} imply that for every $(\alpha,\alpha') \notin \mathcal{S}_2$, we have $$\hmin{Y}{\tilde{A}}_{{\Phi}^{(\alpha,\alpha')}} \geq \delta n - \cO(k) -2a - \log \left( \frac{1}{\eps} \right) \geq 0.99\delta n .$$Similarly, one can also note that for every $(\alpha,\alpha') \notin \mathcal{S}_2$,
  $$\hmin{X}{\tilde{B}}_{{\Phi}^{(\alpha,\alpha')}} \geq   n - \cO(k)  -2a - \log \left( \frac{1}{\eps} \right) \geq (1/2 + 0.99\delta)n . $$
  Further for every $(\alpha,\alpha') \notin \mathcal{S}_2$, using \cref{traceavg1} and  \cref{lemma:block1:point4}, we get 
  \[ d(Z_0|\tilde{B})_{\Phi^{(\alpha, \alpha')}}  \leq 2^{2 \vert G \vert+ \log( \frac{1}{\eps})} d(Z_0|\tilde{B})_{\Phi^{}} \leq \cO(\eps). \]

  Let $\mathcal{S} \defeq  \mathcal{S}_1 \cup \mathcal{S}_2$. From the union bound, $\Pr((\alpha,\alpha') \in\mathcal{S})_{{\Phi}} \leq \cO(\eps).$ 
 
 Let ${\Gamma}^{(\alpha,\alpha')}, {\Gamma}$ be the joint states at the end of the Protocol~\ref{prot:block2} (for $i=a$) when starting with the states  ${\Phi}^{(\alpha,\alpha')} ,{\Phi}$ respectively. From \cref{thm:advcb}, we have for every $(\alpha,\alpha') \notin \mathcal{S}$, 
\begin{equation}\label{mainthm:e3}
d(Z|\tilde{B})_{{\Gamma}^{(\alpha,\alpha')}}  \leq  \cO(\eps).
 \end{equation}

 Consider (register $Z$ is held by Alice and $\tilde{B} = GG'M$ in the state ${\Gamma}$),  
  \begin{align*}
      d(Z|\tilde{B})_{{\Gamma}} = \Delta_B ( {\Gamma}_{ZGG'M} , U_{h} \otimes {\Gamma}_{GG'M}) &= \E_{(\alpha,\alpha') \leftarrow (G,G')} \Delta_B ( {\Gamma}^{(\alpha,\alpha')}_{Z\tilde{B}} , U_{h} \otimes {\Gamma}^{(\alpha,\alpha')}_{\tilde{B}})  \\
      & \leq \Pr((\alpha,\alpha') \notin \mathcal{S})_{{\Gamma}} \cdot  \cO(\eps) + \Pr((\alpha,\alpha') \in \mathcal{S})_{{\Gamma}}  \\
      & \leq \cO(\eps),
  \end{align*}
  where the first equality follows from \cref{traceavg}, the first inequality follows from \cref{mainthm:e3} and the last inequality follows since $\Pr((\alpha,\alpha') \in\mathcal{S})_{{\Gamma}}=\Pr((\alpha,\alpha') \in\mathcal{S})_{{\Phi}} \leq \cO(\eps)$. Using arguments as before (involving \cref{lem:2}), we have 
  $d(L \vert \tilde{B})_\Lambda \leq \cO(\eps)$. 

\end{proof}

We can verify the following claim regarding the state $\Phi$ (the state obtained in Protocol~\ref{prot:2nmExt_full}, just before Protocol~\ref{prot:block2} starts as a subroutine) using arguments from~\cite{BJK21}.
\begin{claim}\label{lemma:block12t}
 $\Pr(G= G^\prime)_{\Phi} = \mathcal{O(\eps)}$. 
\end{claim}

\begin{claim}[Correlation breaker with advice]\label{thm:advcb}
Let Alice and Bob proceed as in Protocol~\ref{prot:block2} with the starting state as ${\Phi}^{(\alpha,\alpha')}$, where $(\alpha, \alpha') \notin \mathcal{S}$. Let ${\Gamma}^{(\alpha,\alpha')}$ be the joint state at the end of the Protocol~\ref{prot:block2} (at $i=a$). Then,
 \[d(Z|\tilde{B})_{{\Gamma}^{(\alpha,\alpha')}} = \cO(\eps). \] 
\end{claim}
\begin{proof} We have 
    \[d(Z_0|\tilde{B})_{{\Phi}^{(\alpha,\alpha')}} \leq \cO(\eps) \quad ; \quad \hmin{Y}{\tilde{A}}_{{\Phi}^{(\alpha,\alpha')}} \geq  \delta n - \cO(k)  \quad ; \quad  \hmin{X}{\tilde{B}}_{{\Phi}^{(\alpha,\alpha')}} \geq   n - \cO(k).\]
    The total communication from Alice to Bob in Protocol~\ref{prot:block2} is at most $6ah \leq \cO(n^{\frac{3}{4}})$. From ~\cref{lem:minentropy}, throughout Protocol~\ref{prot:block2}, we have  $\hmin{X}{\tilde{B}} \geq n-\cO(n^{\frac{3}{4}})$. From repeated applications of \cref{lem:alpha01} after using \cref{claim:100}  (by letting $\rho$ in \cref{claim:100} as $\Phi^{(\alpha, \alpha')}$ here) we have,
 \[d(Z|\tilde{B})_{{\Gamma}^{(\alpha,\alpha')}} \leq \cO(\eps)+2^{\cO(a)}\sqrt{\eps'} = \cO(\eps).\] 
   \end{proof}
	\begin{claim}[Flip flop]\label{lem:alpha01}
Let $\cP$ be any of the Protocols~\ref{prot:Var_GEN(0,1)analysis},~\ref{prot:Var_GEN(1,0)},~\ref{prot:Var_GEN(0,0)NotDiffBefore},~\ref{prot:Var_GEN(0,0)DiffBefore},~\ref{prot:Var_GEN(1,1)NotDiffBefore} or~\ref{prot:Var_GEN(1,1)DiffBefore} and $i \in [a]$. Let $\alpha$ be the initial joint state in~$\cP$ such that 
	$d(Z|\tilde{B})_\alpha \leq \eta.$  Let $\theta$ be the final joint state at the end of~$\cP$. Then, 
 $d(O|\tilde{B})_\theta \leq \cO(\eta + \sqrt{\eps'})$.
	\end{claim}
	\begin{proof}
We prove the claim when $\cP$ is Protocol~\ref{prot:Var_GEN(0,1)analysis} and $i=1$ and the proof for other cases follows analogously.  
From \cref{fact:data}, 
$$d(Z_s|\tilde{B})_\alpha \leq d(Z|\tilde{B})_\alpha \leq \eta.$$

Let $\gamma$ be the joint state just after Bob generates register $A$. From \cref{lem:2}, we have 
$$d(A|\tilde{A})_\gamma \leq 2\eta + \sqrt{\eps'}.$$
Let $\zeta$ be the joint state after Alice sends register $Z'$ to Bob and Bob generates registers $(A'C'B')$. From \cref{fact:data}, we have $$ d(A|\tilde{A})_\zeta \leq d(A|\tilde{A})_\gamma\leq 2\eta + \sqrt{\eps'}.$$
Let $\beta$ be the joint state just after Alice generates register $\overline{Z}$. From \cref{lem:2},  
$$d(\overline{Z}|\tilde{B})_\beta \leq 4\eta + 3\sqrt{\eps'}.$$ Let $\hat{\beta}$ be the state obtained from \cref{claim:100} (by letting $\rho$ in \cref{claim:100} as $\beta$ here) such that
\begin{equation}
 \label{eq:beta}   
 \hmin{Y}{\tilde{A}}_{\hat{\beta}}= \hmin{Y}{\tilde{A}}_{\beta} \geq 0.9\delta n \quad ; \quad \hat{\beta}_{\overline{Z}\tilde{B}}= U_h \otimes \hat{\beta}_{\tilde{B}} \quad ; \quad  \Delta_B(\hat{\beta} ,\beta) \leq 4\eta + 3\sqrt{\eps'}. 
 \end{equation}
 Let $\theta', \hat{\theta}'$ be the joint states just after Alice generates register $\overline{C}$, proceeding from the states  $\beta, \hat{\beta}$ respectively. Since communication between Alice and Bob after Alice generates register $\overline{Z}$ and before generating $\overline{C}$ is $2s+2b$, from arguments as before involving \cref{lem:2} and \cref{lem:minentropy},$$d(\overline{C}|\tilde{B})_{\hat{\theta}'} \leq \cO(\eta +\sqrt{\eps'}).$$
 From \cref{eq:beta},  
 $$d(\overline{C}|\tilde{B})_{\theta'} \leq d(\overline{C}|\tilde{B})_{\hat{\theta'}} + \Delta_B(\hat{\beta} ,\beta) = \cO(\eta +\sqrt{\eps'}).$$
 Proceeding till the last round and using similar arguments involving \cref{lem:2}, \cref{lem:minentropy} and \cref{claim:100}, we get the desired.
\end{proof}
\section{Communication Protocols \label{sec:communication_protocols}}
\begin{changemargin}{0cm}{0cm}
\section*{Protocols for $2\nmre$}
\pagenumbering{gobble}
\begin{Protocol}[H] 
	\begin{center}
	\scalebox{0.88}{
		\begin{tabular}{l c c r}
			Alice:  $(X, \hat{X}, X^\prime, \hat{X^\prime}, N)$ &  & ~~~~~~~~~~~~Bob: $(Y, \hat{Y}, Y^\prime, \hat{Y^\prime}, M)$ & $\quad$ Analysis \\
			\hline\\
			$X_1=\pre(X,3k)$&  &  &$\hmin{X_1}{\tilde{B}} \geq 2k$\\ 
			&&$Y_1=\pre(Y,3k)$ & $\hmin{Y_1}{\tilde{A}} \geq 2k$ \\ \\
		&$X_1 \longrightarrow X_1$ &$R=\mathsf{IP}_1(X_1,Y_1)$& $d(R \vert XX') \leq \cO(\eps)$\\
		& &&\\
		$R=\mathsf{IP}_1(X_1,Y_1)$ & $Y_1 \longleftarrow Y_1$& & $d(R \vert YY') \leq \cO(\eps)$ \\
		$V= \ecc(X)_R$ & & $W= \ecc(Y||0^{(1-\delta)n})_R$ & \\
		 & $V \longrightarrow V$ & $G= X_1 \circ Y_1 \circ V \circ W$& \\ \\
		 $X_1^\prime=\pre(X^\prime,3k)$&  \\ 
		 &&$Y_1^\prime=\pre(Y^\prime,3k)$ & \\
		&$X_1^\prime \longrightarrow X_1^\prime$ &$R^\prime=\mathsf{IP}_1(X_1^\prime,Y_1^\prime)$& \\
		& &$W^\prime= \ecc(Y^\prime||0^{(1-\delta)n})_{R^\prime}$ &\\ \\
		$R^\prime=\mathsf{IP}_1(X_1^\prime,Y_1^\prime)$& $Y_1^\prime \longleftarrow Y_1^\prime$ & & \\ \\
	
		 	$V^\prime= \ecc(X^\prime)_{R^\prime}$& $V^\prime \longrightarrow V^\prime$ &$G^\prime= X_1^\prime \circ Y_1^\prime \circ V^\prime \circ W^\prime$ & \\ \\
		  & &  $Y_2=\pre(Y,3k^3)$& $\hmin{Y_2}{\tilde{A}} \geq 2k^3$\\ 
		 $X_2=\pre(X,3k^3)$ & & & $\hmin{X_2}{\tilde{B}} \geq 2k^3$
		 \\
		 $Z_0=\IP_2(X_2, Y_2) $ &$Y_2 \longleftarrow Y_2$ & & $d(Z_0 \vert \tilde{B}) \leq 2^{-\cO(a)}(\eps)$ \\
		 & &  $Y_2^\prime=\pre(Y^\prime,3k^3)$
		& 
		 \\$X_2^\prime=\pre(X^\prime,3k^3)$
		 \\
		 $Z_0^\prime=\IP_2(X^\prime_2, Y^\prime_2) $ &$Y_2^\prime \longleftarrow Y_2^\prime$ & & \\ \\
		 &$\left(G,G^\prime\right) \longleftarrow \left(G,G^\prime\right) $&& $d(Z_0 \vert \tilde{B}) \leq 2^{-\cO(a)}(\eps)$ \\ \\
		
		 	Alice: $(X,X',Z_0,Z_0',G,G',N)$ &  & Bob: $(Y,Y',G,G',M)$ & \\
		  \hline
	 &&&\\
		 	& Protocol~\ref{prot:block2}~$( X,X',Z_0,Z_0',$&\\
        	&$G,G',N,Y,Y',G,G',M)$ & \\

        		&& \\
				Alice: $(X,X',Z,N)$ &  & Bob: $(Y,Y',Z',M)$ & \\ 
			\hline
			&& \\
			$L' =\Ext_6(X',S')$&$S' \longleftarrow S'$ &  $S' =\Ext_4(Y',Z')$ &$d( Z \vert \tilde{B} ) \leq \cO(\eps)$ \\
		    && \\
		    & $ Z \longrightarrow Z$ &  $S =\Ext_4(Y,Z)$ &$d( S \vert \tilde{A} ) \leq \cO(\eps)$ \\
		    && \\
		     & $ L' \longrightarrow L'$ & &$d( S \vert \tilde{A} ) \leq \cO(\eps)$  \\
		    && \\
		    $L =\Ext_6(X,S)$	&$S \longleftarrow S$ &  &$d( L \vert \tilde{B} ) \leq \cO(\eps)$ \\
			\hline \\
		  Alice($L, N$)                & & Bob($L^\prime, Y, Y^{\prime}, M$) &                                   
		\end{tabular}
		{\small {\caption{ \label{prot:2nmExt_full}
					 $( X, \hat{X},X^\prime,\hat{X}', N, Y,\hat{Y}, Y^\prime,\hat{Y}',M )$.}}
		}}
	\end{center}
	
\end{Protocol}
\end{changemargin}

\begin{Protocol}[htb]

\vspace{0.1in}
For $i=1,2,\ldots,a:$

 \begin{itemize}
        \item Protocol~\ref{prot:Var_GEN(0,1)analysis}~$(X,X',Z,Z',G,G',N, Y,Y',G,G',M)$ for $(\alpha_i,\alpha_i')=(0,1)$.
        \item Protocol~\ref{prot:Var_GEN(1,0)}~$(X,X',Z,Z',G,G',N, Y,Y',G,G',M)$ for $(\alpha_i,\alpha_i')=(1,0)$.
        \item Protocol~\ref{prot:Var_GEN(0,0)NotDiffBefore}~$(X,X',Z,Z',G,G',N, Y,Y',G,G',M)$ for $(\alpha_i,\alpha_i')=(0,0)$ and $\alpha_j=\alpha_j'$ for $j<i$.
        \item Protocol~\ref{prot:Var_GEN(0,0)DiffBefore}~$(X,X',Z,Z',G,G',N, Y,Y',G,G',M)$ for $(\alpha_i,\alpha_i')=(0,0)$ and $\alpha_j \ne \alpha_j'$ for some $j<i$.
        \item Protocol~\ref{prot:Var_GEN(1,1)NotDiffBefore}~$(X,X',Z,Z',G,G',N, Y,Y',G,G',M)$ for $(\alpha_i,\alpha_i')=(1,1)$ and $\alpha_j=\alpha_j'$ for $j<i$.
        \item Protocol~\ref{prot:Var_GEN(1,1)DiffBefore}~$(X,X',Z,Z',G,G',N, Y,Y',G,G',M)$ for $(\alpha_i,\alpha_i')=(1,1)$ and $\alpha_j \ne \alpha_j'$ for some $j<i$.
        
    \end{itemize}

	\quad \quad	$(Z,Z') = (O,O')$.

	\vspace{0.25cm}
	
		{\small {\caption{\label{prot:block2}
					 $( X,X',Z,Z',G,G',N, Y,Y',G,G',M)$.}}
		}
\end{Protocol}

\newpage

\begin{Protocol}
	\begin{center}
		\begin{tabular}{l l r r}
			Alice:  $(X,X^\prime,Z, Z^\prime, G,G',N)$ &  & ~~~~~~~~~~~~Bob: $(Y, Y^\prime,G,G',M)$ & $\quad$ Analysis \\
			\hline\\
			$ Z_s =$ Prefix$(Z,s)$  & &  &$d( Z_s \vert \tilde{B} ) \leq {\eta}$\\ \\
			 & $Z_s \longrightarrow Z_s$ & $A= \Ext_1(Y,Z_s)$ &$ d(A \vert \tilde{A} )  \leq { \cO(\eta + \sqrt{\eps'})}$ \\ \\

			& $Z^\prime \longrightarrow Z^\prime$ & 			 $A^\prime = \Ext_1(Y^\prime, Z_s^\prime)$ & \\ & &  $C^\prime= \Ext_2(Z^\prime, A^\prime)$ & \\ & &  $B^\prime= \Ext_1(Y^\prime, C^\prime)$ &$d (A \vert \tilde{A} ) \leq  {\cO(\eta + \sqrt{\eps'})}$\\ \\
			
			$\overline{Z}= \Ext_3(X,A)$ &$A \longleftarrow A$ && \\ 
				$\overline{Z}_s= \pre(\overline{Z},s)$& & &$d (\overline{Z}_s  \vert  \tilde{B}  ) \leq {\cO(\eta + \sqrt{\eps'})}$ \\ \\
	        
	        
			$\overline{Z}^\prime =\Ext_3(X^\prime, B^\prime)$ & $ B^\prime \longleftarrow B^\prime$ & &$d (\overline{Z}_s  \vert  \tilde{B}  ) \leq {\cO(\eta + \sqrt{\eps'})}$ \\ 
				$\overline{Z}^\prime_s= \pre(\overline{Z}',s)$& & & \\ \\
			& $\overline{Z}_s \longrightarrow \overline{Z}_s $ & $\generateAbar $&$d (\overline{A}  \vert  \tilde{A}  ) \leq {\cO(\eta + \sqrt{\eps'})}$ \\ \\
			
			&$\overline{Z}^\prime_s \longrightarrow \overline{Z}^\prime_s $&$\generateAAbar$& $d (\overline{A}  \vert  \tilde{A}  ) \leq {\cO(\eta + \sqrt{\eps'})}$ \\ \\
			
			$\generateTbar$ & $ \overline{A} \longleftarrow \overline{A} $ &  &$d (\overline{C} \vert \tilde{B} ) \leq {\cO(\eta + \sqrt{\eps'})}$ \\ \\
			 
			 	$O^\prime = \Ext_3(X^\prime, \overline{A}^\prime)$&$\overline{A}^\prime \longleftarrow \overline{A}^\prime$&  &$d(\overline{C} \vert  \tilde{B} ) \leq {\cO(\eta + \sqrt{\eps'})}$\\ 
			 &&\\

			 & $\overline{C} \longrightarrow \overline{C}$  & $\overline{B}= \Ext_1(Y,\overline{C})$ & $d (\overline{B} \vert \tilde{A} ) \leq {\cO(\eta + \sqrt{\eps'})}$\\ \\
			
	 & $O^\prime \longrightarrow O^\prime$ & &$d(\overline{B} \vert \tilde{A} ) \leq {\cO(\eta + \sqrt{\eps'})}$\\ \\



			$O= \Ext_3(X,\overline{B})$ & $\overline{B} \longleftarrow \overline{B}$ & &$d(O \vert  \tilde{B} ) \leq {\cO(\eta + \sqrt{\eps'})}$\\

		\end{tabular}
		{\small {\caption{\label{prot:Var_GEN(0,1)analysis}
					 $( X,X',Z,Z',G,G',N, Y,Y',G,G',M)$.}}
		}
	\end{center}
\end{Protocol}

\begin{Protocol}
	\begin{center}
		\begin{tabular}{l l r r}
			Alice:  $(X,X^\prime,Z, Z^\prime,G,G',N)$ &  & ~~~~~~~~~~~~Bob: $(Y, Y^\prime,G,G',M)$ & $\quad$ Analysis \\
			
			\hline\\
			$ Z_s =$ Prefix$(Z,s)$  &  & &$d(Z_s \vert \tilde{B} ) \leq {\eta}$\\ \\
			& $Z_s \longrightarrow Z_s$ & $A= \Ext_1(Y,Z_s)$ &$d( A \vert \tilde{A} ) \leq {\cO(\eta + \sqrt{\eps'})}$\\ \\
			$Z'_s =\pre(Z',s)$ & $Z_s^\prime \longrightarrow Z_s^\prime$ &  $\generateAA$&$d( A \vert \tilde{A} ) \leq {\cO(\eta + \sqrt{\eps'})}$	\\ \\ 	
			
				$\generateT$ & $A \longleftarrow A$ & & $d( C \vert \tilde{B} ) \leq {\cO(\eta + \sqrt{\eps'})}$\\ \\
			 $\overline{Z}^\prime= \Ext_3(X^\prime,A^\prime)$ & \sendAArl & &$d( C \vert \tilde{B} ) \leq {\cO(\eta + \sqrt{\eps'})}$\\ \\ 
			& \sendTlr & $\generateB$ &$d( B \vert \tilde{A} ) \leq {\cO(\eta + \sqrt{\eps'})}$\\ \\
		
		& \sendYYbarlr &$\generateAAbar$  &\\ 
		 & & $\generateTTbar$ &\\
		 & & $\generateBBbar$ &$d( B \vert \tilde{A} ) \leq {\cO(\eta + \sqrt{\eps'})}$\\ \\
			
			 $\overline{Z}= \Ext_3(X,B)$ & \sendBrl &  &$d( \overline{Z}_s \vert \tilde{B} ) \leq {\cO(\eta + \sqrt{\eps'})}$\\ \\

			$O^\prime= \Ext_3(X^\prime, \overline{B}^\prime)$	& \sendBBbarrl &  &$d( \overline{Z}_s \vert \tilde{B} ) \leq {\cO(\eta + \sqrt{\eps'})}$\\ \\
			
			&  \sendYSbarlr& $\generateAbar$ &$d( \overline{A} \vert \tilde{A} ) \leq {\cO(\eta + \sqrt{\eps'})}$\\ \\
			
		& $O^\prime \longrightarrow O^\prime$& &$d( \overline{A} \vert \tilde{A} ) \leq {\cO(\eta + \sqrt{\eps'})}$ \\ \\

		$O= \Ext_3(X,\overline{A})$& \sendAbarrl & &$d( O \vert \tilde{B} ) \leq {\cO(\eta + \sqrt{\eps'})}$\\ \\
		\end{tabular}
		{\small {\caption{\label{prot:Var_GEN(1,0)}
				$( X,X',Z,Z',G,G',N, Y,Y',G,G',M)$.}}
		}
	\end{center}
\end{Protocol}

\begin{Protocol}
	\begin{center}
		\begin{tabular}{l l r r}
			Alice:  $(X,X^\prime,Z, Z^\prime,G,G',N)$ &  & ~~~~~~~~~~~~Bob: $(Y, Y^\prime,G,G',M)$ & $\quad$ Analysis \\
			
			\hline\\
			$ Z_s =$ Prefix$(Z,s)$  & & &$d(Z_s \vert \tilde{B} ) \leq {\eta}$\\ \\
			
				 & $Z_s \longrightarrow Z_s$ &$A= \Ext_1(Y,Z_s)$ &$d(A \vert \tilde{A} ) \leq {\eta}$\\ \\
				 
				 $ Z'_s =$ Prefix$(Z',s)$ 	& $Z_s^\prime \longrightarrow Z_s^\prime$ & $\generateAA$ 	&$d( A \vert \tilde{A} ) \leq {\cO(\eta + \sqrt{\eps'})}$\\ \\ 
				 	
			 $\overline{Z}= \Ext_3(X,A)$ & $A \longleftarrow A$ &  &\\ 
			 
			 	 $ \overline{Z}_s = \pre(\overline{Z},s)$  &  &  &$d( \overline{Z_s} \vert \tilde{B} ) \leq {\cO(\eta + \sqrt{\eps'})}$\\ \\
				 
			$\overline{Z}^\prime= \Ext_3(X^\prime,A^\prime)$&  \sendAArl & &$d( \overline{Z_s} \vert \tilde{B} ) \leq {\cO(\eta + \sqrt{\eps'})}$\\ \\
		 
			 & \sendYSbarlr &  $\generateAbar$ &$d( \overline{A} \vert \tilde{A} ) \leq {\cO(\eta + \sqrt{\eps'})}$\\ \\
			
			$\overline{Z}^\prime_s= \pre(\overline{Z}^\prime,s)$  & $\overline{Z}^\prime_s \longrightarrow \overline{Z}^\prime_s$  & $\generateAAbar$ &$d( \overline{A} \vert \tilde{A} ) \leq {\cO(\eta + \sqrt{\eps'})}$ \\ \\

		$\generateTbar$& \sendAbarrl & 	&$d( \overline{C} \vert \tilde{B} ) \leq {\cO(\eta + \sqrt{\eps'})}$\\ \\
		
		$\generateTTbar$	& \sendAAbarrl&  &$d( \overline{C} \vert \tilde{B} ) \leq {\cO(\eta + \sqrt{\eps'})}$\\ \\
			
			
		 & \sendTbarlr & $\generateBbar$&$d( \overline{B} \vert \tilde{A} ) \leq {\cO(\eta + \sqrt{\eps'})}$\\ \\

		& \sendTTbarlr & $\generateBBbar$&$d( \overline{B} \vert \tilde{A} ) \leq {\cO(\eta + \sqrt{\eps'})}$\\ \\
		
		$O= \Ext_3(X,\overline{B})$ & \sendBbarrl&  &$d( O \vert \tilde{B} ) \leq {\cO(\eta + \sqrt{\eps'})}$\\ \\

	  $O^\prime= \Ext_3(X^\prime, \overline{B}^\prime)$& \sendBBbarrl	& &$d( O \vert \tilde{B} ) \leq {\cO(\eta + \sqrt{\eps'})}$\\

		\end{tabular}
		{\small {\caption{\label{prot:Var_GEN(0,0)NotDiffBefore}
					 $( X,X',Z,Z',G,G',N, Y,Y',G,G',M)$.}}
		}
	\end{center}
\end{Protocol}
\begin{Protocol}
	\begin{center}
		\begin{tabular}{l l r r}
			Alice:  $(X,X^\prime,Z,G,G',N)$ &  & ~~~~~~~~~~~~Bob: $(Y, Y^\prime,Z',G,G',M)$ & $\quad$ Analysis \\
			
			\hline\\
			
			 $ Z_s =$ Prefix$(Z,s)$ 	&   & &$d( Z_s \vert \tilde{B} ) \leq {\eta}$ \\ \\
			 
		     	$\overline{Z}^\prime= \Ext_3(X^\prime, A^\prime)$ & \sendAArl   & $\generateAA$ &$d( Z_s \vert \tilde{B} ) \leq {\eta}$ \\ \\
		     	
			 & $Z_s \longrightarrow Z_s$ &  $A= \Ext_1(Y,Z_s)$&$d( A \vert \tilde{A} ) \leq {\cO(\eta + \sqrt{\eps'})}$\\ \\
			 
			  
			  	  & $ \overline{Z}^\prime \longrightarrow \overline{Z}^\prime $ &$\generateAAbar$ &\\ 
			  	  && $\generateTTbar$& \\ 
			  	  &&$\generateBBbar$&$d( A \vert \tilde{A} ) \leq {\cO(\eta + \sqrt{\eps'})}$ \\ \\
			 
			$\overline{Z}= \Ext_3(X,A)$ & $A \longleftarrow A$ & &\\ 
			 $\overline{Z}_s = \pre(\overline{Z},s)$&&&$d( \overline{Z}_s \vert \tilde{B} ) \leq {\cO(\eta + \sqrt{\eps'})}$ \\ \\
			
			 	  $O^\prime= \Ext_3(X^\prime, \overline{B}^\prime)$	& \sendBBbarrl &  &$d( \overline{Z}_s \vert \tilde{B} ) \leq {\cO(\eta + \sqrt{\eps'})}$ \\ \\
			

			  &  $ \overline{Z}_s \longrightarrow \overline{Z}_s$ & $\generateAbar$ &$d( \overline{A} \vert \tilde{A} ) \leq {\cO(\eta + \sqrt{\eps'})}$\\ \\


		 $\generateTbar$& \sendAbarrl & &$d(\overline{C}  \vert \tilde{B} ) \leq {\cO(\eta + \sqrt{\eps'})}$\\ \\

		 & \sendTbarlr & $\generateBbar$  &$d( \overline{B} \vert \tilde{A} ) \leq {\cO(\eta + \sqrt{\eps'})}$\\ \\

	     & 	$O^\prime \longrightarrow O^\prime$ & &$d( \overline{B} \vert \tilde{A} ) \leq {\cO(\eta + \sqrt{\eps'})}$ \\ \\
	     
		$O= \Ext_3(X,\overline{B})$& \sendBbarrl &&$d( O \vert \tilde{B} ) \leq {\cO(\eta + \sqrt{\eps'})}$ \\ \\

		\end{tabular}
		{\small {\caption{\label{prot:Var_GEN(0,0)DiffBefore}
					 $(X,X',Z,Z',G,G',N, Y,Y',G,G',M)$.}}
		}
	\end{center}
\end{Protocol}

\begin{Protocol}
	\begin{center}
		\begin{tabular}{l l r r}
			Alice:  $(X,X^\prime,Z, Z^\prime,G,G',N)$ &  & ~~~~~~~~~~~~Bob: $(Y, Y^\prime,G,G',M)$ & $\quad$ Analysis \\
			
			\hline\\
			$ Z_s =$ Prefix$(Z,s)$  &  & &$d( Z_s \vert \tilde{B} ) \leq {\eta}$\\ \\
			 & $Z_s \longrightarrow Z_s$ &$A= \Ext_1(Y,Z_s)$  &$d( A \vert \tilde{A} ) \leq {\cO(\eta + \sqrt{\eps'})}$\\ \\
			
			$ Z'_s =$ Prefix$(Z',s)$ 	& $Z^\prime_s \longrightarrow Z^\prime_s$ & $\generateAA$ &$d( A \vert \tilde{A} ) \leq {\cO(\eta + \sqrt{\eps'})}$	\\ \\ 
				
			$\generateT$& $A \longleftarrow A$ &  &$d( C \vert \tilde{B} ) \leq {\cO(\eta + \sqrt{\eps'})}$ \\ \\
			 $\generateTT$& \sendAArl& &$d( C \vert \tilde{B} ) \leq {\cO(\eta + \sqrt{\eps'})}$\\ \\

		    & \sendTlr & $\generateB$ &$d( B \vert \tilde{A} ) \leq {\cO(\eta + \sqrt{\eps'})}$\\ \\ 
		  
		   &\sendTTlr &$\generateBB$ &$d( B \vert \tilde{A} ) \leq {\cO(\eta + \sqrt{\eps'})}$\\ \\
		   
		   $\overline{Z}= \Ext_3(X,B)$&\sendBrl & & \\ 
		   
		   \generateYSbar&&&$d( \overline{Z}_s \vert \tilde{B} ) \leq {\cO(\eta + \sqrt{\eps'})}$ \\ \\
		   
		   $\overline{Z}^\prime= \Ext_3(X^\prime,B^\prime)$ &\sendBBrl & &$d( \overline{Z}_s \vert \tilde{B} ) \leq {\cO(\eta + \sqrt{\eps'})}$\\ \\

		   & \sendYSbarlr &  $\generateAbar$ &$d( \overline{A} \vert \tilde{A} ) \leq {\cO(\eta + \sqrt{\eps'})}$\\ \\
		   
		 	\generateYSSbar&  \sendYSSbarlr& $\generateAAbar$&$d( \overline{A} \vert \tilde{A} ) \leq {\cO(\eta + \sqrt{\eps'})}$\\ \\

	$O=\Ext_3(X,\overline{A}) $& \sendAbarrl& &$d( O \vert \tilde{B} ) \leq {\cO(\eta + \sqrt{\eps'})}$ \\ \\

			$O^\prime= \Ext_3(X^\prime, \overline{A}^\prime)$ & \sendAAbarrl& &$d( O \vert \tilde{B} ) \leq {\cO(\eta + \sqrt{\eps'})}$\\ \\

		\end{tabular}
		{\small {\caption{\label{prot:Var_GEN(1,1)NotDiffBefore}
					 $( X,X',Z,Z',G,G',N, Y,Y',G,G',M)$. }}
		}
	\end{center}
\end{Protocol}
\begin{Protocol}
	\begin{center}
		\begin{tabular}{l r r r}
			Alice:  $(X,X^\prime,Z,G,G',N)$ &  & ~~~~~~~~~~~~Bob: $(Y, Y^\prime,Z',G,G',M)$ & $\quad$ Analysis \\
			
			\hline\\
			$Z_s = \pre(Z,s)$ & &  &$d( Z_s \vert \tilde{B} ) \leq {\eta}$\\ \\ 
			&&	$Z^\prime_s = \pre(Z^\prime,s)$& \\
				&  & 			 $A^\prime = \Ext_1(Y^\prime, Z_s^\prime)$ & \\ & &  $C^\prime= \Ext_2(Z^\prime, A^\prime)$ & \\ & &  $B^\prime= \Ext_1(Y^\prime, C^\prime)$ &\\ \\
				
			
		  $\overline{Z}^\prime= \Ext_3(X^\prime, B^\prime)$& \sendBBrl& &$d( Z_s \vert \tilde{B} ) \leq {\eta}$\\ \\
		 
		  & \sendYSlr& $\generateA$ &$d( A \vert \tilde{A} ) \leq {\cO(\eta + \sqrt{\eps'})}$\\ \\
		  
		  $\overline{Z}_s^\prime=\pre( \overline{Z}^\prime,s) $	& \sendYSSbarlr&$\generateAAbar$ &$d( A \vert \tilde{A} ) \leq {\cO(\eta + \sqrt{\eps'})}$ \\ \\
		  	
		 $\generateT$& \sendArl&  &$d(C \vert \tilde{B} ) \leq {\cO(\eta + \sqrt{\eps'})}$ \\ \\
	
	  $O^\prime= \Ext_3(X^\prime, \overline{A}^\prime)$& \sendAAbarrl& &$d( C \vert \tilde{B} ) \leq {\cO(\eta + \sqrt{\eps'})}$ \\ \\ 
	 
		 & \sendTlr& $\generateB$ &$d( B \vert \tilde{A} ) \leq {\cO(\eta + \sqrt{\eps'})}$ \\ \\
		 
			 	 & $O^\prime \longrightarrow O^\prime$ & &$d( B \vert \tilde{A} ) \leq {\cO(\eta + \sqrt{\eps'})}$\\ \\
		
		 $\overline{Z}= \Ext_3(X,B)$& \sendBrl& & \\ 
		$\overline{Z}_s=$Prefix$(\overline{Z},s)$&&&$d( \overline{Z}_s \vert \tilde{B} ) \leq {\cO(\eta + \sqrt{\eps'})}$ \\ \\
		  &\sendYSbarlr &$\generateAbar$ &$d( \overline{A} \vert \tilde{A} ) \leq {\cO(\eta + \sqrt{\eps'})}$\\ \\
		  
		 $O= \Ext_3(X, \overline{A})$& \sendAbarrl &  &$d(O \vert \tilde{B} ) \leq {\cO(\eta + \sqrt{\eps'})}$\\ \\

		\end{tabular}
		{\small {\caption{\label{prot:Var_GEN(1,1)DiffBefore}
					 $( X,X',Z,Z',G,G',N, Y,Y',G,G',M)$.}}
		}
	\end{center}
\end{Protocol}

\end{document}